\documentclass[12pt]{article}
\usepackage{setspace}
\usepackage{amsfonts}
\usepackage{comment}
\usepackage{tikz}

\usepackage[margin=1in]{geometry}
\usepackage{enumerate}
\usepackage{amssymb,amsmath,amsthm}
\usepackage{eqnarray}
\usepackage{tabularx}
\usepackage{graphicx}
\usepackage{subcaption} 
\usepackage[authoryear]{natbib}

\usepackage{booktabs}
\usepackage{titlesec}
\usepackage{cite}
\usepackage{rotating}
\usepackage{soul}
\usepackage{fancyhdr}
\usepackage{multirow}
\usepackage{threeparttable}
\usepackage{lscape}
\usepackage{cancel}
\usepackage{xr}
\usepackage{titletoc}
\usepackage{float}

\usepackage[hypertexnames=false]{hyperref}
\hypersetup{
    colorlinks,
    linkcolor={blue!30!black},
    citecolor={blue!30!black},
    urlcolor={blue!30!black},
    breaklinks=true
}
\usepackage{xurl}

\definecolor{coolblack}{rgb}{0.0, 0.18, 0.39}
\definecolor{ao(english)}{rgb}{0.0, 0.5, 0.0}

\titlespacing*{\section}{0pt}{0.5ex plus 1ex minus .2ex}{0.3ex plus .2ex}
\titlespacing*{\subsection}{0pt}{0.25ex plus 1ex minus .2ex}{0.3ex plus .2ex}
\titlespacing*{\subsubsection}
{0pt}{0.25ex plus 1ex minus .2ex}{0.3ex plus .2ex}
\usepackage{array}
\newcolumntype{L}[1]{>{\raggedright\let\newline\\\arraybackslash\hspace{0pt}}m{#1}}
\newcolumntype{C}[1]{>{\centering\let\newline\\\arraybackslash\hspace{0pt}}m{#1}}
\newcolumntype{R}[1]{>{\raggedleft\let\newline\\\arraybackslash\hspace{0pt}}m{#1}}

\usepackage{fancyhdr} 
\usepackage{sectsty}
\sectionfont{\large}
\subsectionfont{\normalsize}
\subsubsectionfont{\normalsize}

\usepackage{xcolor}
\definecolor{forest}{rgb}{0.0, 0.5, 0.0}

\usepackage[normalem]{ulem}
\graphicspath{{./Figures/}}

\newtheorem{theorem}{Theorem}[section]

\newtheorem{proposition}{Proposition}[section]
\newtheorem{corollary}{Corollary}[section]

\newtheorem{lemma}{Lemma}[section]
\newtheorem{example}{Example}[section]
\newtheorem{remark}{Remark}[section]

\newcommand{\indep}{\perp \!\!\! \perp}
\newcommand{\ai}{\alpha_i}
\newcommand{\ais}{\alpha_{is}}
\newcommand{\lt}{\lambda_t}

\newcommand{\eit}{\varepsilon_{it}}
\newcommand{\eist}{\varepsilon_{ist}}
\newcommand{\eisone}{\varepsilon_{is1}}
\newcommand{\eistwo}{\varepsilon_{is2}}
\newcommand{\eitminusone}{\varepsilon_{i(t-1)}}

\newcommand{\eizero}{\varepsilon_{i0}}
\newcommand{\eione}{\varepsilon_{i1}}
\newcommand{\eitwo}{\varepsilon_{i2}}

\newcommand{\eiT}{\varepsilon_{iT}}

\newcommand{\oi}{\omega_i}

\newcommand{\did}{\operatorname{DiD}}
\newcommand{\att}{\operatorname{ATT}}
\newcommand{\eipre}{\varepsilon_i^{1}}
\newcommand{\etaipre}{\eta_i^{1}}
\newcommand{\biasdidpost}{\Delta_{\operatorname{post}}}
\newcommand{\biasmtgpost}{\Delta_{\operatorname{post}}^{\operatorname{dev}}}
\newcommand{\biasselpost}{\Delta_{\operatorname{post}}^{\operatorname{sel}}}
\newcommand{\biasselpre}
{\Delta_{\operatorname{pre}}^{\operatorname{sel}}}

\newcommand{\biasdidposthat}{\widehat{\Delta}_{\operatorname{post}}}

\newcommand{\biasselprehat}{\widehat{\Delta}_{\operatorname{pre}}^{\operatorname{sel}}}

  \newtheorem{innercustomass}{Assumption}
\newenvironment{customass}[1]
  {\renewcommand\theinnercustomass{#1}\innercustomass}
  {\endinnercustomass}

\makeatletter
\def\blfootnote{\xdef\@thefnmark{}\@footnotetext}
\makeatother

\title{\vspace{-.5cm}
Selection and parallel trends\thanks{We are grateful to Isaiah Andrews, Manuel Arellano, Dmitry Arkhangelsky, St\'ephane Bonhomme, Irene Botosaru, Christoph Breunig, Federico Bugni, Brantly Callaway, Ivan Canay, Cl\'ement de Chaisemartin, Gordon Dahl, Aureo de Paula, Graham Elliott, Joachim Freyberger, Bulat Gafarov, Bryan Graham, Lena Janys, Stefan Hoderlein, Christian Hansen, Keisuke Hirano, Peter Hull, Guido Imbens, D\'esir\'e K\'edagni, Pat Kline, Nikolay Kudrin, Matt Masten, Eric Mbakop, David McKenzie, Eduardo Morales, Mikkel Plagborg-M{\o}ller, Vitor Possebom, Niklas Potrafke, Demian Pouzo, Jonathan Roth, Aleksey Tetenov, Andres Santos, Yuya Sasaki, Vira Semenova, Xiaoxia Shi, Valentin Verdier, Chris Walters, and many seminar and conference participants for comments. We used Grammarly for language editing assistance, Github Co-Pilot for coding assistance, and GPT (OpenAI), Claude (Anthropic), and \url{refine.ink} for proofreading assistance. The usual disclaimer applies.}}
\author{Dalia Ghanem\thanks{Department of Agricultural \& Resource Economics, University of California, Davis. One Shields Ave, Davis CA 95616; \href{dghanem@ucdavis.edu}{dghanem@ucdavis.edu}} \quad Pedro H. C. Sant'Anna\thanks{Department of Economics, Emory University, 1602 Fishburne Dr, Atlanta, GA 30322; \href{pedro.santanna@emory.edu}{pedro.santanna@emory.edu}} \quad Kaspar W\"uthrich\thanks{Department of Economics, University of Michigan, 238 Lorch Hall, 611 Tappan Ave., Ann Arbor, MI 48109-1220, CESifo; \href{kasparwu@umich.edu}{kasparwu@umich.edu}}}

\date{First draft on arXiv: March 17, 2022. This draft: \today \vspace{-0.5cm}}

\begin{document}

\maketitle
\thispagestyle{empty}

\begin{abstract}

We study the role of selection into treatment in difference-in-differences (DiD) designs.
We derive necessary and sufficient conditions for parallel trends assumptions under general classes of selection mechanisms. These conditions characterize the empirical content of parallel trends and clarify the trade-offs between assumptions about selection into treatment and restrictions on the time series properties of the potential outcomes required for DiD methods. We use the necessary  and sufficient conditions to provide a selection-based decomposition of the bias of DiD and provide easy-to-implement strategies for benchmarking its components. We also provide templates for justifying DiD in applications with and without covariates. Reanalyses of the causal effect of NSW training programs and the effect of the Medicaid expansion demonstrate the usefulness of our selection-based approach to benchmarking the bias of DiD.

\medskip
\noindent \textbf{Keywords:} causal inference, conditional parallel trends, covariates, difference-in-differences, selection mechanism, time-invariant and time-varying unobservables, treatment effects

\medskip

\noindent \textbf{JEL Codes:} C21, C23

\end{abstract}

\newpage

\setcounter{page}{1}

\begin{quote}
     \footnotesize \textit{\dots while the new papers [in the DiD literature] clarify very well the statistical assumptions needed for estimation, effective use of these methods also requires being able to understand what the threats to these assumptions are in different contexts, and to make a plausible rhetorical argument as to why we should think the assumptions hold.}\\ --- David McKenzie, \emph{World Bank Development Impact Blog} \citep{mckenzie2022a}
\end{quote}

\section{Introduction}

This paper provides a new perspective on difference-in-differences (DiD) identification through the lens of how units select into treatment. Parallel trends, the identifying assumption underlying DiD, requires the change in the expected (untreated) potential outcome over time to be the same in the treatment and control group. It is thus inherently a joint restriction on selection into treatment and the time series properties of the untreated potential outcome. Whether and to what extent there is a trade-off between these two types of restrictions is not well-understood, however. In particular, there is no general framework that characterizes the implications of parallel trends under selection-based restrictions. Such a framework is crucial for researchers to exploit contextual and economic information about how units select into treatment, which is available in many applications. It would enable researchers  --- in the words of David McKenzie \citep{mckenzie2022a} --- ``to understand what the threats to these [parallel trends] assumptions are in different contexts, and to make a plausible rhetorical argument as to why we should think the assumptions hold.'' Indeed, we provide researchers with a selection-based framework for exploiting contextual and economic information to assess the plausibility of parallel trends beyond existing statistical and visual plausibility checks in DiD analyses. 

Our goal is to provide a nonparametric selection-based characterization of the parallel trends assumption. To achieve this goal, we start by providing a general necessary and sufficient condition for parallel trends to hold for all selection mechanisms in a class defined by \emph{the unobservables that determine selection}, which we denote by $\oi$.\footnote{From a theoretical perspective, this necessary and sufficient condition is the condition for nonparametric identification using DiD for a given class of selection mechanisms.} In applications, $\oi$ often corresponds to the units' \emph{information set} when making their selection decisions. We show that parallel trends holds if and only if the trend in the absence of the treatment is mean independent of the unobservable determinants of selection $\oi$. This general necessary and sufficient condition implies necessary and sufficient conditions for many empirically relevant classes of selection mechanisms, indexed by what units select on, $\oi$.

The focus on classes of selection mechanisms defined by $\oi$ is motivated by the reduced-form nature of DiD methods and the myriad of empirical contexts analyzed using these methods. While contextual and economic information about selection may allow researchers to posit what determines selection into treatment, it is generally difficult to specify a particular functional form for the selection mechanism. For instance, if we consider settings where selection into treatment is at the aggregate level, such as at the county or state level as in the Medicaid expansion \citep[e.g.,][]{miller2021medicaid,baker2026difference}, specifying a selection mechanism might be intractable or prohibitively complex.\footnote{This echoes the point made by \citet[][p.404]{abadie2021using} about the difficulty of specifying selection mechanisms in synthetic control analyses of comparative case studies.}  The advantage of our necessary and sufficient conditions is that researchers can use them to assess and justify parallel trends without needing to defend a particular choice of selection mechanism.\footnote{If researchers have sufficient contextual and economic information to fully specify the selection mechanism, then we recommend they rely on identification strategies that directly exploit such information rather than DiD.}

Before considering various restricted classes of selection mechanisms, we take a step back and ask: What are the necessary and sufficient conditions for parallel trends if we do not impose any restrictions on selection? The first corollary of our general necessary and sufficient condition answers this question with a ``no-free-lunch'' result. When researchers are not willing to restrict the selection mechanism at all, $\oi$ can include time-invariant and time-varying determinants of the untreated potential outcomes in the pre- and post-treatment period, $Y_{i1}(0)$ and $Y_{i2}(0)$.
We show that parallel trends holds for all selection mechanisms in this class if and only if the untreated potential outcome is constant across time up to deterministic mean shifts, $Y_{i2}(0)-Y_{i1}(0)=E[Y_{i2}(0)-Y_{i1}(0)]$. This result shows that if one is not willing to restrict selection into treatment, then one needs to essentially rule out time-varying unobservables.

This ``no-free-lunch'' result motivates considering restricted classes of selection mechanisms, such as selection on treatment effects (``Roy-style selection''), selection on pre-treatment unobservables (``imperfect foresight''), selection on lagged outcomes, and selection on time-invariant unobservables (``selection on fixed effects''). The necessary and sufficient conditions for these classes characterize the trade-offs between restrictions on selection and restrictions on the time series properties of the untreated potential outcome and its unobservable determinants. The more restrictions researchers are willing to impose on how units select into treatment, the weaker the time series restrictions necessary and sufficient for parallel trends to hold. Conversely, in settings where the available economic and contextual knowledge is not sufficient to justify strong restrictions on selection, researchers need to impose restrictive assumptions on the time series properties of the potential outcomes to justify parallel trends. 

Consider, for example, a setting with Roy-style selection where the units select into treatment if their treatment effects exceed the costs. If the units know their treatment effects and costs, then the necessary and sufficient condition for parallel trends requires the change in the potential outcomes over time, $Y_{i2}(0)-Y_{i1}(0)$, to be mean independent of those determinants of selection. Alternatively, suppose that the units select into treatment based on pre-treatment unobservables. In this case, the necessary and sufficient condition for parallel trends is a martingale-type condition on $Y_{it}(0)-E[Y_{it}(0)]$. Finally, in the context of selection on time-invariant unobservables, parallel trends is equivalent to a time homogeneity condition on the expectation of $Y_{it}(0)-E[Y_{it}(0)]$ conditional on the time-invariant unobservables. The advantage of our general necessary and sufficient condition is that researchers can use it to derive equivalent conditions for parallel trends tailored to the units' information sets in their application.

The necessary and sufficient conditions we provide can serve as theory-based templates allowing researchers to assess and justify parallel trends in empirical applications based on contextual and economic information about how units select into treatment. To illustrate, we specialize our general results to the standard two-way fixed effects model for $Y_{it}(0)$. The resulting conditions explicitly allow for selection on time-invariant and time-varying unobservables, thus formalizing what ``quasi-random'' assignment means in the context of DiD analyses.

An appealing feature of our selection-based approach is that our necessary and sufficient conditions for parallel trends can help researchers better understand the bias of DiD. For any given information set $\oi$, we decompose the bias of DiD into two components. The first component results from violations of selection on unobservables not contained in $\oi$, whereas the second component captures the bias due to violations of the necessary and sufficient condition for parallel trends when selection is based on $\oi$. Indeed, the information set $\oi$ provides an anchor for the decomposition without imposing any assumptions on selection or the time series properties of the potential outcomes. The proposed selection-based bias decomposition allows us to better understand the bias of DiD and the extent to which pre-treatment information is useful for assessing its magnitude.

For practitioners, we propose simple approaches to benchmark the  components in the bias decomposition and assess the robustness of DiD in empirical applications. Since selection on time-invariant and pre-treatment unobservables is a concern in many applications, including the two empirical applications in this paper, we anchor the decomposition on this case. Specifically, we apply the bias decomposition with $\oi$ containing all time-invariant and pre-treatment unobservables. For this choice of $\oi$, the first component of the bias captures the bias resulting from selection on post-treatment unobservables. The second component captures the bias due to deviations from the martingale condition, which is necessary and sufficient for parallel trends under imperfect foresight. Under a linear relaxation of the martingale condition, the second bias component is equal to the product of the martingale deviation and the pre-treatment difference between the treatment and control group. Based on this decomposition, we provide simple approaches for benchmarking and signing both bias components.

We illustrate the usefulness of the selection-based approach for benchmarking the bias components of DiD in two empirical applications. First, we consider the estimation of the causal effect of the NSW training programs. This application is well-suited for our purposes because there is an experimental benchmark allowing us to estimate the bias of DiD. Without covariates, the bias of DiD relative to the experimental benchmark is large and significant. The proposed benchmarking strategy yields the same sign of the bias as the experimental estimate. The decomposition further demonstrates that the bias of DiD is very sensitive to violations of the martingale property because there is a large pre-treatment difference between the treatment and control group. Incorporating covariates into the analysis reduces the estimated bias relative to the experimental benchmark and also renders DiD more robust by reducing the pre-treatment difference between the treatment and control group. Second, we revisit the DiD analysis of the causal effect of the Medicaid expansion. The proposed benchmarking strategy suggests that the sign of the bias depends on the degree of selection on post-treatment unobservables. As in the NSW application, DiD without covariates is sensitive to violations of the martingale condition. Controlling for covariates almost fully removes pre-treatment differences between the treated and control states, rendering DiD robust to violations of the martingale condition.

\subsection{Related literature}\label{sec:literature} This paper contributes to several branches of the literature on causal inference using panel data. First, we contribute to the classical literature on canonical DiD setups by providing a novel selection-based perspective on the parallel trends assumptions underlying this literature.\footnote{See, e.g., \citet{Ashenfelter1978}, \citet{ashenfelter1985using}, \citet{heckman1985alternative}, \citet{Card1990}, \citet{Card1994}, \citet{Meyer1995}, and \citet{angrist1999empirical} for early developments in the DiD literature, and \citet[][Section 2]{Lechner2010a} for a historical perspective.}  

Our second contribution is to the more recent literature on DiD methods. See, e.g., \citet{de_chaisemartin_two-way-survey_2021} and \citet{roth_et_al_DiD_survey} for surveys. Within this strand of the literature, our paper is most closely related to \citet{roth2021when}, \citet{Arkhangelsky2021_DRTWFE}, and \citet{Arkhangelsky2022_DRId}, though our focus greatly differs from theirs. \citet{roth2021when} discuss necessary and sufficient conditions under which parallel trends holds for all (monotonic) transformations of the untreated potential outcome. We, on the other hand, take the outcome model (and thus the specific transformation) as given and study the connection between parallel trends and selection into treatment. \citet{Arkhangelsky2021_DRTWFE} and \citet{Arkhangelsky2022_DRId}  propose doubly robust estimation methods that leverage restrictions on outcome models and/or selection models with unconfoundedness-type restrictions; see also \citet{Athey2021}. Our results complement theirs as we maintain the parallel trends assumption and discuss the types of restrictions on selection compatible with it. Moreover, our analysis shows that parallel trends is compatible with various types of selection on unobservables, unlike standard unconfoundedness assumptions \citep[e.g.,][]{imbens2004nonparametric,imbens2009recent}.

Our third contribution is to the literature on sensitivity analysis, partial identification, and robust inference under violations of parallel trends. Existing work has proposed different ways to use pre-treatment information to bound the ATT 
\citep[e.g.,][]{manski2018how,rambachan2023a,ban2023generalized}. We complement this work by providing a selection-based decomposition of the bias of DiD and simple empirical strategies for benchmarking its components. Our approach uses pre-treatment periods to learn about the deviations from the selection-based necessary and sufficient conditions for parallel trends, whereas existing approaches rely on more ``reduced-form'' quantities, such as pre-treatment parallel trends violations. 
Our selection-based approach also differs from the analysis by \citet{marx2024parallel}. They derive partial identification results under monotone treatment selection assumptions on the untreated potential outcome, which they motivate using an economic model of learning with binary outcomes. By contrast, we characterize the bias of DiD in terms of deviations from the necessary and sufficient conditions for parallel trends.

Our fourth contribution is to the literature imposing explicit selection and/or outcome models to develop and compare different methods for estimating treatment effects, including DiD.\footnote{See, e.g., \citet{ashenfelter1985using}, \citet{heckman1985alternative}, \citet{card2005estimating}, \citet{chabeferret2015analysis}, \citet{blundell2009alternative}, \citet{dechaisemartin2018fuzzy}, \citet{verdier2020average}, \citet{dechaisemartin2022not}, \citet{marx2024parallel}, \citet{chabetferret2025should}. Most papers in this literature examine sharp DiD designs, as we do, \citet{dechaisemartin2018fuzzy} and \citet{marx2024parallel} also consider fuzzy DiD designs.} We contribute to this literature by providing general necessary and sufficient conditions for parallel trends, which are derived for general selection and outcome models that nest models considered in this literature. Our conditions thus clarify trade-offs between assumptions on selection and time-varying unobservables that are relevant for those models. Within this strand of the literature, our paper is most closely related to contemporaneous work by \citet{marx2024parallel}, though our focus markedly differs from theirs. \citet{marx2024parallel} analyze parallel trends through the lens of various examples of dynamic choice models. In doing so, they focus on explicit models of selection, and many of their examples are for binary outcomes. By contrast, we provide general necessary and sufficient conditions for parallel trends without imposing restrictions on the nature of the outcome variables or explicit models of dynamic choice. The general classes of selection mechanisms we consider nest, and can be motivated by, dynamic choice models. Unlike \citet{marx2024parallel}, we also incorporate covariates into our analysis and study settings with multiple groups and periods. 

Finally, a byproduct of our analysis is an explicit connection between DiD and the literature on nonseparable panel models. In Appendix \ref{app:connections}, we show that our sufficient conditions for parallel trends imply combinations of identifying assumptions in this literature \citep[e.g.][]{altonji2005cross,bester2009identification, hoderlein2012nonparametric,chernozhukov2013average}.

\subsection{Notation} For a random vector $W_{it}$, where $i=1,\dots,n$ and $t=1,2$, we let $\dot{W}_{it}\equiv W_{it}-E[W_{it}]$ and denote its time series by $W_i\equiv (W_{i1},W_{i2})$.\footnote{We define all vectors in this paper as row vectors.} We use $F_W$ to denote the distribution of the random vector $W$ and $\mathcal{W}$ to denote its support. Let $f(z,w)$ be a function defined on $\mathcal{Z}\times\mathcal{W}$. We say that $f(z,w)$ is a trivial function of $w$ if $f(z,w)=f(z,w')=h(z)$ for all $z\in\mathcal{Z}$, $w\neq w'$, and $(w,w')\in\mathcal{W}^2$. We say that $f(z,w)$ is a symmetric function in $z$ and $w$ if $f(z,w)=f(w,z)$ for all $(z,w)\in\mathcal{Z}\times\mathcal{W}$. We use the notation $\overset{d}{=}$ to denote equality of distribution. For random variables, $X_i$, $Z_i$, and $W_i$, $Z_i|W_i,X_i\overset{d}{=}Z_i|X_i,W_i$ denotes that $F_{Z_i|W_i,X_i}(z|w,x)=F_{Z_i|X_i,W_i}(z|w,x)$ for $(z,w,x)\in\mathcal{Z}\times\mathcal{W}\times\mathcal{X}$. 

\section{Setup, selection mechanism, and examples}

In the main text, we consider the classical DiD setup with two groups and two periods, where the selection decision is made at the same level as the unit of observation. We extend our results to settings with disaggregate data (e.g., on individuals) where the selection decision is made at a more aggregate level (e.g., at the state level) in Appendix \ref{app:disaggregate} as well as to  DiD designs with multiple groups and multiple periods in Appendix \ref{app: general did}. 

Let $D_{it}$ and $Y_{it}$ denote the treatment status and outcome for unit $i\in \{1,\dots,n\}$ in period $t\in \{1,2\}$. Here the index $i$ refers to the unit making the decision to select into treatment. This could be an individual or a more aggregate administrative unit, such as a county or state. The treatment group ($G_i=1$) selects the treatment path $D_i=(0,1)$; the control group ($G_i=0$) selects $D_i=(0,0)$. The potential outcomes with and without the treatment are $Y_{it}(1)$ and $Y_{it}(0)$, respectively.\footnote{To focus attention on the role of the parallel trends assumption, we assume that there are no anticipatory effects. This is a standard assumption in the DiD literature. See, for example, \citet{roth_et_al_DiD_survey} for a discussion.} We abstract from covariates for now to focus on the issues arising from selection on time-invariant and time-varying unobservables. We discuss the additional implications of including covariates in Appendix \ref{app:covariates}.

We consider the standard parallel trends assumption. Throughout the paper, we assume that all relevant moments exist and $\{Y_{i1}(0),Y_{i2}(0),G_i\}$ is i.i.d. across $i$.
\begin{customass}{PT}
\label{ass:PT} 
The (unconditional) parallel trends assumption holds: $$E[Y_{i2}(0)-Y_{i1}(0)|G_i=1]=E[Y_{i2}(0)-Y_{i1}(0)|G_i=0].$$
\end{customass}
Under Assumption \ref{ass:PT}, the average treatment effect on the treated group in period $t=2$, $\att\equiv E[Y_{i2}(1)-Y_{i2}(0)|G_i=1]$, is identified from the ``difference-in-differences'' as follows:
$$
\att=E[Y_{i2}-Y_{i1}|G_i=1]-E[Y_{i2}-Y_{i1}|G_i=0]\equiv \did.
$$

We work with a general nonseparable model for $Y_{it}(0)$,
\begin{equation}\label{eq:nonseparable_outcome}
Y_{it}(0)=\xi_t(\ai,\eit), \quad i=1,\dots,n, \quad t=1, 2,
\end{equation}
where $\ai$, $\eione$, and $\eitwo$ are finite-dimensional vector-valued random variables, and $\xi_t(\cdot)$ is an unrestricted time-varying function. The outcome model \eqref{eq:nonseparable_outcome}, while not imposing any restrictions on $Y_{it}(0)$, allows us to distinguish between time-invariant and time-varying unobservables. This is necessary to define selection mechanisms that can directly depend on these unobservables. If, instead, we were to work directly with potential outcomes, this would rule out important examples of selection mechanisms, such as selection on fixed effects \citep[e.g.,][]{ashenfelter1985using}.

The determinants of selection into treatment vary widely across DiD applications. In some applications, the unit $i$ making the selection decision is a state or county, whereas in other settings it is an individual economic agent, such as an individual, a household, or a firm. Below, we therefore introduce a general selection mechanism that  accommodates many different types of selection. To motivate this general selection mechanism, it is helpful to first consider examples of selection mechanisms relevant for our setting. 

\begin{example}[Selection on untreated outcomes]\label{ex:selection_outcomes}
Selection on untreated potential outcomes goes back  to at least the seminal work of \citet{ashenfelter1985using} in the context of individuals selecting into job training programs. It remains relevant in DiD applications where selection decisions are made by more aggregate units, which are prevalent in contemporary empirical research. For example, in a survey of governors on the reasons for expanding Medicaid, \citet{sommers2013us} find that the health outcomes in their states were one of the factors affecting their support for Medicaid expansion.

\citet[][p.651]{ashenfelter1985using} studied the case where individuals select into the training programs if their pre-treatment earnings $Y_{i1}(0)$ fall below a fixed threshold, so that $G_i=1\left\{Y_{i1}(0)\le c\right\}$.\footnote{\citet{ashenfelter1985using} consider a more general setting with multiple pre-treatment periods, where selection depends on earnings in the $k$th pre-treatment period.} In the context of Medicaid expansion, policymakers in a state might expand Medicaid if the health outcomes of their constituents (before the expansion) fall below a certain threshold. 

More generally, let $\omega_i$ denote the information set available to the units when deciding whether to select into the treatment and consider the following mechanism,
\begin{align}
G_i=1\left\{E[Y_{i1}(0)+\beta Y_{i2}(0)|\omega_i]\le E[\kappa_{i2}|\omega_i]\right\},\label{eq:AC_selection}
\end{align}
where $\beta\in [0,1]$ is a discount factor and $\kappa_{i2}$ is a unit-specific random threshold. This example demonstrates the importance of allowing for additional unobservables in addition to the determinants of $Y_{it}(0)$, $(\ai,\eione,\eitwo)$. \qed
\end{example}

\begin{example}[Selection on treatment effects (Roy-style selection)]\label{ex:selection_effects} 
Let $\tau_{i2}$ denote the treatment effect in period $t=2$, $\tau_{i2}=Y_{i2}(1)-Y_{i2}(0)$. Suppose that units select into the treatment if the expected gains from treatment given the information set $\omega_i$, $E[\tau_{i2}|\omega_i]$, exceed the expected cost of treatment, $E[\kappa_{i2}|\omega_i]$, $G_i=1\{ E[\tau_{i2}|\omega_i]\ge E[ \kappa_{i2}|\omega_i]\}.$
In the context of the job training setting, units might decide to participate in the program if the expected gain in their earnings outweighs the expected costs, whereas in the Medicaid setting, policymakers might support a Medicaid expansion if the expected improvements in the health outcomes of their constituents outweigh the expected costs. Indeed, these expected improvements were among the factors mentioned in the survey in \citet{sommers2013us}. \qed
\end{example}
\begin{example}[Selection on fixed effects]\label{ex:selection_fe} DiD methods have traditionally been motivated using two-way fixed effects models. Fixed effects assumptions allow for unrestricted dependence between time-invariant unobservables and the regressors, thereby implicitly allowing for selection on time-invariant unobservables.\footnote{See, e.g., \citet{chamberlain1984panel,arellano2003panel,evdokimov2010identification,wooldridge2010econometric,hoderlein2012nonparametric,chernozhukov2013average}.} A simple example is $G_i=1\{\ai\le c\}$, where $\ai$ is a scalar, which corresponds to the selection mechanism on p.650 in \citet{ashenfelter1985using}. In the training program example, this mechanism captures settings where individuals select into the program if the permanent earnings component $\ai$ falls below a threshold $c$, which is ``based on potential trainees' discount rates, time horizons, and tastes for training'' \citep[][p.651]{ashenfelter1985using}. In the context of aggregate treatments, such as Medicaid expansion, policymakers might decide to expand Medicaid based on their political affiliations or state-specific characteristics, which are plausibly time-invariant. \qed
\end{example}

In the previous examples, all agents make their selection decision in the same way and based on the same information set $\oi$. This is likely an oversimplification in many contexts where DiD is used, especially in aggregate selection settings. In the following example, we consider a setting with heterogeneous units whose selection decisions depend on different unobservables.
\begin{example}[Selection with heterogeneous units]\label{ex:heterogneous_units} For simplicity, we consider a setting with two types of units, noting that the example can be easily generalized to settings with more than two types. 
Let $\mu_i\in \{0,1\}$ be an indicator for the unit's type. If $\mu_i=1$, unit $i$ adopts the treatment if the expected benefits outweigh the expected costs given $\omega_i^1$; if $\mu_i=0$, unit $i$ selects into treatment if the expected discounted sum of untreated outcomes given $\omega_i^0$ falls below a certain threshold, so that
\begin{equation*}
G_i= 1\{E[Y_{i2}(1)-Y_{i2}(0)|\omega_i^{1}]\geq E[\kappa_{i2}|\omega_i^1]\}^{\mu_i}1\{E[Y_{i1}(0)+\beta Y_{i2}(0)|\omega_i^{0}]\leq  E[\kappa_{i2}|\omega_i^0]\}^{1-\mu_i}.
\end{equation*}
Here, $G_i$ depends on the unobserved type $\mu_i$ as well as the information sets used by both types, $\omega_i^0$ and $\omega_i^1$.
\qed
\end{example}

Motivated by these examples, we consider the following general selection mechanism, 
\begin{eqnarray}
G_i=g(\omega_i,\nu_i), \quad i=1,\dots,n. \label{eq:selection_mechanism}
\end{eqnarray}
We allow $\omega_i$ to be a function or a subvector of $(\ai,\eione,\eitwo,\mu_i,\eta_{i1},\eta_{i2})$ and can thereby accommodate the above examples as special cases. Moreover, we can accommodate many other economic models of selection  \citep[e.g.,][]{heckman1985alternative,chabeferret2015analysis, marx2024parallel}. The additional unobservables $(\mu_i,\eta_{i1},\eta_{i2})$ capture any other determinants of selection that may be correlated with $Y_{it}(0)$, such as individual-specific thresholds in Example \ref{ex:selection_outcomes}, treatment effects and costs in Example \ref{ex:selection_effects}, or a unit's type in Example \ref{ex:heterogneous_units}.

The scalar unobservable, $\nu_i$, omitted for simplicity in the above examples, captures determinants of selection that are independent of $(\ai,\eione,\eitwo,\mu_i,\eta_{i1},\eta_{i2})$ and allows us to accommodate the important special case of random assignment. We impose the following assumption.
\begin{customass}
{SEL}\label{ass:SEL}  $P(\nu_i>c)\in(0,1)$ for some $c\in\mathbb{R}$ and $\nu_i\indep (\ai,\eione,\eitwo,\mu_i,\eta_{i1},\eta_{i2})$.
\end{customass}

Note that since $G_i=D_{i2}$, $g$ can be equivalently viewed as the selection mechanism for $D_{i2}$.
Let $\mathcal{O}$ and $\mathcal{V}$ denote the supports of $\oi$ and $\nu_i$, respectively. Denote by $\mathcal{G}_{\omega}$ the class of all selection mechanisms $g$ mapping from $\mathcal{O}\times \mathcal{V}$ to $\{0,1\}$, that is,
$
\mathcal{G}_{\omega}\equiv\{g:\mathcal{O}\times \mathcal{V}\to \{0,1\}\}.
$
Here, with a slight abuse of notation, we index the class of selection mechanisms by $\oi$ since it could be correlated with the potential outcomes.\footnote{Technically, $\mathcal{G}$ should be indexed by $\mathcal{O}$ (or, alternatively, $\mathcal{O}$ and $\mathcal{V}$), but given the central role of selection on $\oi$ in our analysis, we index $\mathcal{G}$ by $\omega$ to improve readability.} Since $\oi$ can be interpreted as the units' information sets, the categorization of selection mechanisms based on $\oi$ is motivated by the usefulness of information sets for analyzing causal inference methods \citep[see, e.g., Section 2 in][]{heckman2007econometric}.

\begin{remark}[Parallel trends and functional form] Throughout this paper, we take the functional form of the outcome as given. We thereby abstract from the issues arising from the sensitivity of DiD to functional form specification  \citep[e.g.,][]{roth2021when}. \qed
\end{remark}

\section{Necessary and sufficient conditions for parallel trends}\label{sec:necessary_sufficient}

In this section, we provide necessary and sufficient conditions for parallel trends. Section \ref{sec:master} provides a general necessary and sufficient condition for parallel trends. We then apply this general condition in  Sections \ref{sec:no_restrictions}--\ref{sec:restricted_classes} to derive necessary and sufficient conditions for various practically relevant scenarios by specifying what units select on.

\subsection{A general necessary and sufficient condition for parallel trends}\label{sec:master}

Here, we provide a general result that allows for specifying necessary and sufficient conditions for parallel trends for any class of selection mechanisms $\mathcal{G}_\omega$. Consistent with the nonparametric treatment of selection mechanisms in DiD analyses, we provide necessary and sufficient conditions for Assumption \ref{ass:PT} to hold for all $g\in \mathcal{G}_\omega$.\footnote{These conditions imply that Assumption \ref{ass:PT} holds for all $g\in \mathcal{G}_\omega$, which is in the spirit of standard notions of nonparametric identification \citep[e.g.,][Definition 2.10]{hansen2022econometrics}.}

Considering necessary and sufficient conditions for Assumption \ref{ass:PT} for all $g\in\mathcal{G}_\omega$ ensures that the invocation of Assumption \ref{ass:PT} is robust to the choice of selection mechanism within this class.\footnote{In Section \ref{subsec:for_all_why}, we further demonstrate that, even in the case where a researcher knows the underlying selection mechanism, the necessary and sufficient condition for all $g\in\mathcal{G}_\omega$ rules out cases where Assumption \ref{ass:PT} holds due to \emph{peculiar} combinations of parameters of the data-generating process (see Figure \ref{fig:numerical}).} While robustness to the exact specification of the selection mechanism is important in all DiD applications, it is especially relevant in settings with aggregate selection (e.g., at the county or state level). When multiple entities or actors are involved, or when the decision is based on aggregating heterogeneous individual preferences, the selection mechanisms are often complex and difficult to model. Consider again the Medicaid expansion example. The survey of governors in \citet{sommers2013us} shows that there are many different factors  affecting the governors' views on expanding Medicaid, ranging from fiscal considerations to the potential health outcomes of their constituents and the potential benefits of the expansion in terms of health insurance coverage and health outcomes.

The following theorem provides a necessary and sufficient condition for any class of selection mechanisms $\mathcal{G}_\omega$. We focus on non-degenerate DiD designs, that is, designs with $P(G_i=1)\in (0,1)$. Recall the notation $\dot{Y}_{it}(0)= Y_{it}(0)-E[Y_{it}(0)]$.

\begin{theorem}[Necessary and sufficient condition for \ref{ass:PT} to hold for all $g\in \mathcal{G}_\omega$] \label{thm:main} Suppose that Assumption \ref{ass:SEL} holds. 
 Suppose further that either $P(E[\dot{Y}_{i2}(0)-\dot{Y}_{i1}(0)|\omega_i]>0)<1$ or $P(E[\dot{Y}_{i2}(0)-\dot{Y}_{i1}(0)|\omega_i]<0)<1$. Then, Assumption \ref{ass:PT} holds for all $g\in \mathcal{G}_{\omega}$ satisfying $P(G_i=1)\in (0,1)$ if and only if $E[\dot{Y}_{i2}(0)-\dot{Y}_{i1}(0)|\omega_i]=0\text{ a.s.}$
\end{theorem}

\begin{proof} \textbf{``$\Longrightarrow$''}: We prove the result for the case where $P(E[\dot{Y}_{i2}(0)-\dot{Y}_{i1}(0)|\omega_i]>0)<1$. The proof for the case where $P(E[\dot{Y}_{i2}(0)-\dot{Y}_{i1}(0)|\omega_i]<0)<1$ follows from the same arguments. 

Under Assumption \ref{ass:SEL} and because $P(E[\dot{Y}_{i2}(0)-\dot{Y}_{i1}(0)|\omega_i]>0)<1$, the selection mechanism 
\begin{equation}
G_i=1\{\nu_i>c\}1\{E[\dot{Y}_{i2}(0)-\dot{Y}_{i1}(0)|\omega_i]\le 0\}\label{eq:least_favorable_mechanism_unc}
\end{equation}
is nondegenerate, that is,
$
P(1\{\nu_i>c\}1\{E[\dot{Y}_{i2}(0)-\dot{Y}_{i1}(0)|\omega_i]\le 0\}=1)\in (0,1).
$

If Assumption \ref{ass:PT} holds for all non-degenerate selection mechanisms, then it holds for the mechanism \eqref{eq:least_favorable_mechanism_unc}. By Lemma \ref{lem: PT equivalent condition}, Assumption \ref{ass:PT} holding for the mechanism in \eqref{eq:least_favorable_mechanism_unc} is equivalent to 
\begin{equation*}
E[1\{\nu_i>c\}1\{E[\dot{Y}_{i2}(0)-\dot{Y}_{i1}(0)|\omega_i]\le 0\}(\dot{Y}_{i2}(0)-\dot{Y}_{i1}(0))]=0,
\end{equation*}
which, by Assumption \ref{ass:SEL}, is equivalent to
\begin{equation*}
E[1\{E[\dot{Y}_{i2}(0)-\dot{Y}_{i1}(0)|\omega_i]\le 0\}(\dot{Y}_{i2}(0)-\dot{Y}_{i1}(0))]=0.
\end{equation*}

By the law of iterated expectations (LIE), this is further equivalent to 
\begin{eqnarray*}
E[1\{E[\dot{Y}_{i2}(0)-\dot{Y}_{i1}(0)|\omega_i]\le 0 \}E[\dot{Y}_{i2}(0)-\dot{Y}_{i1}(0)|\omega_i]]=0
\end{eqnarray*}
Since $E[E[\dot{Y}_{i2}(0)-\dot{Y}_{i1}(0)|\omega_i]]=0$, the result follows by Lemma \ref{lem:necessary}.

\smallskip

\noindent \textbf{``$\Longleftarrow$''}: By the LIE, we have that
\begin{eqnarray*}
E[\dot{Y}_{i2}(0)-\dot{Y}_{i1}(0)|G_i]=E[E[\dot{Y}_{i2}(0)-\dot{Y}_{i1}(0)|\omega_i,\nu_i]|G_i]=E[E[\dot{Y}_{i2}(0)-\dot{Y}_{i1}(0)|\omega_i]|G_i]=0,
\end{eqnarray*}
where the penultimate equality follows from Assumption \ref{ass:SEL}.
\end{proof}

Theorem \ref{thm:main} says that Assumption \ref{ass:PT} holds for all nondegenerate selection mechanisms if and only if the trend in the untreated potential outcomes is mean independent of the systematic determinants of selection.  
Conversely, if the units select on unobservables determining the trend in the untreated potential outcomes, then there exists (at least) one selection mechanism that breaks parallel trends. We summarize this implication in the following corollary.\footnote{We thank anonymous referees for suggesting this way of interpreting and restating the result in Theorem \ref{thm:main}.}
\begin{corollary}[Contrapositive]\label{cor:contrapositive} Under the assumptions of Theorem \ref{thm:main}, if the necessary and sufficient condition $E[\dot{Y}_{i2}(0)-\dot{Y}_{i1}(0)|\omega_i]=0\text{ a.s.}$ fails, then there exists (at least) one selection mechanism $g\in \mathcal{G}_{\omega}$ satisfying $P(G_i=1)\in (0,1)$ such that Assumption \ref{ass:PT} fails.
\end{corollary}

Stating the necessary and sufficient condition in Theorem \ref{thm:main} in terms of what the units select on and their information sets, $\oi$, rather than explicit selection mechanisms, has a key advantage: While it may be possible to determine (based on contextual and economic knowledge) what information units select on, there is likely uncertainty about the exact form of the selection mechanism, especially in settings with aggregate selection. The necessary and sufficient condition in Theorem \ref{thm:main} therefore relieves the researcher from the need to impose additional structure on the selection mechanism that is not justified by their context.

Next, we apply Theorem \ref{thm:main} to various practically relevant classes of selection mechanisms. To do so, we specialize the necessary and sufficient condition, $E[\dot{Y}_{i2}(0)-\dot{Y}_{i1}(0)|\omega_i]=0$, under different choices of $\oi$.

\begin{remark}[Parallel trends for all $g$ in a subclass of $\mathcal{G}_\omega$] \label{rem:subclasses}One might wonder how the necessary and sufficient conditions change when we focus on a subclass $\tilde{\mathcal{G}}\subset \mathcal{G}_\omega$, where $\mathcal{G}_\omega$ is an initial set of selection mechanisms. There are at least two ways to define subclasses in our framework. First, one can consider $\tilde{\mathcal{G}}=\mathcal{G}_{\tilde\omega}$, where $\tilde{\omega}$ is a subvector of $\omega$. The focus on such subclasses can often be motivated based on contextual and economic information. Given that we index the selection mechanisms by $\omega$, necessary and sufficient conditions for Assumption \ref{ass:PT} for all $g\in \mathcal{G}_{\tilde\omega}$ follow as direct corollaries from Theorem \ref{thm:main} with $\omega_i$ replaced by $\tilde\omega_i$, as we illustrate below.  Second, one can consider the same choice of $\omega$ but impose additional parametric and functional restrictions on $g$. The result in Theorem \ref{thm:main} continues to hold as long as these additional restrictions do not exclude the mechanism \eqref{eq:least_favorable_mechanism_unc} used to prove the ``only if'' direction. We do not explore such subclasses here because it is often difficult to justify parametric and functional form assumptions in DiD applications based on contextual and economic knowledge alone. 
\qed

\end{remark}

\subsection{What if we impose no restrictions on selection?}\label{sec:no_restrictions}

Unlike other causal inference methods, DiD does not explicitly restrict selection into treatment. This begs the question: What if researchers are indeed not willing to impose any assumptions on selection so that parallel trends needs to hold for all selection mechanisms? To answer this question, we apply Theorem \ref{thm:main} with $\oi$ including all unobservables that could enter the selection mechanism in \eqref{eq:selection_mechanism}, $\oi=(\ai,\eione,\eitwo,\mu_i,\eta_{i1},\eta_{i2})$. In this case, because $\oi$ contains the unobservable determinants of $\dot{Y}_{i2}(0)$ and $\dot{Y}_{i1}(0)$, we have that
$
E[\dot{Y}_{i2}(0)-\dot{Y}_{i1}(0)|\oi]=\dot{Y}_{i2}(0)-\dot{Y}_{i1}(0), 
$ yielding the following necessary and sufficient condition.
\begin{corollary}[No restrictions on selection]\label{cor:no_restriction} Under the assumptions of Theorem \ref{thm:main} with $\oi=(\ai,\eione,\eitwo,\mu_i,\eta_{i1},\eta_{i2})$, Assumption \ref{ass:PT} holds for all $g\in\mathcal{G}_\omega$ satisfying $P(G_i=1)\in(0,1)$ if and only if $\dot{Y}_{i1}(0)=\dot{Y}_{i2}(0)$ a.s. The result continues to hold if $\oi=(\ai,\eione,\eitwo)$.
\end{corollary}

To interpret the necessary and sufficient condition in Corollary \ref{cor:no_restriction}, it is helpful to rewrite it as 
$$
Y_{i2}(0)-Y_{i1}(0)=E[Y_{i2}(0)-Y_{i1}(0)].
$$
This shows that absent any restrictions on selection, parallel trends implies that the potential outcomes are constant over time, except for common mean shifts. This essentially rules out time-varying unobservables.  To see this, consider the following standard two-way model  \begin{equation}
Y_{it}(0)=\ai + \lt +\eit, \quad E[\eit]=0, \label{eq:separable}
\end{equation}
where $\lt$ is a nonstochastic time trend. The necessary and sufficient condition specialized to this separable outcome model is $\eione=\eitwo$, implying that $\eit$ is time-invariant.

Given that the necessary and sufficient condition for the unrestricted class of selection mechanisms is implausible in most applications, we next consider restricted classes of selection mechanisms. 
\begin{remark}[Testability of necessary and sufficient condition]
The necessary and sufficient condition in Corollary \ref{cor:no_restriction}, $\dot{Y}_{i1}(0)=\dot{Y}_{i2}(0)$, is testable using the data from the control group. This is noteworthy because Assumption \ref{ass:PT} is generally untestable.\footnote{For a discussion on necessary and sufficient conditions for parallel pre-trends, see \citet{ghanem2026when}.} \qed
\end{remark}
\subsection{Necessary and sufficient conditions for restricted classes of mechanisms}
\label{sec:restricted_classes}

DiD applications differ substantially in terms of what determines selection into treatment. Corollary \ref{cor:no_restriction} shows that restrictions on selection are unavoidable in realistic settings. In the following, we consider various restrictions on selection mechanisms that are practically relevant and well-established in the literature. The list of restrictions we consider is not exhaustive. The advantage of Theorem \ref{thm:main} is that researchers can specialize the necessary and sufficient condition for the class of selection mechanisms relevant for their application.

\subsubsection{Imperfect foresight}

Selection on pre-treatment unobservables is prevalent in DiD applications. An example is when units make their selection decision based on expected future potential outcomes and costs, while only having access to pre-treatment information (e.g., $\omega_i=(\ai,\eione,\mu_i,\eta_{i1})$ in Examples \ref{ex:selection_outcomes} and \ref{ex:selection_effects}). Another example is when units select into treatment in response to negative (or positive) pre-treatment shocks or their pre-treatment outcome falling below (or above) a specific threshold (e.g., Example \ref{ex:selection_outcomes} with $\beta=0$ and $\omega_i=(\ai,\eione,\mu_i,\eta_{i1})$) and more broadly when there is feedback \citep[e.g.,][]{bonhomme2025back,chamberlain2022feedback}. Finally, imperfect foresight is relevant when individuals are myopic.

We first consider the case where selection depends on the time-invariant and all pre-treatment unobservables, but not on the unobservables in the post-treatment period (unlike in Corollary \ref{cor:no_restriction}). Thus, $\omega_i=(\ai,\eione,\nu_i,\eta_{i1})$. For this case, Theorem \ref{thm:main} implies the following corollary.\footnote{We are grateful to Eric Mbakop for encouraging us to pursue necessary and sufficient conditions instead of necessary conditions only under imperfect foresight (and selection on fixed effects, discussed below).}
\begin{corollary}[Imperfect foresight: Case 1] \label{cor:if1} Under the assumptions of Theorem \ref{thm:main} with $\oi=(\ai,\eione,\mu_i,\eta_{i1})$, Assumption \ref{ass:PT} holds for all $g\in \mathcal{G}_{\omega}$ satisfying $P(G_i=1)\in (0,1)$ if and only if $E[\dot{Y}_{i2}(0)|\ai,\eione,\mu_i,\eta_{i1}]=\dot{Y}_{i1}(0)$ a.s.
\end{corollary}

The necessary and sufficient condition in Corollary \ref{cor:if1} is a martingale-type condition on the untreated potential outcomes with respect to the unobservables that determine selection in this case. To build further intuition and to compare this result to Corollary \ref{cor:no_restriction}, note that the necessary and sufficient condition can be equivalently written as
$$
Y_{i2}(0)-Y_{i1}(0)=E[Y_{i2}(0)-Y_{i1}(0)]+\zeta_{i2}, \quad E[\zeta_{i2}|\ai,\eione,\mu_i,\eta_{i1}]=0.
$$
That is, the  condition in Corollary \ref{cor:if1} allows the untreated potential outcomes to vary over time beyond deterministic mean shifts but requires the stochastic component of the change over time, $\zeta_{i2}$, to be mean-independent of the pre-treatment unobservables $(\ai,\eione,\mu_i,\eta_{i1})$.

In the two-way model \eqref{eq:separable}, the condition in Corollary \ref{cor:if1} becomes $E[\eitwo|\ai,\eione,\mu_i,\eta_{i1}]=\eione$, a martingale-type property that implies that $\eitwo-\eione+\zeta_{i2}$, where $\zeta_{i2}$ is an innovation satisfying $E[\zeta_{i2}|\ai,\eione,\mu_i,\eta_{i1}]=0$. This necessary and sufficient condition relates to the consistency of the first-difference estimator under violations of strict exogeneity when the idiosyncratic shocks follow a unit root.\footnote{We thank St\'ephane Bonhomme for pointing out this connection.} 

The martingale-type condition in Corollary \ref{cor:if1} arises because the units select on the determinants of $Y_{i1}(0)$, $(\ai,\eione)$. As a result, the condition in  Theorem \ref{thm:main} with $\omega_i=(\ai,\eione,\nu_i,\eta_{i1})$,
$
E[\dot{Y}_{i2}(0)-\dot{Y}_{i1}(0)|\ai,\eione,\mu_i,\eta_{i1}]=0$, simplifies to $E[\dot{Y}_{i2}(0)|\ai,\eione,\mu_i,\eta_{i1}]=\dot{Y}_{i1}(0)$, as stated in Corollary \ref{cor:if1}. If selection is based on pre-treatment unobservables that do not include the determinants of $Y_{i1}(0)$, the martingale condition does not arise, as the next corollary shows.
\begin{corollary}[Imperfect foresight: Case 2]\label{cor:if2} Under the assumptions of Theorem \ref{thm:main} with $\oi=(\mu_i,\eta_{i1})$, Assumption \ref{ass:PT} holds for all $g\in \mathcal{G}_{\omega}$ satisfying $P(G_i=1)\in (0,1)$ if and only if $E[\dot{Y}_{i2}(0)-\dot{Y}_{i1}(0)|\mu_i,\eta_{i1}]=0$ a.s.
\end{corollary}
The necessary and sufficient condition in Corollary \ref{cor:if2} resembles but is different from the standard definition of a martingale-difference condition on the difference $\Delta \dot{Y}_{i2}(0)=\dot{Y}_{i2}(0)-\dot{Y}_{i1}(0)$ \citep[e.g.,][p.189]{hamilton1994time} since the conditioning set neither includes lagged values of $\Delta \dot{Y}_{i2}(0)$ nor its determinants. It can therefore be consistent with a wider class of time series processes than the condition in Corollary \ref{cor:if1}, which implies that $\dot{Y}_{it}(0)$ is a martingale.

\subsubsection{Roy-style selection} 
Here, we consider settings with Roy-style selection, as in Example \ref{ex:selection_effects}. We first consider the case where the units know their treatment effect $\tau_{i2}=Y_{i2}(1)-Y_{i2}(0)$ and costs $\kappa_{i2}$ in period $t=2$. For this case, Theorem \ref{thm:main} implies the following corollary. 
\begin{corollary}[Roy-style selection] \label{cor:roy_selection} Under the conditions of Theorem \ref{thm:main} with $\oi=(\tau_{i2},\kappa_{i2})$, Assumption \ref{ass:PT} holds for all $g\in \mathcal{G}_{\omega}$ satisfying $P(G_i=1)\in (0,1)$ if and only if $E[\dot{Y}_{i2}(0)-\dot{Y}_{i1}(0)|\tau_{i2},\kappa_{i2}]=0$ a.s.
\end{corollary}
Rewriting the necessary and sufficient condition as  $E[\dot{Y}_{i2}(0)|\tau_{i2},\kappa_{i2}]=E[\dot{Y}_{i1}(0)|\tau_{i2},\kappa_{i2}]$
demonstrates that it requires the conditional expectation of the demeaned untreated potential outcome given the treatment effects and costs to be equal across time. The condition would hold immediately if $(\tau_{i2},\kappa_{i2})$ were independent of the untreated potential outcomes. However, this is an arguably unrealistic restriction in many applications. 

In the two-way model \eqref{eq:separable},  the necessary and sufficient condition in Corollary \ref{cor:roy_selection} simplifies to 
$E[\eitwo|\tau_{i2},\kappa_{i2}]=E[\eione|\tau_{i2},\kappa_{i2}].$
The condition would be clearly violated if $\tau_{i2}$ is a monotonic transformation of $\eitwo$. The condition is more plausible if instead $\tau_{i2}$ and $\kappa_{i2}$ were determined by time-invariant factors. This discussion demonstrates that under Roy-style selection, parallel trends implies restrictions on treatment effect heterogeneity.

In some applications, assuming that the units know their treatment effects and costs in $t=2$ might not be plausible. Suppose instead that selection is based on expected treatment effects and costs conditional on all the available pre-treatment information, $E[\tau_{i2}|\oi]$ and $E[\kappa_{i2}|\oi]$, respectively, where $\oi=(\ai,\eione,\mu_i,\eta_{i1})$. For this case, the result in Corollary \ref{cor:if1} applies, and the necessary and sufficient condition for parallel trends is the martingale-type condition, $E[\dot{Y}_{i2}(0)|\ai,\eione,\mu_i,\eta_{i1}]=\dot{Y}_{i1}(0)$.  If, instead, the expectations are conditional on the time-invariant unobservables $(\ai,\mu_i)$, then the result in Corollary \ref{cor:fe} below applies. More generally, Theorem \ref{thm:main} allows for considering many other variants of Roy-style selection.

\subsubsection{Selection on fixed effects} Here, we consider the classical case of selection on fixed effects. Selection on fixed effects is plausible, for example, if the units' information sets only contain the time-invariant unobservables (in addition to $\nu_i$), so that $\omega_i=(\alpha_i,\mu_i)$, or if selection is directly based on fixed effects, as in Example \ref{ex:selection_fe}. 

The following corollary provides the necessary and sufficient condition under selection on fixed effects.
\begin{corollary}[Selection on fixed effects]\label{cor:fe} Under the conditions of Theorem \ref{thm:main} with $\oi=(\ai,\mu_i)$, Assumption \ref{ass:PT} holds for all $g\in \mathcal{G}_{\omega}$ satisfying $P(G_i=1)\in (0,1)$ if and only if  $E[\dot{Y}_{i2}(0)-\dot{Y}_{i1}(0)|\ai,\mu_i]=0$ a.s.
\end{corollary}

Corollary \ref{cor:fe} shows that the necessary and sufficient condition for parallel trends under selection on fixed effects is a time-homogeneity restriction on the conditional mean of the untreated potential outcome. To interpret this necessary condition, note that it implies that
$$
Y_{i2}(0)-Y_{i1}(0)=E[Y_{i2}(0)-Y_{i1}(0)]+\zeta_{i2}, \quad E[\zeta_{i2}|\ai,\mu_i]=0.
$$
This shows that Corollary \ref{cor:fe} implies a weaker mean-independence condition on the stochastic trend component $\zeta_{i2}$ than, for example, Corollary \ref{cor:if1}, thus highlighting a trade-off between restrictions on selection and the evolution of the untreated potential outcomes over time.

In the context of the two-way model \eqref{eq:separable}, the necessary condition in Corollary \ref{cor:fe} simplifies to $E[\eione|\ai,\mu_i]=E[\eitwo|\ai,\mu_i]$, a time-homogeneity assumption on the conditional mean of the idiosyncratic shocks.\footnote{The time-homogeneity condition in Corollary \ref{cor:fe} relates to the strict exogeneity condition in fixed effects models. Suppose that $G_i=g(\ai)$, then the strict exogeneity assumption $E[\eit|G_i,\ai]=0$ implies that $E[\eit|\ai]=0$, which in turn implies the time homogeneity condition in Corollary \ref{cor:fe} with $\oi=\ai$.} While it is not surprising that the condition $E[\eione|\ai,\mu_i]=E[\eitwo|\ai,\mu_i]$ is sufficient for parallel trends if $G_i=g(\ai,\mu_i,\nu_i)$, Corollary \ref{cor:fe} demonstrates that this condition is in fact \emph{necessary} for parallel trends under model \eqref{eq:separable}.

\subsubsection{Selection on lagged outcomes (unconfoundedness)} 

A popular alternative to the parallel trends assumption is to assume that selection is based on lagged dependent variables \citep[e.g.,][]{angrist2009mostly,ding2019bracketing}, $Y_{i2}(0)\indep G_i\mid Y_{i1}(0)$, which we will refer to as \emph{unconfoundedness}. Here we apply Theorem \ref{thm:main} to characterize the parallel trends assumption under unconfoundedness and shed light on the connection between these popular assumptions. Setting $\oi=Y_{i1}(0)$ in Theorem \ref{thm:main}, so that $G_i$ satisfies unconfoundedness, provides such a characterization. 
\begin{corollary}[Unconfoundedness] \label{cor:unconfoundedness} Under the assumptions of Theorem \ref{thm:main} with $\oi=Y_{i1}(0)$, Assumption \ref{ass:PT} holds for all $g\in \mathcal{G}_{\omega}$ satisfying $P(G_i=1)\in (0,1)$ if and only if $E[\dot{Y}_{i2}(0)|Y_{i1}(0)]=\dot{Y}_{i1}(0)$ a.s.
\end{corollary}

Researchers imposing unconfoundedness typically do not impose additional explicit assumptions on the exact form of the selection mechanism. This provides an additional motivation for focusing on nonparametric conditions and deriving necessary and sufficient conditions for Assumption \ref{ass:PT} holding for all $g\in \mathcal{G}_{\omega}$.

Corollary \ref{cor:unconfoundedness} shows that parallel trends  is equivalent to the demeaned potential outcomes $\dot{Y}_{it}(0)$ satisfying a martingale property under unconfoundedness. Written in terms of original outcomes $Y_{it}(0)$, the condition becomes
$
E[Y_{i2}(0)|Y_{i1}(0)]=Y_{i1}(0)+E[Y_{i2}(0)-Y_{i1}(0)]
$, which holds, for example, if $Y_{it}(0)$ is a random walk with drift. 

It is interesting to relate the result in Corollary \ref{cor:unconfoundedness} to results in the existing literature.

\begin{remark}[Connection to \citet{ding2019bracketing}] Here, we connect the analysis to \citet{ding2019bracketing}  \citep[see also][]{angrist2009mostly}. \citet{ding2019bracketing} assume that 
\begin{equation}
E[Y_{i2}| Y_{i1},G_i]=\theta_1+\theta_2 Y_{i1}+\tau G_i.\label{eq:ding_li_cef}
\end{equation}
Proposition 1 in \citet{ding2019bracketing}, written in terms of population coefficients, implies that the estimand under \eqref{eq:ding_li_cef} is
\begin{equation}
E[Y_{i2}|G_i=1]-E[Y_{i2}|G_i=1]-\theta_2(E[Y_{i1}|G_i=1]-E[Y_{i1}|G_i=0])\label{eq:ding_li_ldv}
\end{equation}
This estimand is equivalent to $\did$ if and only if $\theta_2=1$.\footnote{Note that the main bracketing result, Theorem 1 in \citet{ding2019bracketing}, does not apply in this case since Condition 1 (stationarity) is violated.} 

To relate the result in \citet{ding2019bracketing} to Corollary \ref{cor:unconfoundedness}, note that under unconfoundedness, \eqref{eq:ding_li_cef} implies that
$
E[Y_{i2}(0)| Y_{i1}(0)]=\theta_1+\theta_2 Y_{i1}(0).
$
Corollary \ref{cor:unconfoundedness} implies that $\theta_2=1$ is necessary and sufficient for Assumption \ref{ass:PT} to hold. Since $\did=\att$ under Assumption \ref{ass:PT}, the result in \citet[][Proposition 1]{ding2019bracketing} is consistent with the necessary and sufficient condition in Corollary \ref{cor:unconfoundedness} under the linearity assumption \eqref{eq:ding_li_cef}. \qed
\end{remark}

\subsubsection{Other selection mechanisms and trade-offs}\label{sec:other}
The previous subsections illustrate the implications of Theorem \ref{thm:main} for various empirically relevant classes of selection mechanisms. Importantly, the result in Theorem \ref{thm:main} is very general and allows us to characterize the empirical content of parallel trends for many other relevant classes of selection mechanisms, including many of the existing selection models discussed in the literature and reviewed in Section \ref{sec:literature}. Applying Theorem \ref{thm:main} only requires specifying $\oi$, that is, what information the units select on. In practice, the specification of $\oi$ should be guided by contextual and economic knowledge.

Varying $\oi$ allows us to characterize trade-offs between restrictions on selection, encoded in $\oi$, and restrictions on the time series properties of the untreated potential outcomes. The richer the information that the units select on, the more restrictive the time series restriction necessary for parallel trends to hold.  The time series restrictions are particularly strong if selection is based on the unobservable determinants of $Y_{it}(0)$, that is, if $\oi$ includes or depends on $(\ai,\eione,\eitwo)$. 

The trade-offs between assumptions on selection and the time series properties of the untreated potential outcomes are particularly easy to see under explicit models for the untreated potential outcomes. To illustrate, suppose that $Y_{it}(0)$ is given by model \eqref{eq:separable}. Then, our necessary and sufficient conditions  imply that parallel trends holds, for example, in the following scenarios:\footnote{Theorem \ref{thm:main} provides a general framework for deriving sufficient conditions, depending on what unit select on. However, in some applications, researchers might be interested in imposing other types of assumptions on selection. In Appendix \ref{app:sufficient_condition_exchangeability}, we provide a sufficient condition based on symmetry of the selection mechanism.}
\begin{enumerate}[(a)]\setlength\itemsep{0pt}
\item Imperfect foresight (case 1): (i) $\oi=(\ai,\eione,\mu_i,\eta_{i1})$ and (ii) $E[\eitwo|\ai,\eione,\mu_i,\eta_{i1}]=\eione$

\item Imperfect foresight (case 2): (i) $\oi=(\mu_i,\eta_{i1})$ and (ii) $E[\eitwo-\eione|\mu_i,\eta_{i1}]=0$

\item Roy-style selection: (i) $\oi=(\tau_{i2},\kappa_{i2})$ and (ii) $E[\eitwo|\tau_{i2},\kappa_{i2}]=E[\eione|\tau_{i2},\kappa_{i2}]$
\item Selection on fixed effects: (i) $\oi=(\ai,\nu_i)$ and (ii) $E[\eitwo|\ai,\nu_i]=E[\eione|\ai,\nu_i]$

\end{enumerate}
The conditions (a)--(d) provide practitioners with explicit theory-based templates for assessing and justifying parallel trends assumptions and can be used in conjunction with the selection mechanisms in Examples \ref{ex:selection_outcomes}, \ref{ex:selection_effects}, \ref{ex:selection_fe}, and \ref{ex:heterogneous_units}, or other selection mechanisms in the literature. These conditions allow researchers to provide, in the words of \citet{mckenzie2022a}, ``plausible rhetorical arguments as to why we should think the [parallel trends] assumptions hold.''

\subsection{The relevance of parallel trends for all $g\in\mathcal{G}_{\omega}$ for empirical practice}
\label{subsec:for_all_why}

Theorem \ref{thm:main} and Corollaries \ref{cor:no_restriction}--\ref{cor:fe} present necessary and sufficient conditions for parallel trends to hold for all $g\in\mathcal{G}_{\omega}$ because we are interested in nonparametric conditions, consistent with the nonparametric (model-agnostic) treatment of selection mechanisms in the DiD analyses. Here we  elaborate on the practical relevance of focusing on \emph{parallel trends for all $g\in\mathcal{G}_{\omega}$}. This is a crucial question, as it might not be obvious why practitioners should consider  \emph{parallel trends for all $g\in\mathcal{G}_{\omega}$}, when it is clearly stronger than \emph{parallel trends for a specific $g\in\mathcal{G}_{\omega}$}.\footnote{Alternatively, one could consider parallel trends for a subclass $\tilde{\mathcal{G}}\subset \mathcal{G}$, as discussed in Remark \ref{rem:subclasses}. Suppose, for example, that contextual and economic information suggests that $\tilde{\mathcal{G}}=\mathcal{G}_{\tilde\omega}$, where $\tilde\omega$ is a subvector of $\omega$. In this case, the same discussion and issues are relevant for parallel trends for all $g\in \mathcal{G}_{\tilde\omega}$.} To keep our discussion concrete, we focus on the case of imperfect foresight where the units select on pre-treatment unobservables, so that $\oi=(\ai,\eione,\mu_i,\eta_{i1})$, but the arguments we make are relevant for any class of selection mechanisms.

First, in applications where contextual or economic knowledge  suggests that units select on pre-treatment unobservables, there is likely uncertainty about the exact form of the selection mechanism. For instance, practitioners might not want to take a stance on (i) how expectations are formed (e.g., subjective expectations may be different from conditional expectations) or (ii) whether selection is based on the discounted sum of expected untreated outcomes (Example \ref{ex:selection_outcomes}), on expected gains (Example \ref{ex:selection_effects}), or other quantities motivated by economic models of selection. As discussed above, the uncertainty about the exact form of the selection mechanism can be particularly pronounced in settings with aggregate selection.

Second, even if one is certain about the exact parametric form of the units' selection mechanism, the exact distribution of unobservables, and how expectations are formed, one would typically not want parallel trends to depend on specific parameter choices. To illustrate, consider Example \ref{ex:selection_outcomes} with $\omega_i=(\ai,\eione)$. Figure \ref{fig:numerical} plots the parallel trends violation, $E[Y_{i2}(0)-Y_{i1}(0)|G_i=1]-E[Y_{i2}(0)-Y_{i1}(0)|G_i=0]$, for different discount factors $\beta$ against $Cov(\eione,\eitwo)$ for two different outcome models: (a) a separable model with autocorrelated shocks, (b) an autoregressive model with a drift. For both models, regardless of $\beta$, the parallel trends violation is exactly zero when $Cov(\eione,\eitwo)=1$, which corresponds to the martingale condition (Panels (a) and (b) in Figure \ref{fig:numerical}). For the separable model, parallel trends additionally holds for very specific combinations of $\beta$ and $Cov(\eione,\eitwo)$ (Panel (a) in Figure \ref{fig:numerical}).  These additional instances of parallel trends, however, require a researcher to not only be willing to choose a parametric distribution, but also to rely on very particular combinations of the discount factor $\beta$ and $Cov(\varepsilon_{i1},\varepsilon_{i2})$ (in addition to a specific selection and outcome model). By focusing on parallel trends for all $g\in \mathcal{G}_{\omega}$, we rule out these additional cases and focus on ``robust'' instances of parallel trends.

\begin{figure}[ht]\caption{Numerical Illustration: Example \ref{ex:selection_outcomes} with $\omega_i=(\ai,\eione)$}
\begin{tabular}{cc}
(a) Separable model  &(b) Autoregressive model\\
    \includegraphics[width=8cm]{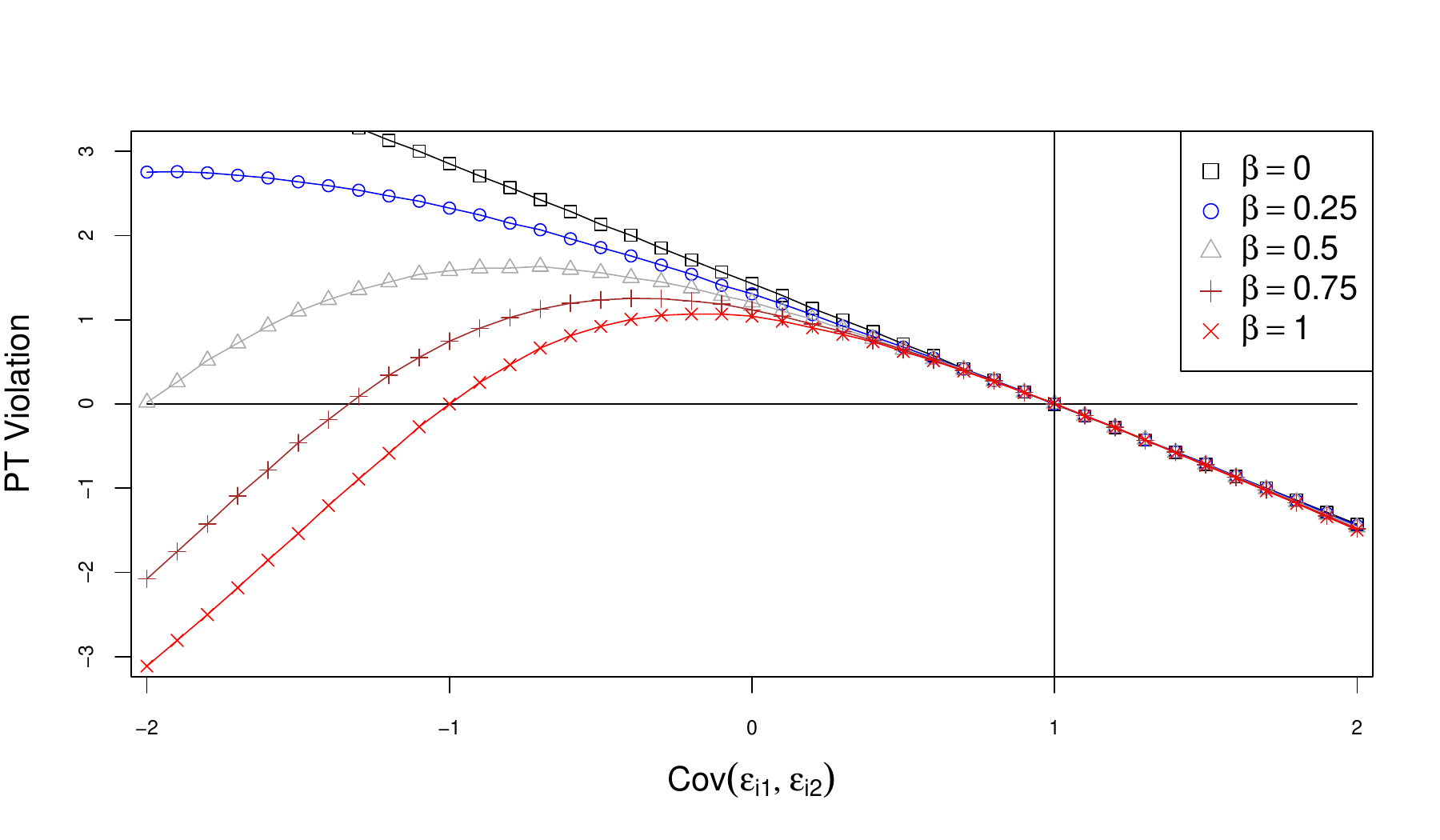}&    \includegraphics[width=8cm]{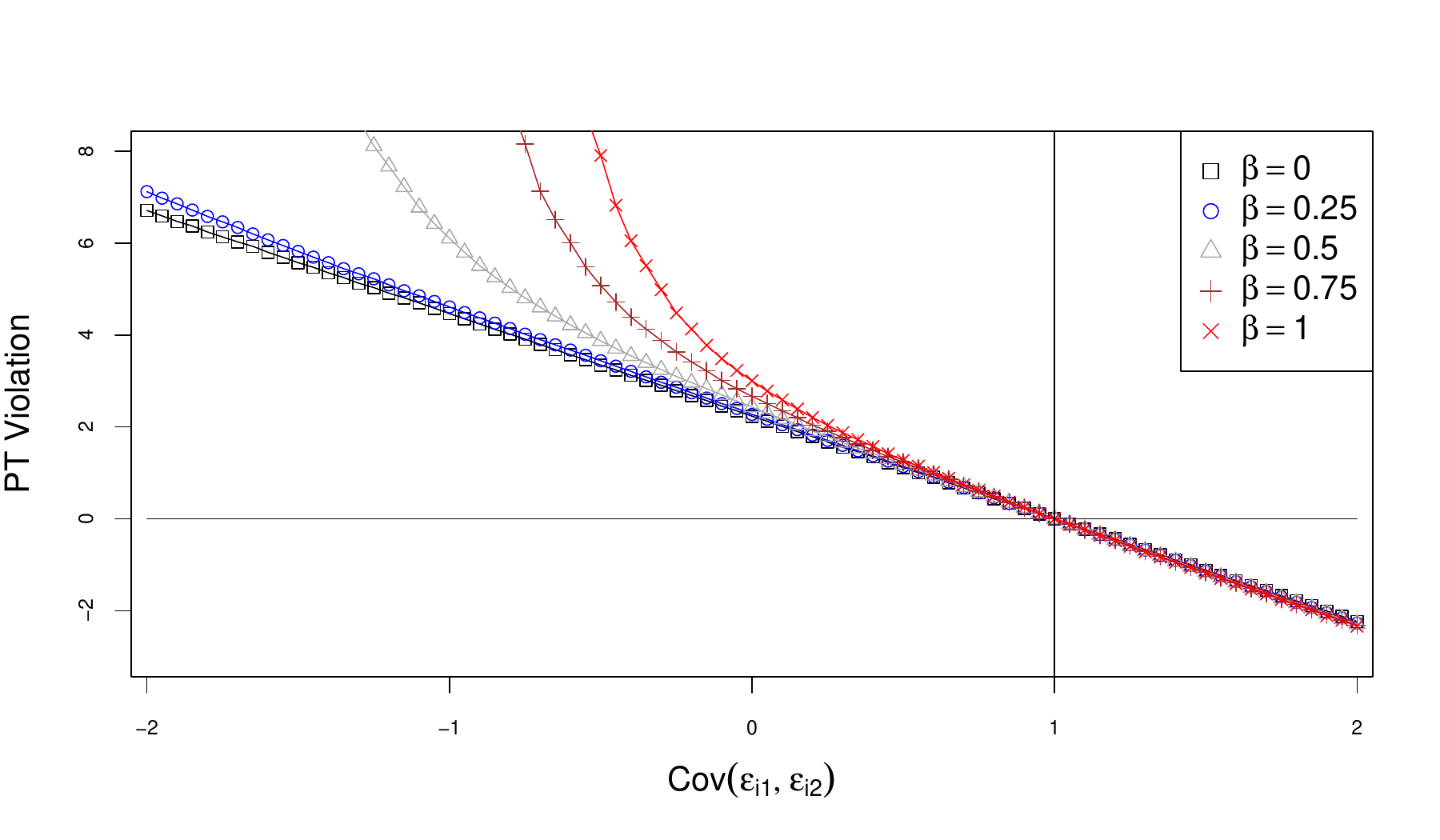}\\
    \end{tabular}\\
\footnotesize{\textit{Notes:} For both (a) and (b), $G_i=1\{Y_{i1}(0)+\beta E[Y_{i2}(0)|\ai,\eione]\le c\alpha_i\}$ for $i=1,\dots, n$. The parallel trends violation (PT violation), $E[Y_{i2}(0)-Y_{i1}(0)|G_i=1]-E[Y_{i2}(0)-Y_{i1}(0)|G_i=0]$, is computed numerically using $n=10^6$. For (a), the untreated potential outcomes are generated as follows: $Y_{it}(0)=\ai+\lambda_t+\eit$ for $t=1,2$, where $\lambda_1=0$, $\ai\sim N(0,1)$, $\ai\indep (\eione,\eitwo)$, and $(\eione,\eitwo)$ are jointly normal. We normalize $Var(\eione)=1$ such that $E[\dot{Y}_{i2}(0)|\ai,\eione]=\ai+Cov(\eione,\eitwo)\eione$, where $Cov(\eione,\eitwo)=\rho_{\varepsilon}\sigma_2$, $\sigma_2^2\equiv Var(\eitwo)$ and $\rho_{\varepsilon}\equiv Corr(\eione,\eitwo)$. The plot in (a) presents the PT violations for  a grid of values of $Cov(\eione,\eitwo)$ with $\lambda_2=\sigma_2=2,~ c=0.5$. For (b), $Y_{i1}(0)=\ai+\eione$, $Y_{i2}(0)=\rho_2\ai+\lambda_2+\eitwo$, where $\ai\indep (\eione,\eitwo)$, $\eitwo=\rho_2\eione+\zeta_{i2}$, and $(\eione,\zeta_{i2})$ are jointly normal. In this model, $E[\dot{Y}_{i2}(0)|\ai,\eione]=\rho_2\dot{Y}_{i1}(0)$. We normalize $Var(\eione)=1$ such that $Cov(\eione,\eitwo)=\rho_2$ and use $\sigma_2^2$ to denote $Var(\zeta_{i2})$. The plot in (b) presents the PT violation for a grid of values of $Cov(\eione,\eitwo)$ with $\lambda_2=\sigma_2=2,~ c=0.25$.}
\label{fig:numerical}
\end{figure}

\subsection{Extensions}
Here, we summarize three main extensions. We refer to the corresponding appendices for details.

\subsubsection{Disaggregate data and aggregate decisions} 
In many DiD applications, the selection decisions are made at the aggregate level (e.g., at the county or state level), while outcome data are available at the disaggregate level (e.g., at the individual or firm level). In Appendix \ref{app:disaggregate}, we consider a sharp DiD design with $n_s$ individuals, indexed by $i=1,\dots,n_s$, belonging to aggregate unit $s$. Let $Y_{st}(0)$ denote the untreated potential outcome for aggregate unit $s$ in period $t$ (e.g., the average of the disaggregate outcomes $Y_{ist}(0)$, $Y_{st}(0)=n_s^{-1}\sum_{i=1}^{n_s}Y_{ist}(0)$) and $G_s$ the selection decision of unit $s$. The necessary and sufficient conditions in this section directly apply to this setting by replacing $i$ with $s$ and interpreting the unobservables and potential outcomes as aggregate quantities.
That said, being explicit about the aggregation can help ``microfound'' restrictions on selection, as we discuss in Appendix \ref{app:disaggregate}.

\subsubsection{Multiple periods and groups}
In Appendix \ref{app: general did}, we extend our results to DiD designs with multiple periods and multiple groups.\footnote{Our setup and notation build on \citet{callaway_SantAnna_2021}, \citet{sun_abraham_2021}, and \citet{roth_et_al_DiD_survey}.} Specifically, we consider a staggered adoption setting with $T$ periods, where no units are treated at $t=1$ and some units remain untreated at $t=T$. We provide an analog of the general necessary and sufficient condition in Theorem \ref{thm:main}, and based on this condition, all our theoretical results can be extended to DiD settings with multiple groups and staggered adoption.

\subsubsection{Covariates}

In many applications, parallel trends may only be plausible conditional on covariates \citep[e.g.,][]{heckman1997matching,abadie2005DiD,santanna2020drdid, callaway_SantAnna_2021}. Therefore, we study the role of covariates through the lens of selection into treatment in Appendix \ref{app:covariates}. We explicitly allow for a vector of both time-invariant and time-varying  covariates, $X_{it}$, assuming that $X_{it}$ is not affected by the treatment. In Appendix \ref{app:separability}, we show that the necessary and sufficient conditions for conditional parallel trends imply separability requirements on how the covariates can enter the outcome model.   Appendix \ref{sec:selection-templates-PTX} provides selection-based templates for justifying conditional parallel trends assumptions for separable models. In Appendix \ref{sec:covariates_nonseparable}, we propose a weaker conditional parallel trends assumption that accommodates a rich class of nonseparable models and provide sufficient conditions for this assumption.

\section{Selection-based bias decomposition with an application to benchmarking the bias relative to imperfect foresight}
\label{sec:bias_characterization}

The necessary and sufficient conditions in Section \ref{sec:necessary_sufficient} demonstrate that if we allow for selection on time-varying shocks and in particular on the determinants of the untreated potential outcomes, parallel trends implies strong restrictions on the time series properties of these outcomes. Here we analyze the bias of DiD when these necessary and sufficient conditions are violated.

The bias analysis accommodates, but does not require, data on additional pre-treatment periods.
Suppose that there is one additional pre-treatment period, $t=0$, in which no units are treated, so that $Y_{i0}=Y_{i0}(0)$ for $i=1,\dots,n$. We allow selection to also depend on the shocks in period $t=0$, that is, we allow $\oi$ to be a function of $(\eizero,\eta_{i0})$.

\subsection{A general bias decomposition}

The following lemma provides a decomposition of the bias of DiD.

\begin{lemma}[Bias decomposition]\label{lem:bias_decomposition} Suppose that $P(G_i=1)\in (0,1)$. Then, for a given $\oi$, the bias of DiD can be decomposed as follows,
$$
\did-\att=\biasselpost+\biasmtgpost,
$$
where 
\begin{eqnarray*}
    \biasselpost&=&\frac{E[(G_i-E[G_i|\oi])(\dot{Y}_{i2}(0)-\dot{Y}_{i1}(0))]}{P(G_i=1)P(G_i=0)},\\
    \biasmtgpost&=&\frac{E[E[G_i|\oi]E[\dot{Y}_{i2}(0)-\dot{Y}_{i1}(0)|\oi]]}{P(G_i=1)P(G_i=0)}.
\end{eqnarray*}
\end{lemma}
\begin{proof}
See Appendix \ref{proof:lem:bias_decomposition}.
\end{proof}

The decomposition in Lemma \ref{lem:bias_decomposition} shows that the bias of DiD is equal to the sum of two components. The component $\biasselpost$ captures the parallel trends violation due to selection on unobservables not contained in $\oi$. The component $\biasmtgpost$ captures the parallel trends violation due to deviations from the necessary and sufficient condition for Assumption \ref{ass:PT} when selection is based on $\oi$ (Theorem \ref{thm:main}).

The bias decomposition in Lemma \ref{lem:bias_decomposition} is specific to a choice of information set $\oi$. One can think of the choice of $\oi$ as an anchor for the bias decomposition. We emphasize that the bias decomposition neither relies on assumptions on selection nor the time series process of $\dot{Y}_{it}(0)$. Indeed, the goal is to characterize  the extent to which the bias depends both on how much the units' selection behavior departs from selection purely on $\oi$ ($\biasselpost$), and on deviations from the corresponding time series restriction ($\biasmtgpost$).

In the following, we demonstrate how to benchmark the bias components in empirical applications. To do so, we need to choose an empirically plausible information set $\oi$ to anchor the decomposition. Selection on pre-treatment unobservables (imperfect foresight) is prevalent in DiD applications, including those in Section \ref{sec:empirical_illustration}. Therefore, it provides an empirically compelling anchor for the decomposition. Moreover, our analysis provides a template for tailoring the bias benchmarking to other information sets.

\subsection{Application to imperfect foresight}\label{sec:application_if}

Suppose that researchers deem selection on pre-treatment unobservables likely. In this case, there are two potential sources of bias: selection on post-treatment unobservables and violations of the martingale condition. In the following, we characterize these two bias terms and provide stratgies for benchmarking them using pre-treatment data.

To simplify the notation, define $\varepsilon_i^{t}\equiv (\varepsilon_{i0},\dots,\varepsilon_{it})$ and $\eta_i^{t}\equiv(\eta_{i0},\dots,\eta_{it})$ for $t>0$. Suppose that $\oi=(\ai,\eipre,\mu_i,\etaipre)$. In this case, the necessary and sufficient condition in Corollary \ref{cor:if1} generalizes to
\begin{equation}
E[\dot{Y}_{i2}(0)|\ai,\eipre,\mu_i,\etaipre]=\dot{Y}_{i1}(0).\label{eq:mtg_with_pre_period}
\end{equation}
To aid interpretation and assess the magnitude of the bias components in Lemma \ref{lem:bias_decomposition}, we characterize them under the following linear relaxation of the martingale condition \eqref{eq:mtg_with_pre_period}.\footnote{We focus on linear relaxations for convenience. Extensions to nonparametric relaxations of the form $
E[\dot{Y}_{it}(0)|\ai,\eizero,\dots, \eitminusone]=\sigma_{t}\rho(\dot{Y}_{i(t-1)}(0)),$
where $\rho(\cdot)$ is an arbitrary nonparametric function and $\sigma_1$ is normalized to one, are straightforward.} 
\begin{customass}{REL}\label{ass:REL} The following relaxation of the martingale condition holds:\footnote{Assumption \ref{ass:REL} yields a linear autoregressive model. This class of models has been studied extensively in the time series literature under restrictions on the heterogeneity of the coefficient \citep[e.g.,][]{nicholls1982random,regis2022random}.}
$$
E[\dot{Y}_{it}(0)|\ai, \varepsilon_i^{t-1},\mu_i,\eta_i^{t-1}]=\rho_t\dot{Y}_{i(t-1)}(0), \quad i=1,\dots,n,\quad t=1,2
$$ 
\end{customass}
Assumption \ref{ass:REL} imposes an AR(1) model with time-varying coefficients on $\dot{Y}_{it}(0)$,
\begin{equation}
\dot{Y}_{it}(0)=\rho_t\dot{Y}_{i(t-1)}(0)+\zeta_{it},\quad E[\zeta_{it}|\ai,\varepsilon_i^{t-1},\mu_i,\eta_i^{t-1}]=0.\label{eq:zeta}
\end{equation}
Note that Assumption \ref{ass:REL} is imposed on the demeaned potential outcomes and thus allows for mean shifts in $Y_{it}(0)$. If $\rho_2=1$, Assumption \ref{ass:REL} reduces to the martingale assumption, $E[\dot{Y}_{i2}(0)|\ai,\eipre,\mu_i,\etaipre]=\dot{Y}_{i1}(0)$. As a result, deviations from this martingale property under Assumption \ref{ass:REL} are fully characterized by the deviation of $\rho_2$ from $1$, $(\rho_2-1)$. 

The following proposition characterizes the bias components under Assumption \ref{ass:REL}.
\begin{proposition}[Bias characterization under linear martingale relaxation]\label{prop:bias_characterization}
    Suppose that $P(G_i=1)\in(0,1)$, $\oi=(\ai,\eipre,\mu_i,\etaipre)$, and Assumption \ref{ass:REL} holds. Then, 
\begin{eqnarray*}
    \biasselpost &=& E[\zeta_{i2}|G_i=1]-E[\zeta_{i2}|G_i=0],\\
    \biasmtgpost &=& (\rho_2-1)(E[Y_{i1}|G_i=1]-E[Y_{i1}|G_i=0]).
\end{eqnarray*}
\end{proposition}
\begin{proof}
See Appendix \ref{proof:prop:bias_characterization}.
\end{proof}

Proposition \ref{prop:bias_characterization} shows that $\biasselpost$ is equal to the mean difference in $\zeta_{i2}$ between both groups. Recall that $\zeta_{i2}$ is the difference between $\dot{Y}_{i2}(0)$ and its conditional expectation given pre-treatment unobservables $E[\dot{Y}_{i2}(0)|\ai,\eipre,\mu_i,\etaipre]$. If selection depends on post-treatment unobservables including $\eitwo$, then $\zeta_{i2}$ is  correlated with selection $G_i$, so that $E[\zeta_{i2}|G_i=1]$ is not equal to $E[\zeta_{i2}|G_i=0]$.

Proposition \ref{prop:bias_characterization} further shows that $\biasmtgpost$ is equal to the product of the martingale deviation, $(\rho_2-1)$, and the observed pre-treatment difference, $E[Y_{i1}|G_i=1]-E[Y_{i1}|G_i=0]$. This shows that the sensitivity of DiD with respect to violations of the martingale assumption depends on the pre-treatment group difference. It underscores that only if the martingale property holds exactly can we ignore pre-treatment differences (and the selection on unobservables they are indicative of). This discussion motivates using covariate adjustment for reducing the pre-treatment difference and thus the potential bias of DiD due to violations of the martingale assumption. We illustrate this point in Section \ref{sec:empirical_illustration} and defer the formal analysis of the DiD bias decomposition with covariates to Appendix \ref{app: bias decomposition covariates}. If the treatment is randomly assigned, then $E[Y_{i1}|G_i=1]-E[Y_{i1}|G_i=0]=0$ and $\biasmtgpost=0$. 

An important takeaway from Proposition \ref{prop:bias_characterization} is that the bias of DiD is an affine function of the martingale deviation $(\rho_2-1)$, where the slope is the pre-treatment difference, $E[Y_{i1}|G_i=1]-E[Y_{i1}|G_i=0]$, and the intercept is the post-treatment difference, $E[\zeta_{i2}|G_i=1]-E[\zeta_{i2}|G_i=0]$. The only component that is directly observable from the data is the pre-treatment difference. We next demonstrate how we can benchmark $\rho_2$ and $E[\zeta_{i2}|G_i=1]-E[\zeta_{i2}|G_i=0]$ using pre-treatment data.

\subsection{Benchmarking DiD bias components relative to imperfect foresight}\label{sec:benchmarking}
Here, we provide benchmarks for $\biasselpost$ and $\biasmtgpost$ based on pre-treatment data that allow practitioners to assess and sign the bias of DiD in applications.

The unobservable component in $\biasmtgpost$  is $\rho_2$ for which there are at least two natural benchmarks. First, under Assumption \ref{ass:REL}, $\rho_1$ can be identified from the pre-treatment data by noting that $E[\dot{Y}_{i1}(0)|\dot{Y}_{i0}(0)]=E[E[\dot{Y}_{i1}(0)|\ai,\eizero,\mu_i,\eta_{i0}]|\dot{Y}_{i0}(0)]=\rho_1\dot{Y}_{i0}(0)$, such that $\rho_1$ is identified as the coefficient of a population regression of $\dot{Y}_{i1}$ on $\dot{Y}_{i0}$. This is a useful benchmark if it is plausible that $\rho_2=\rho_1$, that is, if $\rho_t$ in Assumption \ref{ass:REL} is time-homogeneous. Second, we can consider the persistence in the control group, $E[\tilde{Y}_{i2}(0)|\tilde{Y}_{i1}(0),G_i=0]=\rho^{0}_2\tilde{Y}_{i1}(0)$, where $\tilde{Y}_{it}(0)\equiv Y_{it}(0)-E[Y_{it}(0)|G_i=0]$, and use $\rho^{0}_2$ as a benchmark value for $\rho_2$. This is useful in settings where unconfoundedness is plausible since $\rho_2=\rho_2^0$ in this case.\footnote{Specifically, the unconfoundedness assumption $Y_{i2}(0)\indep G_i|Y_{i1}(0)$ implies that $\rho_2=\rho^{0}_2$. This follows because $Y_{i2}(0)\indep G_i|Y_{i1}(0)$ implies that $E[\dot{Y}_{i2}(0)|\dot{Y}_{i1}(0),G_i=0]=E[\dot{Y}_{i2}(0)|\dot{Y}_{i1}(0)]$.} We emphasize that the proposed benchmarking strategy allows researchers to consider a range of values for $\rho_2$, including $\rho_1$ and $\rho_2^0$. 

As for $\biasselpost$, it is helpful to consider its observable pre-treatment analogue,
$$\biasselpre\equiv E[\zeta_{i1}|G_i=1]-E[\zeta_{i1}|G_i=0], \quad \text{where}~~\zeta_{i1}=\dot{Y}_{i1}-\rho_1\dot{Y}_{i0}.$$
 To relate $\biasselpost$ and $\biasselpre$, note that (assuming $\biasselpre\neq 0$)
$$\biasselpost= \frac{\rho_{G,\zeta_2}\sigma_{\zeta_2}}{\rho_{G,\zeta_1}\sigma_{\zeta_1}}\biasselpre,$$
where $\rho_{G,\zeta_t}\equiv Corr(G_i,\zeta_{it})$ and $\sigma_{\zeta_t}^2\equiv Var(\zeta_{it})$.

In the case where $\sigma_{\zeta_1}=\sigma_{\zeta_2}$, the relative magnitude of $\biasselpost$ to $\biasselpre$  is simply the ratio of the (scale-free) correlation coefficients, $\rho_{G,\zeta_2}/\rho_{G,\zeta_1}$. This ratio measures the relative degree of selection on pre- vs. post-treatment unobservables. For example, if contextual knowledge suggests that there is ``more selection'' on pre-treatment than on post-treatment unobservables, then $|\biasselpost|\le |\biasselpre|$, assuming $\operatorname{sgn}\left(\rho_{G,\zeta_1}\right)=\operatorname{sgn}\left(\rho_{G,\zeta_2}\right)$. The edge case where $\biasselpost=\biasselpre$ captures settings where the extent of selection on post- and pre-treatment unobservables is the same. 

In practice, we recommend that researchers determine the robustness of DiD results based on the benchmarks we discuss above. We illustrate this approach in two empirical applications in Section \ref{sec:empirical_illustration}.

\section{Empirical illustration of bias decomposition}
\label{sec:empirical_illustration}

Here we illustrate the bias decomposition in two applications. In the first application, we revisit the NSW training program, where we have access to an experimental estimate of the ATT and thus an estimate of the bias of DiD. This bias estimate allows us to directly evaluate the bias decomposition and benchmarking strategy using a \citet{lalonde1986evaluating}-style exercise. In the second application, we consider the Medicaid expansion to demonstrate the usefulness of the bias decomposition to DiD applications with aggregate selection. 
Selection on pre-treatment outcomes (and pre-treatment information more generally) is likely in both empirical contexts, as discussed in Example \ref{ex:selection_outcomes}, which makes them well-suited for illustrating the bias decomposition relative to selection on pre-treatment information. 

Both applications demonstrate that pre-treatment differences in mean outcomes between the treatment and control groups matter. While we can ignore such differences under parallel trends, our analysis underscores the central role they play once we entertain the possibility of violations of parallel trends. Both applications further highlight that covariates are crucial for reducing the pre-treatment difference in mean outcomes and thereby rendering DiD less sensitive to martingale violations.

\subsection{NSW training program}

\subsubsection{Setup and DiD analysis}
The evaluation of job training programs is one of the classical applications of DiD in economics. Here, we revisit the analysis of the causal effect of the NSW training programs on post-treatment earnings \citep[e.g.,][]{lalonde1986evaluating}. We use the same dataset as \citet{santanna2020drdid} and consider the ``\citet{Dehejia1999,Dehejia2002} sample.''\footnote{The data are from the \texttt{DRDID} \texttt{R}-package \citep{santanna_zhao_DRDID}.} 
This sample combines the experimental treatment group (185 individuals)  with an observational control group (15,992 individuals).

The outcome of interest is earnings. We observe individual-level data on earnings for two pre-treatment periods, 1974 and 1975, and one post-treatment period, 1978. We also have access to a set of baseline covariates: age, years of education, and indicators for high school dropouts, married individuals, Black and Hispanic individuals. 
 
The unconditional DiD estimate using 1975 as the pre-treatment period ($t=1$) and 1978 as the post-treatment period ($t=2$) is equal to $\widehat{\did}=$~3,621 (s.e.\ 610). A comparison to the experimental benchmark, which is 1,794 (s.e.\ 671), shows that the unconditional DiD substantially overestimates the returns to the training program. The estimated bias relative to the experimental benchmark is statistically and economically significant at 1,827, comparable in magnitude to the experimental benchmark.

With covariates, the regression-adjusted DiD estimate under conditional parallel trends is equal to $E_n[\widehat{\did}(X_i)|G_i=1]=$~2,436 (s.e.\ 653), where $E_n$ denotes the sample average and $\widehat{\did}(X_i)$ is the conditional DiD estimate obtained using the regression-adjusted DiD estimator. 
This shows that adjusting for differences in baseline covariates reduces the bias of DiD to 642, about a third of the bias of the unconditional DiD relative to the experimental benchmark.

It is standard to report the results from pre-trends tests when there are additional pre-treatment periods. Based on the pre-treatment data from 1974 and 1975, the unconditional and regression-adjusted DiD estimates are 198 (s.e.\ 280) and 335 (s.e.\ 309), respectively. 

Despite the non-rejections of the pre-trends tests, the sensitivity of the DiD estimates to parallel trends violations remains a major concern for three reasons. First, pre-tests are, by construction, not direct tests of parallel trends assumptions. Second, these tests can be substantially underpowered \citep[e.g.,][]{roth_pre-test_2022}.  Finally, building on the necessary and sufficient conditions in Section \ref{sec:necessary_sufficient}, \citet{ghanem2026when} show that pre-trends can be uninformative under imperfect foresight. These issues are particularly evident in this application, where the pre-test does not reject, despite unconditional DiD being significantly biased relative to the experimental benchmark. Next, we demonstrate how our selection-based bias decomposition can help us better understand the difference between the DiD and the experimental estimate in this application.

\subsubsection{Decomposing the bias of DiD}\label{sec:benchmarking_nsw}

We start by illustrating the bias decomposition without covariates.  
Replacing the population expectations by sample averages, we obtain
\begin{eqnarray*}
\biasdidposthat=\biasdidposthat\Big(\biasselpost,\rho_2\Big)&=&\biasselpost+(\rho_2-1)(E_n[Y_{i1}|G_i=1]-E_n[Y_{i1}|G_i=0]),\\
&=&\biasselpost+(\rho_2-1)(-\text{12,119}).
\end{eqnarray*} 
There is a substantial pre-treatment difference: average earnings in 1975 are much lower in the treatment than in the control group.

Figure \ref{fig:sensitivity_wo_covariates_nsw} displays $\biasdidposthat$ as a function of $\rho_2$ together with the bias estimate based on the experimental benchmark. Suppose first that $\biasselpost=0$. In this case, the bias of DiD equals $\biasdidposthat=(\rho_2-1)(-\text{12,119})$, depicted by the blue line. It is solely driven by violations of the martingale property (i.e., differences between $\rho_2$ and $1$). Alternatively, consider the edge case where $\biasselpost=\biasselpre$. This corresponds to the case with equal sign and strength of selection on $\zeta_{i1}$ and $\zeta_{i2}$, as discussed in Section \ref{sec:benchmarking}. The sample analogue of $\biasselpre$ equals $-2,049$, resulting in the following bias estimate (red line in Figure \ref{fig:sensitivity_wo_covariates_nsw}), $$
\biasdidposthat=-2,049+(\rho_2-1)(-\text{12,119}).
$$ 

\begin{figure}[!ht]
\caption{Benchmarking the Bias of DiD: Application to NSW Training Program}
\begin{subfigure}[b]{\textwidth}
         \centering
        \caption{Without covariates}
          \includegraphics[width=0.8\textwidth, keepaspectratio]{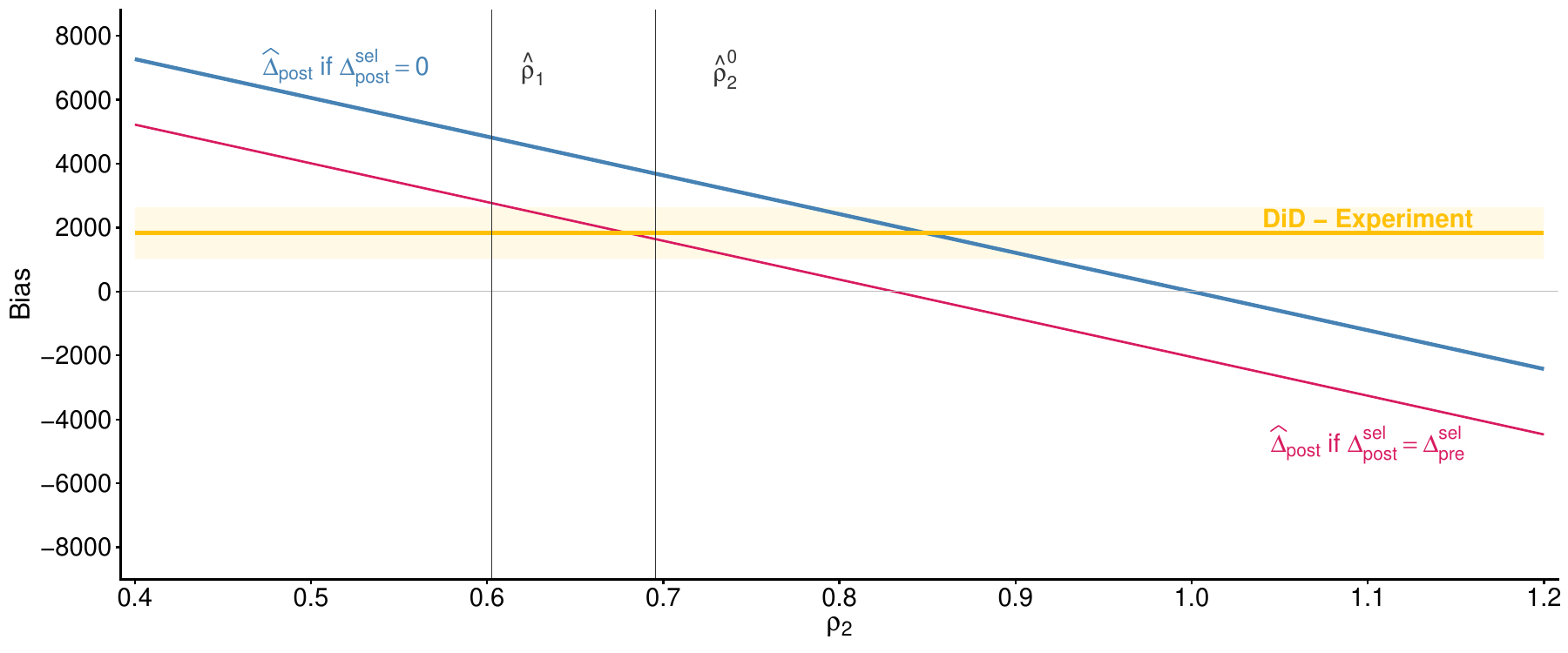}

         \label{fig:sensitivity_wo_covariates_nsw}
     \end{subfigure}
     
     \begin{subfigure}[b]{\textwidth}
         \centering
        \caption{With covariates}
         \includegraphics[width=0.8\textwidth, keepaspectratio]{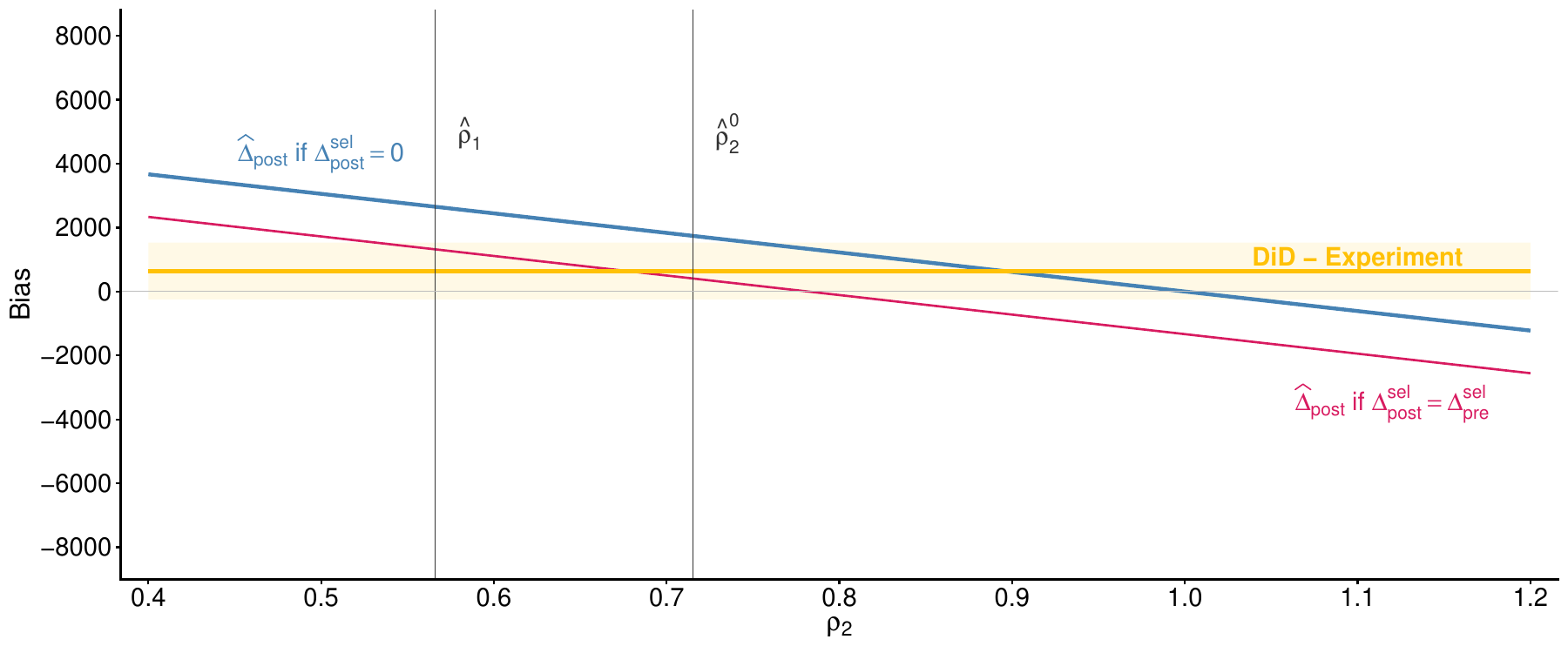}

         \label{fig:sensitivity_with_covariates_nsw}
     \end{subfigure}
     
\footnotesize{\textit{Notes:} Figure \ref{fig:sensitivity_wo_covariates_nsw} displays the results from the bias decomposition without covariates. Figure \ref{fig:sensitivity_with_covariates_nsw} shows the results from the bias decomposition with regression adjustment using age, years of education, indicators for high school dropouts, married individuals, Black and Hispanic individuals, age squared, age cubed (divided by 1,000), and years of schooling squared. The shaded areas depict 95\% confidence intervals for the difference between the unconditional DiD (regression-adjusted DiD in Figure \ref{fig:sensitivity_with_covariates_nsw}) and the experimental estimates using Bayesian-bootstrapped standard errors clustered at the individual level, based on 10,000 bootstrap draws.  Data: \citet{santanna_zhao_DRDID}.}
     
\end{figure}

As we discussed in Section \ref{sec:benchmarking}, there are two natural benchmarks for $\rho_2$: $\rho_1$, the pre-treatment counterpart of $\rho_2$, and $\rho_2^0$, its control group counterpart. The corresponding estimates are $\widehat\rho_1=0.603$ and $\widehat\rho^0_2=0.695$ and are depicted in Figure \ref{fig:sensitivity_wo_covariates_nsw}.\footnote{Recall that the post-treatment earnings are measured in 1978, so that $\rho_2$ measures the persistence over three years. To account for the difference in periodicity when estimating $\rho_1$, we proceed in two steps. First, we regress $\dot{Y}_{i1975}$ on $\dot{Y}_{i1974}$ to obtain an estimate of the yearly persistence in the pre-treatment period, $\tilde{\rho}_1=0.845$. Second, we adjust for the difference in periodicity by computing $\widehat\rho_1$ as $\widehat\rho_1=(\widehat{\tilde{\rho}}_1)^3=0.603$. This is justified under a linear AR(1) model for the demeaned outcomes in the pre-treatment period.} Both benchmark values would suggest that the unconditional DiD is upwardly biased, consistent with the  experimental bias estimate.

The analysis without covariates demonstrates that the bias of DiD is very sensitive to deviations from the martingale property.
The lack of robustness is driven by the treatment and control groups being very different before the treatment. This discussion suggests that we may reduce the pre-treatment difference and improve the robustness of DiD by adjusting for differences in baseline covariates. 

We therefore incorporate covariates into our analysis in Figure \ref{fig:sensitivity_with_covariates_nsw}. In Appendix \ref{app: bias decomposition covariates}, we show that under a linear relaxation of the conditional martingale property, the unconditional bias of DiD with covariates can be decomposed as 
\begin{eqnarray*}
\biasdidpost&=&\biasselpost+(\rho_2-1)(E[Y_{i1}|G_i=1]-E[E[Y_{i1}|G_i=0,X_i]|G_i=1]),
\end{eqnarray*}
where $\biasselpost\equiv E[\biasselpost(X_i)|G_i=1]$. 

Analogous to the unconditional bias decomposition, consider first the case where $\biasselpost=0$. Using the regression-adjusted estimator for the pre-treatment difference described in Appendix \ref{app:implementation}, we obtain the following bias estimate (blue line in Figure \ref{fig:sensitivity_with_covariates_nsw}),
$
\biasdidposthat=(\rho_2-1)(-6,113).
$
Adjusting for differences in baseline covariates reduces the magnitude of the pre-treatment difference by approximately 50\%. As a result, incorporating covariates makes the bias of DiD less sensitive to violations of the martingale property.

Alternatively, consider the case where $\biasselpre=\biasselpost$, which is implied by $\biasselpre(X_i)=\biasselpost(X_i)$. Using the regression-adjusted estimator of $\biasselpre$ described in Appendix \ref{app:implementation}, we obtain (red line in Figure \ref{fig:sensitivity_with_covariates_nsw})
$$
\biasdidposthat=-1,333+(\rho_2-1)(-6,113).
$$
The estimates of $\rho_1$ and $\rho_2^0$ with covariates are $\widehat\rho_1=0.566$, which is somewhat smaller than without covariates, and $\widehat\rho_2^0=0.715$, which is somewhat larger than without covariates.\footnote{Under the linear relaxation of the martingale assumption, the yearly persistence in the pre-treatment period, $\tilde\rho_1$, can be estimated by regressing the regression-adjusted pre-treatment outcome on its lag (see Appendix \ref{app: bias decomposition covariates} for details). The resulting estimate is $\widehat{\tilde{\rho}}_1=0.827$. Adjusting for the difference in periodicity yields $\widehat\rho_1=(\hat{\tilde{\rho}}_1)^3=0.566$.} Both benchmark values suggest the same sign and a similar magnitude of the bias as the experimental benchmark.

This analysis demonstrates how the proposed bias decomposition can help empirical practitioners assess the bias of DiD and its sensitivity. This is especially important in applications such as this one, where the (unconditional) pre-trends tests do not reject, even though DiD is biased relative to the experimental benchmark. 

\subsection{Medicaid Expansion}
\subsubsection{Setup and DiD analysis}
We revisit the DiD evaluation of Medicaid expansion to illustrate the relevance of our selection-based bias decomposition to DiD settings with aggregate selection. We use the sample from the 2$\times$2 DiD implementation in \citet{baker2026difference}, but consider one additional pre-treatment period.\footnote{The data are available in the following GitHub repository \url{https://github.com/pedrohcgs/JEL-DiD}.} The treatment group consists of states that have expanded Medicaid in 2014, whereas the control group consists of states that have not expanded by 2019. The pre-treatment periods are 2012 and 2013 ($t=0,1$), and the post-treatment period is 2014 ($t=2$).

In the context of Medicaid expansion, the outcome of interest, observed at the county level, is the crude mortality rate for people aged 20-64 (measured per 100,000). In our conditional DiD analysis, we also include the percentages of a county’s population that are female, white, or Hispanic; the unemployment rate; the poverty rate; and county-level median income (in thousands of dollars)---all measured in 2012---in our regression adjustment.\footnote{We use covariate values from 2012, so we can treat them as time-invariant in our analysis, simplifying the exposition and avoiding the strong, possibly unrealistic assumption that our covariates are strictly exogenous in this application. In line with \citet{callaway_SantAnna_2021}'s implementation in the \texttt{did R} package, \citet{baker2026difference} fixed covariate values at the 2013 values for post-treatment analysis, and at the 2012 values for pre-treatment periods, 2012-2013. 
We refer the reader to \citet{caetano2022difference}, as well as to \citet{ghanem2026when}, for additional discussion.} All estimates are weighted by county population in 2013.

We first examine the unconditional DiD estimate using 2013 and 2014, which equals $-2.6$ (s.e. 1.5), indicating a reduction in mortality due to Medicaid expansion that is statistically significant at the $10\%$ level. Once we account for covariates, however, the results are no longer significant with a regression-adjusted DiD estimate of $-2.1$ (s.e. 2.2). 

Before we proceed to the bias decomposition, we conduct the pre-trends tests. We find that unconditional and conditional pre-trends tests are not rejected at the 5\% level, with differences in pre-trends of $-2.8$ ($1.5$) and $-2.6$ (s.e. $2.5$), respectively. However, note that the pre-trends for the unconditional DiD are significant at the $10\%$ level.

\subsubsection{Decomposing the bias of DiD}
We next present the sample analogues of the bias decomposition, as described in Section \ref{sec:benchmarking_nsw}. For the unconditional DiD case, the sample analogue of the bias can be decomposed as follows
\begin{eqnarray*}
\biasdidposthat=\biasdidposthat\Big(\biasselpost,\rho_2\Big)&=&\biasselpost+(\rho_2-1)(E_n[Y_{i1}|G_i=1]-E_n[Y_{i1}|G_i=0]),\\
&=&\biasselpost+(\rho_2-1)(-\text{53.7}).
\end{eqnarray*} 
where $-53.7$ denotes the pre-treatment difference in means between the treatment and control group, statistically significant at the 1\% level.   

Figure \ref{fig:sensitivity_wo_covariates_medicaid} demonstrates that this substantial pre-treatment difference translates to  the bias of DiD being very sensitive to martingale violations. When considering the benchmark values for $\rho_2$, however, we note that both $\hat{\rho}_1$ and $\hat{\rho}_2^0$ are fairly close to 1, and therefore, the bias of DiD due to martingale deviations is relatively small for those values. If one is willing to assume that $\biasselpost=0$, then our analysis suggests a positive bias of DiD for these benchmark values of $\rho_2$ (Figure \ref{fig:sensitivity_wo_covariates_medicaid}). If we instead assume that $\biasselpost=\biasselpre$, then our analysis indicates a negative bias of DiD for these benchmark values.

Once we adjust for covariates, the pre-treatment difference is no longer significant. Indeed, it is a negligible difference yielding the following sample analogue of the bias of the regression-adjusted DiD 
\begin{eqnarray*}
\biasdidposthat=\biasdidposthat\Big(\biasselpost,\rho_2\Big)&=&\biasselpost+(\rho_2-1)(-3.5).
\end{eqnarray*} 
As a result, the bias is insensitive to violations of the martingale condition and mostly driven by the magnitude of $\biasselpost$, as illustrated in Figure \ref{fig:sensitivity_with_covariates_medicaid}. If $\biasselpost=0$, then our analysis suggests that the bias of DiD is negligible, whereas if $\biasselpost=\biasselpre$, then the bias is negative.

The bias component $\biasselpost$ captures the bias due to selection on post-treatment unobservables. Selection on post-treatment unobservables is unlikely if the time-varying determinants of mortality are difficult to predict. This is the case, for example, if mortality is determined by factors such as adverse weather shocks and disease outbreaks which are arguably difficult to perfectly foresee,  even one year ahead. In this case, a researcher could argue that $\biasselpost$ is close to zero, which implies that the bias of DiD is negligible.

\begin{figure}[!ht]
\caption{Benchmarking the Bias of DiD: Application to Medicaid Expansion}
\begin{subfigure}[b]{\textwidth}
         \centering
          \includegraphics[width=0.8\textwidth, keepaspectratio]{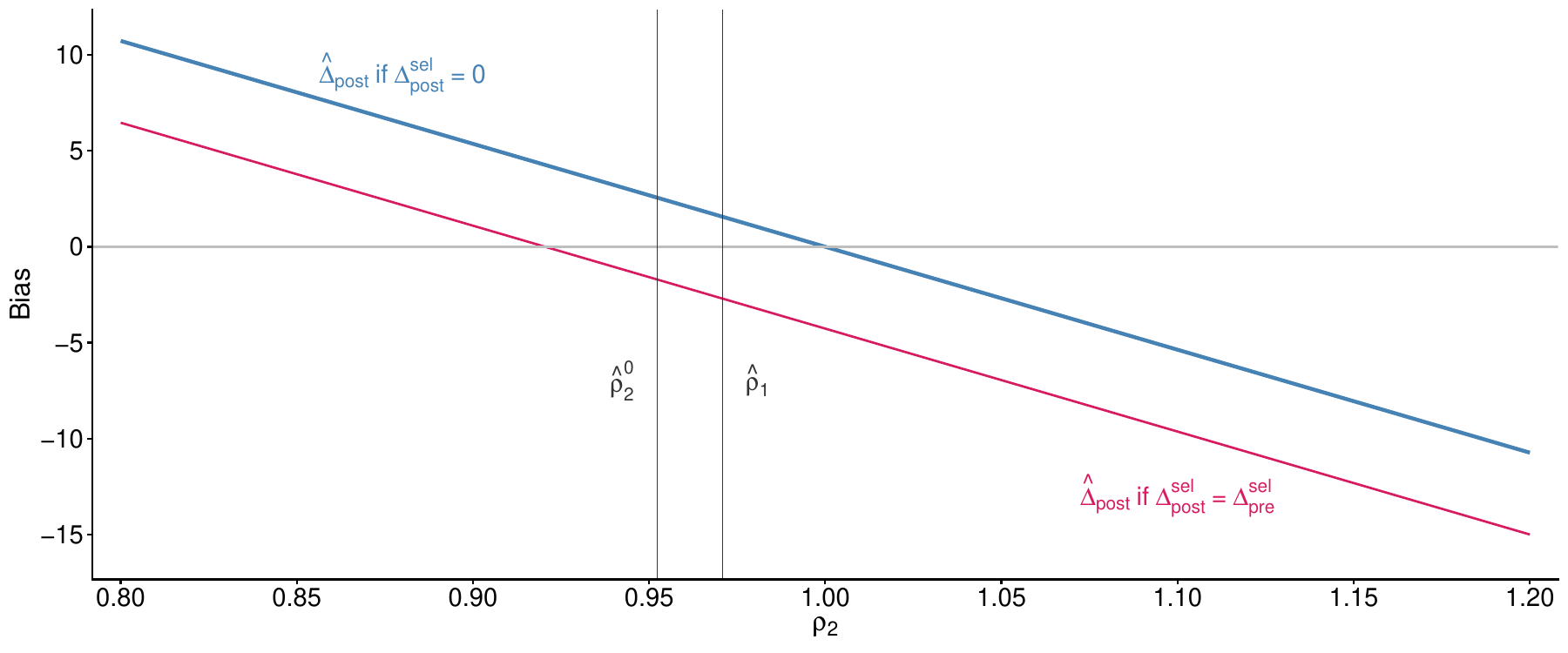}
         \caption{Without covariates}
         \label{fig:sensitivity_wo_covariates_medicaid}
     \end{subfigure}
     
     \begin{subfigure}[b]{\textwidth}
         \centering
         \includegraphics[width=0.8\textwidth, keepaspectratio]{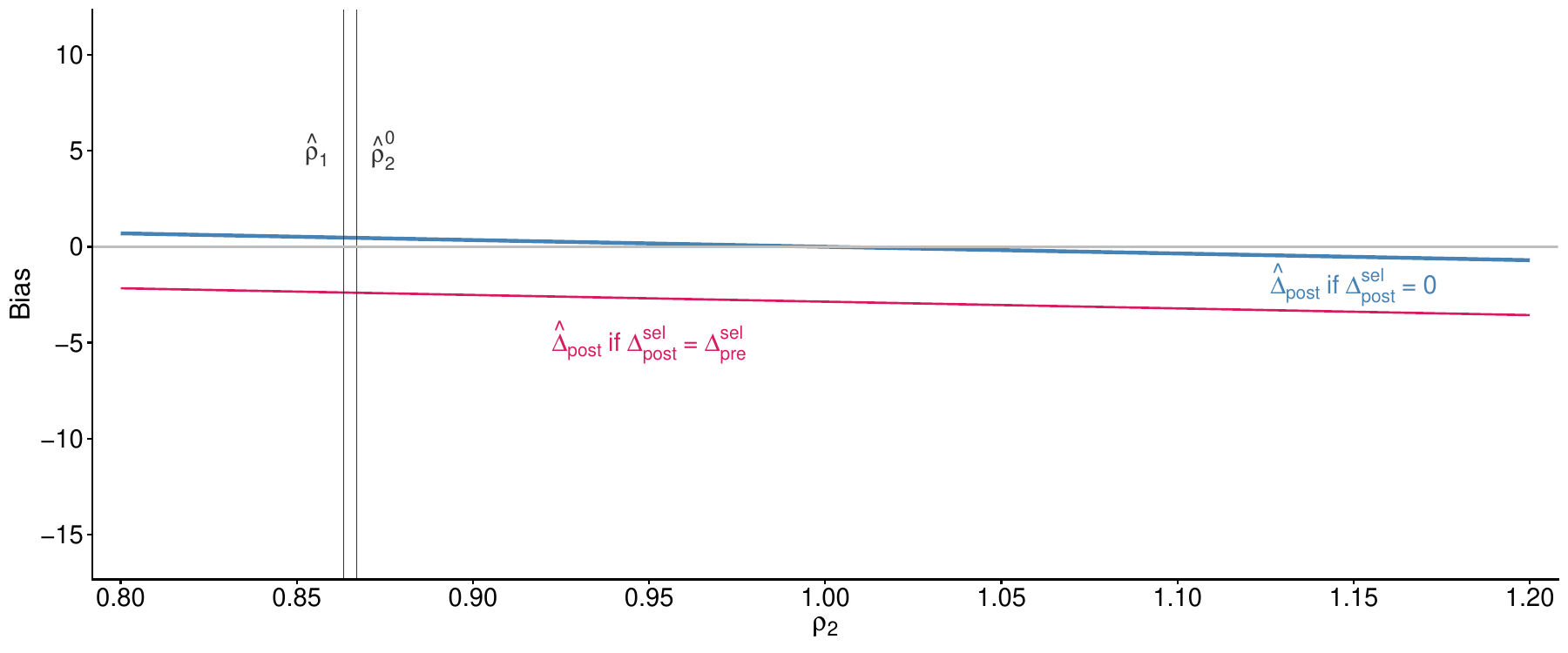}
         \caption{With covariates}
         \label{fig:sensitivity_with_covariates_medicaid}
     \end{subfigure}
     \footnotesize{\textit{Notes:} 
     Figure \ref{fig:sensitivity_wo_covariates_medicaid} displays the results from the bias decomposition without covariates. Figure \ref{fig:sensitivity_with_covariates_medicaid} shows the results from the bias decomposition with regression adjustment using the percentages of a county’s population that are female, white, or Hispanic, the unemployment rate, the poverty rate, and county-level median income (in thousands of dollars)---all measured in 2012. Data: \citet{baker2026difference}.}
\label{fig:medicaid}
     \end{figure}

\section{Implications for empirical practice}\label{sec:implications_PT}

In this paper, we study parallel trends assumptions through the lens of selection into treatment. We derive necessary and sufficient conditions that clarify the empirical content of parallel trends, shed light on the trade-offs between assumptions on selection and time series restrictions, motivate DiD bias decompositions and benchmarking strategies, and provide theory-based templates for assessing and justifying parallel trends in applications with and without covariates. Below, we summarize the main implications of our results for practitioners.

\smallskip

\noindent\textbf{Restrictions on selection are unavoidable in DiD designs.} The necessary and sufficient condition in Corollary \ref{cor:no_restriction} underscores that if researchers are not willing to impose any restrictions on selection, then parallel trends is equivalent to the untreated potential outcomes being constant over time up to deterministic mean shifts. Therefore, in realistic settings, relying on parallel trends assumptions implicitly imposes restrictions on the time-varying unobservables and how selection depends on them. 

\smallskip
\noindent \textbf{Contextual and economic knowledge about selection can be used to assess and justify parallel trends.} Our analysis provides a general approach to derive necessary and sufficient conditions for parallel trends with and without covariates. Importantly, these conditions do not require the researchers to specify explicit selection mechanisms, which may be difficult in practice. Instead, the researchers only need to specify what the units select on. When doing so, it is crucial for researchers to consider the periodicity of the data, the timing of the selection decision, the information set available to the units, who make the selection decision (e.g., the individuals themselves or caseworkers in the training program example), as well as at which level the selection decision is made (e.g., at the level of an individual economic agent vs. at the aggregate level via the political process).\footnote{The importance of the information available to units is underscored by the results in \citet{marx2024parallel}, who study specific economic models of selection including learning and optimal stopping.} Another practical byproduct of our analysis is a menu of selection-based templates for assessing and justifying parallel trends with and without covariates, see Section \ref{sec:other} and Appendix \ref{app:covariates}.

\smallskip

\noindent \textbf{Selection-based bias decompositions are useful to sign and benchmark the bias of DiD.} In Section \ref{sec:bias_characterization}, we provide a general selection-based decomposition of the bias of DiD. We then apply this decomposition to benchmark the bias relative to imperfect foresight. Exploiting a linear martingale relaxation, we show that the bias of DiD can be decomposed into two components: (i) the bias due to selection on post-treatment unobservables, (ii) the bias due to deviations from the martingale property necessary and sufficient for parallel trends under imperfect foresight. This characterization can be used in practice to sign and benchmark these two bias components, as we demonstrate in Section \ref{sec:empirical_illustration}. A practical implication of this characterization is that the pre-treatment difference between the treatment and control group is a key determinant of the bias of DiD when parallel trends is violated.

\smallskip

\noindent \textbf{Pre-trends tests are insufficient for establishing parallel trends.\footnote{We thank an anonymous referee for encouraging us to emphasize this implication of our results.}} The theoretical results in Section \ref{sec:necessary_sufficient} show that the validity of parallel trends depends on whether specific combinations of selection behavior and time series properties are present. Comparing trends between the treated and control units during the pre-treatment period using pre-trends tests is generally insufficient to establish that. Building on the results in this paper, we formally study when pre-trends tests can be informative through the lens of selection into treatment in \citet{ghanem2026when}.

\setlength{\bibsep}{0pt}
\bibliographystyle{apalike}
\bibliography{bibtex}

\newpage

\appendix

\setcounter{page}{1} 

\begin{center}
    \huge{Appendix (for online publication)}
    
\end{center}

\startcontents[sections]
\printcontents[sections]{l}{1}{\setcounter{tocdepth}{2}}

\section{Disaggregate data and aggregate decisions}
\label{app:disaggregate}
In some DiD applications, the data are available at the disaggregate level (e.g., at the individual or firm level), while the decision to select into the treatment is made at the aggregate level (e.g., at the county or state level). The results in the main text directly apply to such settings by interpreting $i$ as indexing the aggregate unit making the selection decision and the unobservables and potential outcomes as aggregate quantities. However, to justify restrictions about selection into treatment, it can be helpful to be more explicit about how selection at the aggregate level is related to the disaggregate level. In the following, we provide a formal framework for doing so. A leading example is when aggregate decisions are based on aggregating preferences at the disaggregate level (e.g., based on voting mechanisms). 

Consider a canonical DiD setting with $S$ groups, indexed by $s\in \{1,\dots,S\}$. Each group contains $n_s$ units, indexed by $i\in \{1,\dots n_s\}$. To simplify the exposition, suppose that all groups are the same size, $n_s=n$ for $s\in \{1,\dots,S\}$. Following the analysis in the main text, we impose general nonseparable models for the disaggregate potential outcomes,
\begin{equation*}
    Y_{ist}(0)=\xi_{st}(\ais,\eist).
\end{equation*}
The aggregate potential outcomes are given by 
$$
Y_{st}(0)=A_{Y(0)}(Y_{1st}(0),\dots,Y_{nst}(0)),
$$
where $A_{Y(0)}(\cdot)$ is a potentially nonlinear aggregation function that can depend on $n$. A simple example is when the aggregate outcomes are averages of the disaggregate outcomes, $Y_{st}(0)=n^{-1}\sum_{i=1}^nY_{ist}(0)$. 

Consider a sharp DiD setting in which the treatment decisions are made at the group level, so that $G_{s}=G_{is}$ for all $i\in \{1,\dots,n\}$, and researchers rely on parallel trends at the group level,
\begin{equation}
    E[Y_{s2}(0)-Y_{s1}(0)|G_s=1]=E[Y_{s2}(0)-Y_{s1}(0)|G_s=0]. \label{eq:aggregate_PT}
\end{equation}
The aggregate selection decision can depend on all unit-level unobservables,
\begin{equation}
G_s=g(\omega_s,\nu_s).\label{eq:aggregate_selection_mechanism}
\end{equation}
Here $\omega_{s}$ is a function or subvector of $(\alpha_s,\varepsilon_{s1},\varepsilon_{s2},\mu_s,\eta_{s1},\eta_{s2})$, where $\alpha_s=(\alpha_{1s},\dots,\alpha_{ns})$, $\varepsilon_{s1}=(\varepsilon_{1s1},\dots,\varepsilon_{ns1})$, and $\varepsilon_{s2}=(\varepsilon_{1s2},\dots,\varepsilon_{ns2})$. The vectors $\mu_s=(\mu_{1s},\dots,\mu_{ns})$, $\eta_{s1}=(\eta_{1s1},\dots,\eta_{ns1})$, and $\eta_{s2}=(\eta_{1s2},\dots,\eta_{ns2})$ contain additional time-invariant and time-varying unobservables. The scalar unobservable $\nu_s$ captures determinants of selection that are independent of $\omega_s$.

All results in the main text directly apply in this setting with $i$ replaced by $s$, such that there are no additional theoretical complications. However, being explicit about the disaggregate level can help ``microfound'' restrictions on the aggregate selection mechanism, as we illustrate in the following example.

\begin{example}[Simple majority voting] Suppose that the aggregate selection decision is based on simple majority voting. Each unit submits a vote $V_{is}\in \{0,1\}$, 
\begin{equation}
V_{is}=v(\omega_{is},\nu_{is}).\label{eq:voting_mechanism}
\end{equation}
where $\omega_{is}$ is a function or subvector of $(\ais,\eisone,\eistwo,\mu_{is},\eta_{is1},\eta_{is2})$.
The voting mechanism \eqref{eq:voting_mechanism} accommodates voting based on group-level unobservables and outcomes since the additional unobservables $(\mu_{is},\eta_{is1},\eta_{is2})$ are unrestricted and can contain group-level quantities. Votes can be based on potential outcomes, expected gains, and fixed effects (as in Examples \ref{ex:selection_outcomes}, \ref{ex:selection_effects}, and \ref{ex:selection_fe}), or other considerations.

The aggregate selection decision under simple majority voting is 
\begin{equation}
G_{s}=1\left\{\frac{1}{n}\sum_{i=1}^{n}V_{is}\ge 0.5\right\}.\label{eq:voting_mechanism_agg}
\end{equation}
This selection mechanism is a special case of mechanism \eqref{eq:aggregate_selection_mechanism}. Restrictions on the aggregate mechanism \eqref{eq:voting_mechanism_agg} can be directly motivated based on assumptions on the units' voting behavior, their information sets, and discount factors. \qed
\end{example}

\section{Multiple periods and multiple groups}\label{app: general did}
Here we generalize our results to DiD designs with multiple periods and multiple groups. The setup and notation are based on \citet{callaway_SantAnna_2021}, \citet{sun_abraham_2021}, and \citet{roth_et_al_DiD_survey}. 

Let $t\in \{1,2,\dots,T\}$ index the periods. Suppose that at time $t=1$, no units are treated, at $t=2$, some units become treated, while others remain untreated, and so on. Previously treated units remain treated for all periods.
Units can be categorized based on their treatment adoption pattern $D_i=(D_{i1},\dots,D_{iT})$. We define the group indicator $G_i$ as the first period in which units are treated, $G_i=\min\{t\in\{1,\dots,T\}:D_{it}=1\}$, and set $G_i=\infty$ for the never-treated units so that $G_i\in \{2,\dots,T,\infty\}$.\footnote{Since $G_i$ is a random variable with finite support, we emphasize that $\{\infty\}$ is merely a label.}

Potential outcomes are indexed by the entire treatment sequence $(d_1,\dots,d_T)\in\{0,1\}^T$, $Y_{it}(d_1,\dots,d_T)$. Since treatment is an absorbing state, the potential outcomes can be indexed by the first treatment period only. Define $Y_{it}(g)=Y_{it}(\mathbf{0}_{g-1},\mathbf{1}_{T-g+1})$ for $g\in \{2,\dots,T\}$ and $Y_{it}(\infty)=Y_{it}(\mathbf{0}_{T})$, where $\mathbf{0}_s\equiv (0,\dots,0)\in \mathbb{R}^s$ and $\mathbf{1}_s\equiv (1,\dots,1)\in \mathbb{R}^s$. 
We maintain a standard no-anticipation assumption \citep[e.g.,][]{roth_et_al_DiD_survey}.
\begin{customass}{NA}
\label{ass:NA}
For $g\in \{2,\dots,T,\infty\}$ and $t<g$, $Y_{it}(g)=Y_{it}(\infty)$.
\end{customass}

Our objects of interest are the group-time ATTs, 
\begin{equation}
\operatorname{ATT}(g,t) = E[Y_{it}(g) - Y_{it}(\infty) | G_i=g].
\end{equation}
We impose the following parallel trends assumption to identify the $\operatorname{ATT}(g,t)$.\footnote{In our setting, this parallel trends assumption corresponds to the ones made by \citet{callaway_SantAnna_2021}, \citet{gardner_two-stage_2021}, \citet{sun_abraham_2021}, \citet{borusyak2024revisiting}, and \citet{wooldridge2021a}; see also \citet{de_chaisemartin_two-way_2020} and \citet{Marcus2021} for related assumptions.}
\begin{customass}{PT-MP} 
\label{ass:PT-MP}
For $(g,t)\in \{2,\dots,T\}^2$,
\begin{equation}
E[Y_{it}(\infty) - Y_{i(t-1)}(\infty) \,|\, G_i = g ] =  E[ Y_{it}(\infty) - Y_{i(t-1)}(\infty) | G_i = \infty]
\end{equation} 
\end{customass}

We consider a general nonseparable outcome model,
$$
Y_{it}(\infty)=\xi_t(\ai,\eit), \quad i=1,\dots,n,\quad  t=1,\dots,T.
$$
Selection into treatment can depend on the unobservable determinants of $Y_{it}(\infty)$ as well as additional unobservables, 
$$
G_i=g(\omega_i,\nu_i),
$$
where $\omega_i$ is a function or a subvector of $(\ai,\eione,\dots,\eiT,\mu_i,\eta_{i1},\dots,\eta_{iT})$
\begin{customass}{SEL-MP}\label{ass:SEL-MP}
Assume that $\nu_i\indep(\ai,\eione,\dots,\eiT,\mu_i,\eta_{i1},\dots,\eta_{iT})$. Furthermore, there exists a non-overlapping partition of the support of $v_i$, $\{B_g\}_{g\in\{2,\dots,T,\infty\}}$, such that $P(v_i\in B_g)\in (0,1)$ for $g\in \{2,\dots,T,\infty\}$. 
\end{customass}
The following theorem presents the necessary and sufficient condition for Assumption \ref{ass:PT-MP}.
\begin{theorem}[Necessary and sufficient condition for \ref{ass:PT-MP} to hold for all $g\in \mathcal{G}_\omega$] \label{thm:main-MP} Suppose that Assumptions \ref{ass:NA} and \ref{ass:SEL-MP} hold. 
 Suppose further that for $t\in\{2,\dots,T\}$ either $P(E[\dot{Y}_{it}(\infty)-\dot{Y}_{i(t-1)}(\infty)|\omega_i]>0)<1$ or $P(E[\dot{Y}_{it}(\infty)-\dot{Y}_{i(t-1)}(\infty)|\omega_i]<0)<1$. Then, Assumption \ref{ass:PT-MP} holds for all $g\in \mathcal{G}_{\omega}$ satisfying $P(G_i=g)\in (0,1)$ for $g\in\{2,\dots,T,\infty\}$ if and only if $E[\dot{Y}_{it}(\infty)-\dot{Y}_{i(t-1)}(\infty)|\omega_i]=0 \text{ a.s.}$ for $t\in\{2,\dots,T\}$.
\end{theorem}
\begin{proof} See Appendix \ref{proof:thm:main-MP}.
\end{proof}

Following the logic of Section \ref{sec:necessary_sufficient}, Theorem \ref{thm:main-MP} implies necessary and sufficient conditions for many different classes of selection mechanisms. For example, if selection is unrestricted, so that $\oi= (\ai,\eione,\dots,\eiT,\mu_i,\eta_{i1},\dots,\eta_{iT})$, we obtain the following necessary and sufficient condition as a corollary of Theorem \ref{thm:main-MP}, 
$$\dot{Y}_{i1}(\infty)=\dot{Y}_{i2}(\infty)=\dots=\dot{Y}_{iT}(\infty),$$
which is a natural extension of Corollary \ref{cor:no_restriction}. 

Similarly, we can consider various classes of restricted selection mechanisms. For example, if $\oi=(\ai,\mu_i)$, we obtain the following necessary and sufficient condition,  
$$E[\dot{Y}_{i1}(\infty)|\ai,\mu_i]=\dots=E[\dot{Y}_{iT}(\infty)|\ai,\mu_i],$$
which is a multi-period version of Corollary \ref{cor:fe}. There are many other corollaries of Theorem \ref{thm:main-MP} depending on the choice of $\omega_i$.

\section{Sufficient conditions for parallel trends beyond Theorem \ref{thm:main}}\label{app:sufficient_condition_exchangeability}
Theorem \ref{thm:main} allows researchers to derive a broad set of sufficient (and necessary) conditions for parallel trends, depending on what unit select on. However, in some applications, researchers might be interested in imposing other types of restrictions not covered by Theorem \ref{thm:main}. 

To illustrate, consider the separable model \eqref{eq:separable}. The following is a sufficient condition for parallel trends based on a symmetry condition on the selection mechanism:
\begin{enumerate}
\item[(e)] Symmetric selection:
\begin{enumerate}[(i)]\setlength\itemsep{0pt}
    \item $g$ is a symmetric function in $\eione$ and $\eitwo$ 
    \item $\eione,\eitwo|\ai\overset{d}{=}\eitwo,\eione|\ai$ and $(\mu_i,\eta_{i1},\eta_{i2})|\ai,\eione,\eitwo\overset{d}{=}(\mu_i,\eta_{i1},\eta_{i2})|\ai,\eitwo,\eione$
\end{enumerate}
\end{enumerate}
Condition (e) can be shown to imply parallel trends. In addition to the symmetry of the selection mechanism, this sufficient condition imposes two different types of exchangeability restrictions. First, it requires that the conditional distribution of $(\mu_i,\eta_{i1},\eta_{i2})$ is exchangeable in $\eione$ and $\eitwo$ after conditioning on $\ai$. This notion of exchangeability has been employed, for example, in \citet{altonji2005cross}. Second, it requires the distribution of $(\eione,\eitwo)$ to be exchangeable conditional on $\ai$. 

To illustrate condition (e), consider the selection mechanism in Example \ref{ex:selection_outcomes} and suppose that the units have complete information, so that $\omega_i=(\ai,\eione,\eitwo,\kappa_{i2})$. In this case, the selection mechanism \eqref{eq:AC_selection} simplifies to $G_i=1\left\{\ai(1+\beta)+\eione+\beta \eitwo\le \kappa_{i2}\right\}.$
This selection mechanism is symmetric if there is no discounting ($\beta=1$), which may be plausible if the time span between the pre- and post-treatment period is short. In addition, the condition (ii) requires $(\eione,\eitwo)$ to be exchangeable and the conditional distribution of the costs to be symmetric in $\eione$ and $\eitwo$, $\kappa_{i2}|\ai,\eione,\eitwo\overset{d}{=}\kappa_{i2}|\ai,\eitwo,\eione$. 

\section{Covariates}\label{app:covariates}
In many applications, parallel trends may only be plausible conditional on covariates \citep[e.g.,][]{heckman1997matching,abadie2005DiD,santanna2020drdid, callaway_SantAnna_2021}. In this appendix, we therefore study the role of covariates through the lens of selection into treatment. While many existing approaches focus on time-invariant covariates, we explicitly allow for a vector of both time-invariant and time-varying  covariates, $X_{it}$, assuming that $X_{it}$ is not affected by the treatment.

In Appendix \ref{app:separability}, we derive necessary and sufficient conditions for parallel trends to hold conditional on covariates and show that these conditions imply separability restrictions on how covariates enter the outcome model. In Appendix \ref{sec:covariates_nonseparable}, we propose a modified parallel trends assumption that accommodates nonseparable models and provide sufficient conditions for this modified parallel trends assumption.

\subsection{Covariates and the role of separability}\label{app:separability}

Suppose that parallel trends holds conditional on the time series of covariates,  $X_i=(X_{i1},X_{i2})$. 
\begin{customass}{PT-X}\label{ass:PT-X}
The conditional parallel trends assumption holds:
$$
E[Y_{i2}(0)-Y_{i1}(0)|G_i=1,X_{i}]=E[Y_{i2}(0)-Y_{i1}(0)|G_i=0,X_{i}]~a.s.
$$
\end{customass}
\noindent Under Assumption \ref{ass:PT-X}, the unconditional ATT is identified as
\begin{align*}
    E[Y_{i2}(1)-Y_{i2}(0)|G_i=1]=E\left[\att(X_i)|G_i=1\right]=E\left[\did(X_i)|G_i=1\right],
\end{align*}
where $\att(X_i)\equiv E[Y_{i2}(1)-Y_{i2}(0)|G_i=1,X_{i}]$ and $\did(X_i)\equiv E[Y_{i2}-Y_{i1}|G_i=1,X_{i}]-E[Y_{i2}-Y_{i1}|G_i=0,X_{i}]$.

In the presence of covariates, potential outcomes and selection into treatment may naturally depend on them. We therefore consider the following outcome model and selection mechanism,
\begin{eqnarray*}
Y_{it}(0)&=&\xi_t(X_{it},\ai,\eit),\quad i=1,\dots,n,\quad t=1,2,\\
G_i&=&g(X_{i},\oi,\nu_i),\quad i=1,\dots,n.
\end{eqnarray*}
Let $\mathcal{X}$ denote the support of $X_{i}$ and define $\mathcal{G}^x_{\omega}\equiv\{g:\mathcal{X}\times \mathcal{O}\times \mathcal{V}\to \{0,1\}\}$.

The necessary and sufficient condition in Theorem \ref{thm:main} generalizes straightforwardly to Assumption \ref{ass:PT-X} as we show in the following theorem. Before we proceed, we introduce the notation $\ddot{Y}_{it}(0)\equiv Y_{it}(0)-E[Y_{it}(0)|X_i]$ and the following regularity condition on $\nu_i$.
\begin{customass}{SEL-X}
$P(\nu_i>c|X_i)\in(0,1)$ for some $c\in \mathbb{R}$ and $\nu_i\indep (X_i,\ai,\eione,\eitwo,\mu_i,\eta_{i1},\eta_{i2})$.\label{ass:SEL-X}
\end{customass}

\begin{theorem}[Necessary and sufficient condition for Assumption \ref{ass:PT-X} for all $g\in \mathcal{G}^x_\omega$] \label{thm:main-X} Suppose that Assumption \ref{ass:SEL-X} holds. 
 Suppose further that either $P(E[\ddot{Y}_{i2}(0)-\ddot{Y}_{i1}(0)|\omega_i,X_i]>0)<1$ or $P(E[\ddot{Y}_{i2}(0)-\ddot{Y}_{i1}(0)|\omega_i,X_i]<0)<1$. Then, Assumption \ref{ass:PT-X} holds for all $g\in \mathcal{G}_{\omega}^x$ satisfying $P(G_i=1|X_i)\in (0,1)$ if and only if $E[\ddot{Y}_{i2}(0)-\ddot{Y}_{i1}(0)|\omega_i,X_i]=0\text{ a.s.}$
\end{theorem}
\begin{proof}
    See Appendix \ref{proof:thm:main-X}.
\end{proof}

The proof follows from the same arguments as in the proof of Theorem \ref{thm:main}, conditional on covariates $X_i$. Given the necessary and sufficient condition in Theorem \ref{thm:main-X}, the results in the main text generalize straightforwardly to settings with covariates.

An important practical implication of the necessary and sufficient condition in Theorem \ref{thm:main-X} is that it implies separability requirements on how the covariates can enter the outcome model. To illustrate, consider the simple case where $\oi=\ai$. In this case, the necessary and sufficient condition in Theorem \ref{thm:main-X} can be written as
$$
E[Y_{i2}(0)|X_i,\ai]-E[Y_{i1}(0)|X_i,\ai]=E[Y_{i2}(0)|X_i]-E[Y_{i1}(0)|X_i].
$$
To illustrate the separability restrictions, consider a generalized random coefficient model \citep[e.g.,][]{chamberlain1992efficiency} in which $\ai$ interacts with $X_{it}$,
\begin{equation}Y_{it}(0)=\ai \chi_t(X_{it})+\lambda_t+\eit.\label{eq:chamberlain_model}\end{equation}
Here $\chi_t(\cdot)$ is an arbitrary time-varying function. Even under the assumption that $E[\eit|X_i,\ai]=0$, this model generally violates the necessary condition due to the combination of nonseparability between $\ai$ and $X_{it}$ and the time variability in the structural function through $\chi_t(\cdot)$,  
$$E[Y_{i2}(0)|X_i,\ai]-E[Y_{i1}(0)|X_i,\ai]=\ai (\chi_2(X_{i2})-\chi_1(X_{i1}))+\lambda_2-\lambda_1.$$
This example demonstrates that for parallel trends to hold in the presence of interactions between covariates and $\ai$, it is not sufficient to focus on subpopulations with $X_{i1}=X_{i2}$. We additionally require that the component that interacts with $\ai$, $\chi_t(\cdot)$, does not vary across time.

Allowing for interactions between the unobservable determinants of selection and some covariates is important in applications. We therefore consider a weaker conditional parallel trends assumption that allows for such interactions in Section \ref{sec:covariates_nonseparable}.

\subsection{Selection-based templates for justifying parallel trends in separable models with covariates}\label{sec:selection-templates-PTX}
The necessary and sufficient conditions with covariates can serve as theory-based templates for assessing and justifying parallel trends in DiD applications with covariates. The exact form of these conditions depends on the model for $Y_{it}(0)$. Consider, for example, the separable model $$Y_{it}(0)=\ai+\chi_t(X_{it})+\lt+\eit,$$ where $\chi_t(\cdot)$  is a potentially time-varying function. For this model, we can provide a menu of necessary and  sufficient conditions that specify the determinants of selection as well as conditions on the idiosyncratic shocks $\eit$. For instance,
\begin{enumerate}[(a)]\setlength\itemsep{0pt}
\item Imperfect foresight (case 1): (i) $\oi=(\ai,\eione,\mu_i,\eta_{i1})$ and (ii) $E[\eitwo-\eione|X_{i},\ai,\eione,\mu_i,\eta_{i1}]=E[\eitwo-\eione|X_{i}]$
\item Imperfect foresight (case 2): (i) $\oi=(\mu_i,\eta_{i1})$ and (ii) $E[\eitwo-\eione|X_{i},\mu_i,\eta_{i1}]=E[\eitwo-\eione|X_{i}]$

\item Roy-style selection: (i) $\oi=(\tau_{i2},\kappa_{i2})$ and (ii) $E[\eitwo-\eione|X_{i},\tau_{i2},\kappa_{i2}]=E[\eitwo-\eione|X_{i}]$

\item Selection on fixed effects: (i) $\oi=(\ai,\mu_i)$ and (ii) $E[\eitwo-\eione|X_{i},\ai,\mu_i]=E[\eitwo-\eione|X_{i}]$
\end{enumerate}

\subsection{A parallel trends assumption for nonseparable models}\label{sec:covariates_nonseparable}
Motivated by Section \ref{app:separability}, we consider a weaker (than Assumption \ref{ass:PT-X}) conditional parallel trends assumption that accommodates nonseparable models. To state this assumption, we explicitly differentiate between two types of covariates: (i) $X_{it}^\mu$ are covariates that interact with the unobservable determinants of selection in the outcome model; (ii) $X_{it}^\lambda$ are covariates that do not interact with these unobservables. Both types of covariates can enter the selection mechanism in an arbitrary way. The following conditional parallel trends assumption holds for subpopulations that experience no change in $X_{it}^\mu$ and the same trajectory in $X_{it}^\lambda$.
\begin{customass}{PT-NSP}\label{ass:PT-NSP} The (modified) conditional parallel trends assumption holds:
$$E[Y_{i2}(0)-Y_{i1}(0)|G_i=1,X_{i}^\lambda,X_{i1}^\mu=X_{i2}^\mu]=E[Y_{i2}(0)-Y_{i1}(0)|G_i=0,X_{i}^\lambda,X_{i1}^\mu=X_{i2}^\mu]~a.s.$$
\end{customass}
Under Assumption \ref{ass:PT-NSP}, we can no longer identify the ATT, $E[Y_{i2}(1)-Y_{i2}(0)|G_i=1]$, because we cannot identify the conditional ATT, $E[Y_{i2}(1)-Y_{i2}(0)|G_i=1,X_{i}^\lambda,X_{i}^\mu]$. Instead, we can identify $E[Y_{i2}(1)-Y_{i2}(0)|G_i=1,X_{i}^\lambda,X_{i1}^\mu=X_{i2}^\mu].$\footnote{With a slight abuse of notation, we use $(X_{i1}^\mu=X_{i2}^\mu)$ in the conditioning set as a short-hand for $(X_{i1}^\mu,X_{i2}^\mu=X_{i1}^\mu)$.}
After integrating out with respect to the distribution of covariates, we can identify the ATT for subpopulations that do not experience changes in $X_{it}^\mu$, $$E[Y_{i2}(1)-Y_{i2}(0)|G_i=1,X_{i1}^\mu-X_{i2}^\mu=0].$$
Note that if $X_{it}^\mu$ is time-invariant, then $X_{i1}^\mu=X_{i2}^\mu$ holds by definition and Assumptions \ref{ass:PT-X} and \ref{ass:PT-NSP} are equivalent.

Next, we provide different sets of sufficient conditions for Assumption \ref{ass:PT-NSP} under the following nonseparable outcome model.\footnote{Without further restrictions on the unobservables, the additive structure is without loss of generality and the superscripts $\mu$ and $\lambda$ are merely labels. Indeed, if $X_{it}^\mu=X_{it}^\lambda$, $\ai^\mu=\ai^\lambda$, and $\eit^\mu=\eit^\lambda$, the model is fully nonseparable and time-varying in an arbitrary way.}
\begin{customass}{NSP-X}\label{ass:NSP-X}
 \begin{align}Y_{it}(0)&=\mu(X_{it}^\mu,\ai^\mu,\eit^\mu)+\lt(X_{it}^\lambda,\ai^\lambda,\eit^\lambda), \quad i=1,\dots,n, \quad t=1, 2,\nonumber
 \end{align}
 where $X_{it}^\mu$, $X_{it}^\lambda$, $\ai^\mu$, $\ai^\lambda$, $\eit^\mu$, and $ \eit^\lambda$ are finite-dimensional random vectors. 
 \end{customass}
Let $\mathcal{X}_{\mu}$, $\mathcal{X}_{\lambda}$, $\mathcal{A}$, and $\mathcal{E}$ denote the supports of $X_{it}^\mu$, $X_{it}^\lambda$, $\ai^\mu$, and $\eit^\mu$, respectively. 

It is natural to consider selection based on the unobservables entering $\mu(\cdot)$. We therefore impose the following condition on the projected selection mechanism.
\begin{customass}{SEL-CI}\label{ass:SEL-CI}
$$E[G_i| X_{i1}^\mu,X_{i2}^\mu,X_{i1}^\lambda, X_{i2}^\lambda,\ai^\mu,\ai^\lambda,\eione^\mu,\eitwo^\mu,\eione^\lambda,\eitwo^\lambda]=E[G_i| X_{i1}^\mu,X_{i2}^\mu,X_{i1}^\lambda, X_{i2}^\lambda,\ai^\mu,\eione^\mu,\eitwo^\mu].$$ 
\end{customass}

Assumption \ref{ass:SEL-CI} allows the projected selection mechanism to depend on all covariates, but only on the unobservables that enter $\mu(\cdot)$. In view of Assumption \ref{ass:SEL-CI}, we define
\begin{eqnarray}
&&\bar{g}(x_1^\mu,x_2^\mu,x_1^\lambda,x_2^\lambda,a^\mu,e_1^\mu,e_2^\mu)\nonumber\\
&&\equiv E[G_i| X_{i1}^\mu=x_1^\mu,X_{i2}^\mu=x_2^\mu,X_{i1}^\lambda=x_1^\lambda, X_{i2}^\lambda=x_2^\lambda,\ai^\mu=a^\mu,\eione^\mu=e_1^\mu,\eitwo^\mu=e_2^\mu].\nonumber\end{eqnarray}

In the following, we present different sets of sufficient conditions for Assumption \ref{ass:PT-NSP}, inspired by the sufficient conditions in Section \ref{sec:other} and Appendix \ref{app:sufficient_condition_exchangeability}. For concreteness, we focus on three sets of sufficient conditions; other sets of sufficient conditions could also be considered. Each set of conditions consists of assumptions on the projected selection mechanism as well as distributional restrictions on the unobservables. Our first sufficient condition allows selection to depend on all covariates as well as the unobservables that enter the time-invariant component of the structural function, while imposing a symmetry restriction on the projected selection mechanism.
\begin{customass}{SC1-NSP}\label{ass:SC1-NSP} The following conditions hold:
\begin{enumerate}[(i)]\setlength\itemsep{0pt}
\item $\bar{g}(x_1^\mu,x_2^\mu,x_1^\lambda,x_2^\lambda,a^\mu,e_1^\mu,e_2^\mu)$ is a symmetric function in $e_1^\mu$ and $e_2^\mu$.\label{ass:SC1-NSP_selection_symmetry}
\item $(\eione^\mu,\eitwo^\mu)|X_i^\mu,X_i^\lambda,\ai^\mu\overset{d}{=}(\eitwo^\mu,\eione^\mu)|X_i^\mu,X_i^\lambda,\ai^\mu$.\label{ass:SC1-NSP_exchangeability}
\item $(\ai^\mu,\eione^\mu,\eitwo^\mu)\indep (\ai^\lambda,\eione^\lambda,\eitwo^\lambda)|X_{i}^\mu,X_i^\lambda$.\label{ass:SC1-NSP_independence}
\end{enumerate}
\end{customass}
Here we require the conditional distribution of $(\eione^\mu,\eitwo^\mu)|X_i^\mu,X_i^\lambda,\ai^\mu$ to be exchangeable. Since the projected selection mechanism depends on $(\ai^\mu, \eione^\mu, \eitwo^\mu)$, we require them to be independent of the unobservables entering $\lt(\cdot)$ conditional on $(X_i^\mu,X_i^\lambda)$.

The exchangeability restriction in Assumption \ref{ass:SC1-NSP} is different from the exchangeability assumption in \citet{altonji2005cross}.  The exchangeability assumption in \citet{altonji2005cross} requires the conditional distribution of all unobservables that enter $\mu(\cdot)$ and $\lambda_t(\cdot)$ to be invariant to permutations of covariates in the conditioning set, which is a nonparametric correlated random effects restriction \citep{ghanem2017testing}.  By contrast, we assume that the time-varying unobservables are exchangeable conditional on $(X_i^\mu,X_i^\lambda,\ai^\mu)$ without imposing any restrictions on the distribution of $\ai^\mu|G_i,X_i^\mu,X_i^\lambda$.

Next, we consider a projected selection mechanism that is a trivial function of $\eitwo^\mu$ in the following sufficient condition.
\begin{customass}{SC2-NSP}\label{ass:SC2-NSP}
The following conditions hold:
 \begin{enumerate}[(i)]\setlength\itemsep{0pt}
\item $\bar{g}(x_1^\mu,x_2^\mu,x_1^\lambda,x_2^\lambda,a^\mu,e_1^\mu,e_2^\mu)$ is a trivial function of $e_2^\mu$.\label{ass:SC2-NSP_selection}
\item $(\ai^\mu,\eione^\mu)\indep \Delta_{\mu,i}|X_{i}^\lambda,X_{i1}^\mu=X_{i2}^\mu$, where $\Delta_{\mu,i}\equiv \mu(X_{i2}^\mu,\ai^\mu,\eitwo^\mu)-\mu(X_{i1}^\mu,\ai^\mu,\eione^\mu)$. \label{ass:SC2-NSP_separability}
\item $(\ai^\mu,\eione^\mu)\indep (\ai^\lambda,\eione^\lambda,\eitwo^\lambda)|X_{i}^\mu,X_i^\lambda$.\label{ass:SC2-NSP_independence}
\end{enumerate}
\end{customass}
Assumption \ref{ass:SC2-NSP}.\ref{ass:SC2-NSP_separability} implicitly imposes separability conditions on $\mu(\cdot)$ (but not on $\lambda_t(\cdot)$) and restrictions on time series dependence.\footnote{To see this, note that since $\Delta_{\mu,i}=\mu(X_{i2}^\mu,\ai^\mu,\eitwo^\mu)-\mu(X_{i1}^\mu,\ai^\mu,\eione^\mu)$, for $\Delta_{\mu,i}$ to be conditionally independent of $(\ai^\mu,\eione^\mu,\eitwo^\mu)$, a sufficient condition would be that $\Delta_{\mu,i}$ is separable in $\ai^\mu$ and  $\eit^\mu$, such that $\Delta_{\mu,i}=\gamma\left(X_{i2}^\mu\right)-\gamma\left(X_{i1}^\mu\right)+\eitwo^\mu-\eione^\mu$, as well as that $(\eitwo^\mu-\eione^\mu)$ and $(\ai^\mu,\eione^\mu)$ are conditionally independent.} The independence condition in Assumption \ref{ass:SC2-NSP}.\ref{ass:SC2-NSP_independence} requires that the unobservable determinants of selection are independent of the unobservables that enter $\lt(\cdot)$ conditional on the time series of covariates.

The last sufficient condition restricts the projected selection mechanism to only depend on covariates and the time-invariant unobservables.
\begin{customass}{SC3-NSP}\label{ass:SC3-NSP}
The following conditions hold:
 \begin{enumerate}[(i)]\setlength\itemsep{0pt}
 \item $\bar{g}(x_1^\mu,x_2^\mu,x_1^\lambda,x_2^\lambda,a^\mu,e_1^\mu,e_2^\mu)$ is a trivial function of $e_1^\mu$ and $e_2^\mu$.\label{ass:SC3-NSP_selection_ai}
   \item $\eione^\mu|X_{i}^\mu,X_{i}^\lambda,\ai^\mu\overset{d}{=}\eitwo^\mu|X_{i}^\mu,X_{i}^\lambda,\ai^\mu$.\label{ass:SC3-NSP_time_homogeneity}
\item $\ai^\mu\indep (\ai^\lambda,\eione^\lambda,\eitwo^\lambda)|X_{i}^\mu,X_{i}^\lambda$.\label{ass:SC3-NSP_independence}
\end{enumerate}
\end{customass}
Assumption \ref{ass:SC3-NSP} requires the distribution of $\eit^\mu$, which enters $\mu(\cdot)$, to be time-invariant conditional on $(\ai^\mu,X_{i}^\mu,X_i^\lambda)$. The unobservables entering $\lambda_t(\cdot)$, $(\ai^\lambda,\eione^\lambda,\eitwo^\lambda)$, are required to be independent of the unobservables that determine selection, $\ai^\mu$, conditional on $(X_i^\mu,X_i^\lambda)$.  

Each of the sufficient conditions consists of three components: (i) a restriction on how/which unobservables enter the projected selection mechanism, (ii) a restriction on the unobservables entering the time-invariant component of the structural function, and (iii) an independence assumption that ensures that the time-varying component of the structural function is independent of $G_i$ conditional on the time series of covariates. 

The following proposition formally establishes sufficiency of each set of conditions.

\begin{proposition}[Sufficient conditions]\label{prop:sufficient-NSP} Suppose that Assumptions \ref{ass:NSP-X} and \ref{ass:SEL-CI} hold and $P(G_i=1|X_{i}^\lambda, X_{i1}^\mu=X_{i2}^\mu)\in (0,1)~a.s.$ Then (i) Assumption \ref{ass:SC1-NSP} implies Assumption \ref{ass:PT-NSP}, (ii) Assumption \ref{ass:SC2-NSP} implies  Assumption \ref{ass:PT-NSP}, and (iii) Assumption \ref{ass:SC3-NSP} implies Assumption \ref{ass:PT-NSP}.
\end{proposition}
\begin{proof}
    See Appendix \ref{proof:prop:sufficient-NSP}.
\end{proof}
\begin{remark}[Connection to unconfoundedness]
All sufficient conditions in Proposition \ref{prop:sufficient-NSP} allow for selection on unobservable determinants of the untreated potential outcome. This is in contrast to the unconfoundedness assumptions commonly used in cross-sectional studies \citep[e.g.,][]{imbens2004nonparametric,imbens2009recent}. Therefore, these results elucidate the differences between conditional parallel trends and unconfoundedness-type assumptions.\qed
\end{remark}

\section{Connections to identification assumptions in panel models}\label{app:connections}
Here we relate the selection-based sufficient conditions in Section \ref{ass:PT-NSP} to the identification assumptions in the nonseparable panel literature.\footnote{See, e.g., \citet{altonji2005cross,athey2006identification,bester2009identification,hoderlein2012nonparametric, chernozhukov2013average,arellano2016nonlinear,ghanem2017testing}. This work extends notions of fixed effects and correlated random effects that originated in the linear model \citep{mundlak1961empirical,mundlak1978on,chamberlain1982multivariate,chamberlain1984panel}. Surveys \citep{arellano2001panel,arellano2011nonlinear} and textbook treatments \citep{arellano2003panel,wooldridge2010econometric} further describe the role of restrictions on time and individual heterogeneity in linear and nonlinear models. Such restrictions have been imposed in the context of identification in limited dependent variable models \citep[e.g.][]{manski1987semiparametric,honore1993orthogonality,kyriazidou1997estimation,honore2000estimation,honore2000panel} and random coefficient models \citep[e.g.][]{chamberlain1992efficiency,graham2012identification,arellano2012identifying}. Nonparametric identification of panel models with additivity restrictions has been examined, e.g., in \citet{evdokimov2010identification} and \citet{freyberger2017non-parametric}.} The literature on nonseparable panel models has considered two broad categories of identification assumptions. First, time homogeneity conditions \citep[e.g.,][]{hoderlein2012nonparametric,chernozhukov2013average} require the distribution of time-varying unobservables to be stationary across time while allowing for unrestricted individual heterogeneity (fixed effects). Second, nonparametric correlated random effects restrictions \citep[e.g.,][]{altonji2005cross,bester2009identification} allow for unrestricted time heterogeneity by imposing restrictions on individual heterogeneity, generalizing the classical notion of correlated random effects \citep[e.g.,][]{mundlak1978on,chamberlain1984panel}. However, neither category of assumptions is explicit about the selection mechanism and, in particular, about how unobservables determine selection.

The existing identification results based on time homogeneity or correlated random effects assumptions suggest a trade-off between restrictions on time and individual heterogeneity. Here we show that our sufficient conditions for Assumption \ref{ass:PT-NSP} constitute interpretable primitive conditions on the selection mechanism that imply \emph{combinations} of time homogeneity and correlated random effects restrictions from the nonseparable panel literature.

The following assumption is the time homogeneity assumption from \citet{chernozhukov2013average} imposed on $\eit^\mu$ in Assumption \ref{ass:NSP-X}, conditional on the time series of all covariates that enter the outcome equation. 
\begin{customass}{TH}\label{ass:TH-X}
$\eione^\mu|G_i,X_i^\mu,X_i^\lambda,\ai^\mu\overset{d}{=}\eitwo^\mu|G_i,X_i^\mu,X_i^\lambda,\ai^\mu$
\end{customass}

Assumption \ref{ass:TH-X} requires the distribution of $\eit^\mu$ to be homogeneous across time conditional on $G_i$, $X_i^\mu$, $X_i^\lambda$, and $\ai^\mu$. However, it does not impose any restrictions on the conditional distribution of $\eit^\mu$. Furthermore, there are no restrictions imposed on the distribution of $\ai^\mu|G_i,X_i^\mu,X_i^\lambda$, consistent with the notion of fixed effects.

The next assumption is a nonparametric correlated random effects assumption \citep[e.g.,][]{altonji2005cross,ghanem2017testing}. 
\begin{customass}{CRE}\label{ass:RE-X}
$(\ai^\lambda,\eione^\lambda,\eitwo^\lambda)|G_i,X_i^\mu,X_i^\lambda\overset{d}{=}(\ai^\lambda,\eione^\lambda,\eitwo^\lambda)|X_i^\mu,X_i^\lambda$.
\end{customass}
Assumption \ref{ass:RE-X} is a conditional independence condition between $G_i$ and the unobservables that enter the time-varying component of the structural function, $\lt(\cdot)$. This assumption does not imply conditional random assignment, $(Y_{i1}(0),Y_{i2}(0))\indep G_i|X_{i}^\mu,X_i^\lambda$, since selection into treatment can depend on the unobservables entering the time-invariant component $\mu(\cdot)$.

Together, Assumptions \ref{ass:TH-X} and \ref{ass:RE-X} imply Assumption \ref{ass:PT-NSP}.\footnote{\citet[][Appendix B]{ghanem2017testing} discusses the nonparametric identification of the ATT through DiD either through time homogeneity or random effects assumptions.} 
\begin{proposition}[\ref{ass:TH-X} and \ref{ass:RE-X} imply \ref{ass:PT-NSP}]\label{prop:TH+CRE imply parallel trends}
Suppose that Assumption \ref{ass:NSP-X} holds and $P(G_i=1|X_{i1}^\mu=X_{i2}^\mu,X_i^\lambda)\in(0,1)$ a.s. Then Assumptions \ref{ass:TH-X} and \ref{ass:RE-X} imply Assumption \ref{ass:PT-NSP}.
\end{proposition}
\begin{proof}
    See Appendix \ref{proof:prop:TH+CRE imply parallel trends}.
\end{proof}
In view of Proposition \ref{prop:TH+CRE imply parallel trends}, it is interesting to explore the connection between selection, time homogeneity, and correlated random effects in the nonseparable DiD framework. To this end, Proposition \ref{prop:connection-NSP-X} shows that Assumptions \ref{ass:SC1-NSP} and \ref{ass:SC3-NSP} are primitive sufficient conditions on the selection mechanism for the nonseparable model satisfying Assumptions \ref{ass:TH-X} and \ref{ass:RE-X}.\footnote{In the context of correlated random coefficient models, \citet{graham2012identification} impose a similar structure on their model.} In the following, with a slight abuse of notation, we let $\mathcal{A}$ denote the support of $\ai^\mu$. For $t=1,2$, let $\mathcal{E}$ denote the support of $\eit^\mu$, $\mathcal{X}_\mu$ the support of $X_{it}^\mu$ and $\mathcal{X}_\lambda$ the support of $X_{it}^\lambda$.

\begin{proposition}[Connection between selection, time homogeneity, and correlated random effects]\label{prop:connection-NSP-X} Suppose that Assumption \ref{ass:NSP-X} holds and $G_i=g(X_{i1}^\mu,X_{i2}^\mu,X_{i1}^\lambda,X_{i2}^\lambda,\ai^\mu,\eione^\mu,\eitwo^\mu)$.\footnote{For simplicity, we assume that the selection mechanism only depends on the unobservables that enter the time-invariant component of the structural function $\mu(\cdot)$ in addition to covariates.} Assume that the conditional density of $(\eione^\mu,\eitwo^\mu)|X_i^\mu,X_i^\lambda,\ai^\mu$,  $f_{\eione^\mu,\eitwo^\mu|X_i^\mu,X_i^\lambda,\ai^\mu}(e_1,e_2|x^\mu,x^\lambda,a)$, exists and the conditional density of $\eit^\mu|X_i^\mu,X_i^\lambda,\ai^\mu$, $f_{\eit|X_i^\mu,X_i^\lambda,\ai}(e|x^\mu,x^\lambda,a)$, is strictly positive  for $(e,x^\mu,x^\lambda,a)\in\mathcal{E}\times\mathcal{X}_{\mu}^2\times\mathcal{X}_\lambda^2\times \mathcal{A}$, $t=1,2$. Then (i) Assumption \ref{ass:SC1-NSP} with $g(\cdot)$ in lieu of $\bar{g}(\cdot)$ implies Assumptions \ref{ass:TH-X} and \ref{ass:RE-X} if $P(G_i=1|X_i^\mu,X_i^\lambda,\ai^\mu)\in (0,1)$ a.s., (ii) Assumption \ref{ass:SC3-NSP} with $g(\cdot)$ in lieu of $\bar{g}(\cdot)$ implies Assumptions \ref{ass:TH-X} and \ref{ass:RE-X}.
\end{proposition}
\begin{proof}
    See Appendix \ref{proof:prop:connection-NSP-X}.
\end{proof}
Proposition \ref{prop:connection-NSP-X}  demonstrates how restrictions on selection can be used to justify combinations of Assumptions \ref{ass:TH-X} and \ref{ass:RE-X}.

\section{Bias decomposition with covariates}
\label{app: bias decomposition covariates}

Here we extend the analysis in Section \ref{sec:bias_characterization} to the setup with covariates in Appendix \ref{app:covariates}. 

\subsection{Decomposition}
The necessary and sufficient condition in Corollary \ref{cor:if1} with covariates and one additional pre-treatment period is
$
E[\ddot{Y}_{i2}(0)|X_i,\ai,\eipre,\mu_i,\etaipre]=\ddot{Y}_{i1}(0),
$ 
where $\ddot{Y}_{it}(0)\equiv Y_{it}(0)-E[Y_{it}(0)|X_i]$.
Consider the following linear relaxation of the martingale property,
\begin{equation}
E[\ddot{Y}_{it}(0)|X_i,\ai,\varepsilon_i^{t-1},\mu_i,\eta_i^{t-1}]=\rho_t\ddot{Y}_{i(t-1)}(0), \quad i=1,\dots,n,\quad t=1,2,\label{eq:relaxation_covariates_linear}
\end{equation}
and define $\zeta_{it}\equiv \ddot{Y}_{it}(0)-\rho_t\ddot{Y}_{i(t-1)}$. 

By the same arguments as in  Proposition \ref{prop:bias_characterization}, conditional on $X_i$, we can decompose the bias of the conditional ATT as follows
\begin{eqnarray*}
        \biasdidpost(X_i)\equiv\did(X_i)-\att(X_i)=\biasselpost(X_i)+\biasmtgpost(X_i),
\end{eqnarray*}
where 
\begin{eqnarray*}\biasselpost(X_i)&\equiv&E[\zeta_{i2}|G_i=1,X_i]-E[\zeta_{i2}|G_i=0,X_i],\\
\biasmtgpost(X_i)&\equiv&(\rho_2-1)(E[{Y}_{i1}|G_i=1,X_i]-E[{Y}_{i1}|G_i=0,X_i]).
\end{eqnarray*}
Therefore, the unconditional bias is (with a slight abuse of notation)
\begin{eqnarray*}
\biasdidpost&\equiv &E[\biasdidpost(X_i)|G_i=1]\\
&=&E[E[\zeta_{i2}|G_i=1,X_i]-E[\zeta_{i2}|G_i=0,X_i]|G_i=1]\\
&&+(\rho_2-1)(E[Y_{i1}|G_i=1]-E[E[Y_{i1}|G_i=0,X_i]|G_i=1])\\
&\equiv &\biasselpost+\biasmtgpost
\end{eqnarray*}

Analogous to Section \ref{sec:benchmarking},  we rely on pre-treament and control group counterparts of $\rho_2$ and the pre-treatment counterpart of $\biasselpost$, $
\biasselpre=E[\zeta_{i1}|G_i=1]-E[E[\zeta_{i1}|G_i=0,X_i]|G_i=1]$, to benchmark the bias.

\subsection{Implementation}\label{app:implementation}
 For a random variable $W_{it}$, let $E_n[W_{it} | G_i=1] = \left. \sum_{i=1}^n G_i W_{it} \right/ \sum_{i=1}^n G_i$ and $\widehat{E}[W_{it}|G_i=g,X_i]$ ($\widehat{E}[W_{it}|X_i]$) be a linear regression estimator of $E[W_{it}|G_i=g,X_i]$ ($E[W_{it}|X_i]$). 

 Given $\biasselpost$ and $\rho_2$, an estimator of $\biasdidpost$ is given by
\begin{eqnarray*}\biasdidposthat=\biasdidposthat\left(\biasselpost,\rho_2\right)=\biasselpost+(\rho_2-1)\left(E_n[Y_{i1}|G_i=1]-E_n[\widehat{E}[Y_{i1}|G_i=0,X_i]|G_i=1]\right).
\end{eqnarray*}
If $\biasselpost=0$, then $$\biasdidposthat=(\rho_2-1)\left(E_n[Y_{i1}|G_i=1]-E_n[\widehat{E}[Y_{i1}|G_i=0,X_i]|G_i=1]\right).$$
Alternatively, if $\biasselpost=\biasselpre$, then
$$\biasdidposthat=\biasselprehat+(\rho_2-1)\left(E_n[Y_{i1}|G_i=1]-E_n[\widehat{E}[Y_{i1}|G_i=0,X_i]|G_i=1]\right),$$
where $\biasselprehat=E_n[\hat{\zeta}_{i1}|G_i=1]-E_n\left[\widehat{E}\left[\widehat{\zeta}_{i1}|G_i=0,X_i\right]|G_i=1\right]$, $\widehat{\zeta}_{i1}=\widehat{\ddot{Y}}_{i1}-\widehat{\rho}_1\widehat{\ddot{Y}}_{i0}$, $\widehat{\ddot{Y}}_{it}=Y_{it}-\widehat{E}[Y_{it}|X_i]$, and $\widehat{\rho}_1$ is obtained from a regression of $\widehat{\ddot{Y}}_{i1}$ on $\widehat{\ddot{Y}}_{i0}$.

\section{Proofs of the results in the main text}

\subsection{Auxiliary lemmas}

\begin{lemma}\label{lem: PT equivalent condition}
Let $W_i$ denote a vector of random variables. Suppose that $P(G_i=1|W_i)\in (0,1)$ a.s. Then $E[Y_{i2}(0)-Y_{i1}(0)|G_i=1,W_i]=E[Y_{i2}(0)-Y_{i1}(0)|G_i=0,W_i]$ if and only if $E[G_i(Y_{i2}(0)-Y_{i1}(0))|W_i]=E[G_i|W_i]E[Y_{i2}(0)-Y_{i1}(0)|W_i]$ a.s.
\end{lemma}
\begin{proof}
In the following, all equalities involving conditional expectations are understood as a.s.\ equalities.

\textbf{``$\Longrightarrow$''}: First, note that by the law of total probability, $E[Y_{i2}(0)-Y_{i1}(0)|G_i=1,W_i]=E[Y_{i2}(0)-Y_{i1}(0)|G_i=0,W_i]$ implies $$E[Y_{i2}(0)-Y_{i1}(0)|G_i=1,W_i]=E[Y_{i2}(0)-Y_{i1}(0)|W_i].$$ The result follows from noting that $E[Y_{i2}(0)-Y_{i1}(0)|G_i=1,W_i]=\frac{E[G_i(Y_{i2}(0)-Y_{i1}(0))|W_i]}{P(G_i=1|W_i)}$ by definition.

\textbf{``$\Longleftarrow$''}: Since $P(G_i=1|W_i)\in(0,1)$, it follows that $E[Y_{i2}(0)-Y_{i1}(0)|G_i=1,W_i]=E[Y_{i2}(0)-Y_{i1}(0)|W_i]$. It then follows that
\begin{eqnarray*}&&E[Y_{i2}(0)-Y_{i1}(0)|G_i=1,W_i]P(G_i=1|W_i)+E[Y_{i2}(0)-Y_{i1}(0)|G_i=0,W_i]P(G_i=0|W_i)\\
&&=E[Y_{i2}(0)-Y_{i1}(0)|G_i=1,W_i].\end{eqnarray*}
The result follows from subtracting the first term on the left-hand side and dividing by $P(G_i=0|W_i)$.
\end{proof}
\begin{lemma}\label{lem:necessary}
For a scalar random variable $W_i$, let $\dot{W}_i=W_i-E[W_i]$. If $E[\dot{W}_i1\{\dot{W}_i\leq 0\}]=0$ or $E[\dot{W}_i1\{\dot{W}_i\geq 0\}]=0$, then $W_i=E[W_i]$ a.s.
\end{lemma}

\begin{proof} We prove the result for the case where $E[\dot{W}_i1\{\dot{W}_i\leq 0\}]=0$, since the proof for the other case follows by identical arguments.
First, note that by definition $E[\dot{W}_i]=0$, which is equivalent to \begin{align}E[\dot{W}_i^+]=E[\dot{W}_i^-],\label{eq:zeta_meanzero}\end{align}
where $\dot{W}_i^+=|\dot{W}_i|1\{\dot{W}_i>0\}$ and $\dot{W}_i^-=|\dot{W}_i|1\{\dot{W}_i< 0\}$.

Now suppose that $E[\dot{W}_i1\{\dot{W}_i\leq 0\}]=0$ holds, which is equivalent to 
\begin{align}E[\dot{W}_i^+1\{\dot{W}_i\leq 0\}]=E[\dot{W}_i^-1\{\dot{W}_i\leq 0\}],\end{align}
since, by definition, $\dot{W}_i=\dot{W}_i^+-\dot{W}_i^-$. Note that the left-hand side equals zero by the definition of $\dot{W}_i^+$. As a result, $E[\dot{W}_i^-1\{\dot{W}_i\leq 0\}]=E[\dot{W}_i^-]=0$. Since $\dot{W}_i^-\geq 0$, this implies that $P(\dot{W}_i^-=0)=1$. Now note that $P(\dot{W}_i^-=0)=P(|\dot{W}_i|1\{\dot{W}_i<0\}=0)=P(1\{\dot{W}_i<0\}=0)=1$, which implies $P(\dot{W}_i<0)=0$.

Since $E[\dot{W}_i]=0$, \eqref{eq:zeta_meanzero} further implies that $E[\dot{W}_i^-]=E[\dot{W}_i^+]=0$. Since $\dot{W}_i^+\geq 0$, it follows that $P(\dot{W}_i^+=0)=1$. Now note that $P(\dot{W}_i^+=0)=P(|\dot{W}_i|1\{\dot{W}_i>0\}=0)=P(1\{\dot{W}_i>0\}=0)=1$, which implies $P(\dot{W}_i>0)=0$.

Together, $P(\dot{W}_i<0)=0$ and $P(\dot{W}_i>0)=0$ imply that $P(\dot{W}_i=0)=1-(P(\dot{W}_i<0)+P(\dot{W}_i>0))=1$, which completes the proof.
\end{proof}

\subsection{Proof of Lemma \ref{lem:bias_decomposition}}\label{proof:lem:bias_decomposition}

We rewrite $\biasdidpost$ as follows,
\begin{eqnarray*}
    \biasdidpost&\equiv &E[Y_{i2}(0)-Y_{i1}(0)|G_i=1]-E[Y_{i2}(0)-Y_{i1}(0)|G_i=0]\\&=&\frac{E[G_i(Y_{i2}(0)-Y_{i1}(0))]}{P(G_i=1)}-\frac{E[(1-G_i)(Y_{i2}(0)-Y_{i1}(0)]}{P(G_i=0)}\\
    &=&\frac{(1-P(G_i=1))E[G_i(Y_{i2}(0)-Y_{i1}(0))]-P(G_i=1)E[Y_{i2}(0)-Y_{i1}(0)]}{P(G_i=1)P(G_i=0)}\\&&+\frac{P(G_i=1)E[G_i(Y_{i2}(0)-Y_{i1}(0))]}{P(G_i=1)P(G_i=0)}\\
    &=&\frac{E[G_i(\dot{Y}_{i2}(0)-\dot{Y}_{i1}(0))]}{P(G_i=1)P(G_i=0)}\\
    &=&\frac{E[(G_i-E[G_i|\oi])(\dot{Y}_{i2}(0)-\dot{Y}_{i1}(0))]}{P(G_i=1)P(G_i=0)}+\frac{E[E[G_i|\oi](\dot{Y}_{i2}(0)-\dot{Y}_{i1}(0))]}{P(G_i=1)P(G_i=0)}\\
&=&\frac{E[(G_i-E[G_i|\oi])(\dot{Y}_{i2}(0)-\dot{Y}_{i1}(0))]}{P(G_i=1)P(G_i=0)}+\frac{E[E[G_i|\oi]E[\dot{Y}_{i2}(0)-\dot{Y}_{i1}(0)|\oi]]}{P(G_i=1)P(G_i=0)}
\end{eqnarray*}
The first equality follows by definition of $\biasdidpost$. The second equality follows by the LIE. The penultimate equality follows from subtracting and adding $E[G_i|\oi](\dot{Y}_{i2}(0)-\dot{Y}_{i1}(0))$ inside the expectation in the numerator. The last equality follows from the LIE.
\qed

\subsection{Proof of Proposition \ref{prop:bias_characterization}}\label{proof:prop:bias_characterization}

We first simplify $\biasselpost$ in Lemma \ref{lem:bias_decomposition} with $\omega_i=(\ai,\eipre,\mu_i,\etaipre)$ as follows, 
\begin{eqnarray*}
    \biasselpost
    &=&\frac{E[(G_i-E[G_i|\ai,\eipre,\mu_i,\etaipre])(\dot{Y}_{i2}(0)-\dot{Y}_{i1}(0))]}{P(G_i=1)P(G_i=0)},\\
    &=&\frac{E[G_i(\dot{Y}_{i2}(0)-\dot{Y}_{i1}(0))]}{P(G_i=1)P(G_i=0)}\\
    &&-\frac{E[E[G_i|\ai,\eipre,\mu_i,\etaipre](E[\dot{Y}_{i2}(0)|\ai,\eipre,\mu_i,\etaipre]-\dot{Y}_{i1}(0))]}{P(G_i=1)P(G_i=0)}\\    
    &=&\frac{E[G_i(\dot{Y}_{i2}(0)-E[\dot{Y}_{i2}(0)|\ai,\eipre,\mu_i,\etaipre])]}{P(G_i=1)P(G_i=0)}\\
    &=&E[\zeta_{i2}|G_i=1]-E[\zeta_{i2}|G_i=0]
\end{eqnarray*}
The second and the third equality follow from the LIE. The last equality follows from Assumption \ref{ass:REL} and because $G_i\in \{0,1\}$.

Next, we simplify $\biasmtgpost$ in Lemma \ref{lem:bias_decomposition} with $\omega_i=(\ai,\eipre,\mu_i,\etaipre)$ as follows,
\begin{eqnarray*}
\biasmtgpost&=&\frac{E[E[G_i|\ai,\eipre,\mu_i,\etaipre](E[\dot{Y}_{i2}(0)|\ai,\eipre,\mu_i,\etaipre]-\dot{Y}_{i1}(0))]}{P(G_i=1)P(G_i=0)}\\&=&(\rho_2-1)\frac{E[G_i\dot{Y}_{i1}(0))]}{P(G_i=1)P(G_i=0)}\\
&=&(\rho_2-1)(E[Y_{i1}(0)|G_i=1]-E[Y_{i1}(0)|G_i=0].
\end{eqnarray*}
The second equality follows from Assumption \ref{ass:REL} and the LIE. The last equality follows because $G_i\in \{0,1\}$. \qed

\section{Proofs of results in the Appendix}

\subsection{Auxiliary lemmas}

\begin{lemma}[Equivalence with multiple periods]
\label{lem: PT-MP equivalent condition}
Suppose that Assumption \ref{ass:NA} holds and $P(G_i=g)\in(0,1)$ for $g\in \{2,\dots,T,\infty\}$. Then Assumption \ref{ass:PT-MP} is equivalent to $E[1\{G_i=g\}(\dot{Y}_{it}(\infty)-\dot{Y}_{i(t-1)}(\infty))]=0$ for $g\in \{2,\dots,T,\infty\}$ and $t\in \{2,\dots,T\}$.
\end{lemma}
\begin{proof}
Assumption \ref{ass:PT-MP} is equivalent to
$$
E[\dot{Y}_{it}(\infty)-\dot{Y}_{i(t-1)}(\infty)|G_i=g]=E[\dot{Y}_{it}(\infty)-\dot{Y}_{i(t-1)}(\infty)|G_i=\infty]\quad \text{for }(g,t)\in \{2,\dots,T\}^2, 
$$
which, since $E[\dot{Y}_{it}(\infty)]=0$, is also equivalent to 
\begin{equation}
E[\dot{Y}_{it}(\infty)-\dot{Y}_{i(t-1)}(\infty)|G_i=g]=0\quad \text{for }(g,t)\in \{2,\dots,T,\infty\}\times \{2,\dots,T\}. \label{eq: PT for all g}
\end{equation}
Thus, we need to show that \eqref{eq: PT for all g} is equivalent to $E[1\{G_i=g\}(\dot{Y}_{it}(\infty)-\dot{Y}_{i(t-1)}(\infty))]=0$ for $g\in \{2,\dots,T,\infty\}$ and $t\in \{2,\dots,T\}$. This follows because
$$
E[\dot{Y}_{it}(\infty)-\dot{Y}_{i(t-1)}(\infty)|G_i=g]=\frac{E[1\{G_i=g\}(\dot{Y}_{it}(\infty)-\dot{Y}_{i(t-1)}(\infty))]}{P(G_i=g)} \quad 
$$
for $(g,t)\in \{2,\dots,T,\infty\}\times \{2,\dots,T\}$, since $P(G_i=g)\in (0,1)$ for $g\in\{2,\dots,T,\infty\}$ by assumption. 
\end{proof}

\begin{lemma}\label{lem:necessary-X}
For a scalar random variable $W_i$, let $\ddot{W}_i=W_i-E[W_i|X_i]$. If $E[\ddot{W}_i1\{\ddot{W}_i\leq 0\}|X_i]=0$ or $E[\ddot{W}_i1\{\ddot{W}_i\geq 0\}|X_i]=0$ a.s., then $W_i=E[W_i|X_i]$ a.s.
\end{lemma}

\begin{proof} We prove the result for the case where $E[\ddot{W}_i1\{\ddot{W}_i\leq 0\}|X_i]=0$, since the proof for the other case follows by identical arguments. In the following, all statements involving conditional expectations hold a.s.

First, note that by definition $E[\ddot{W}_i|X_i]=0$, which is equivalent to \begin{align}E[\ddot{W}_i^+|X_i]=E[\ddot{W}_i^-|X_i],\label{eq:zeta_meanzero_X}\end{align}
where $\ddot{W}_i^+=|\ddot{W}_i|1\{\ddot{W}_i>0\}$ and $\ddot{W}_i^-=|\ddot{W}_i|1\{\ddot{W}_i< 0\}$.

Now suppose that $E[\ddot{W}_i1\{\ddot{W}_i\leq 0\}|X_i]=0$ holds, which is equivalent to 
\begin{align}E[\ddot{W}_i^+1\{\ddot{W}_i\leq 0\}|X_i]=E[\ddot{W}_i^-1\{\ddot{W}_i\leq 0\}|X_i],\end{align}
since, by definition, $\ddot{W}_i=\ddot{W}_i^+-\ddot{W}_i^-$. Note that the left-hand side equals zero by the definition of $\ddot{W}_i^+$. As a result, $E[\ddot{W}_i^-1\{\ddot{W}_i\leq 0\}|X_i]=E[\ddot{W}_i^-|X_i]=0$. Since $\ddot{W}_i^-\geq 0$, this implies that $P(\ddot{W}_i^-=0|X_i)=1$. Now note that $P(\ddot{W}_i^-=0|X_i)=P(|\ddot{W}_i|1\{\ddot{W}_i<0\}=0|X_i)=P(1\{\ddot{W}_i<0\}=0|X_i)=1$, which implies $P(\ddot{W}_i<0|X_i)=0$.

Since $E[\ddot{W}_i|X_i]=0$, \eqref{eq:zeta_meanzero_X} further implies that $E[\ddot{W}_i^-|X_i]=E[\ddot{W}_i^+|X_i]=0$. Since $\ddot{W}_i^+\geq 0$, it follows that $P(\ddot{W}_i^+=0|X_i)=1$. Now note that $P(\ddot{W}_i^+=0|X_i)=P(|\ddot{W}_i|1\{\ddot{W}_i>0\}=0|X_i)=P(1\{\ddot{W}_i>0\}=0|X_i)=1$, which implies $P(\ddot{W}_i>0|X_i)=0$.

Together, $P(\ddot{W}_i<0|X_i)=0$ and $P(\ddot{W}_i>0|X_i)=0$ imply that $P(\ddot{W}_i=0|X_i)=1-(P(\ddot{W}_i<0|X_i)+P(\ddot{W}_i>0|X_i))=1$, which further implies that $P(\ddot{W}_i=0)=1$ and thereby completes the proof.
\end{proof}

\begin{lemma}\label{lem:exchangeability}
Let $(\ai,\eione,\eitwo)$ denote a vector of random variables on $\mathcal{A}\times\mathcal{E}^2$. 
Suppose that $\eione,\eitwo|\ai\overset{d}{=}\eitwo,\eione|\ai$ holds. Then,  
\begin{enumerate}[(i)]\setlength\itemsep{0pt}
\item $F_{\eione|\ai}(e|a)=F_{\eitwo|\ai}(e|a)$ $a.e.$ $(a,e)\in\mathcal{A}\times\mathcal{E}$. \label{lem:exchangeability_i}
\item  Suppose further that the conditional density of $(\eione,\eitwo)|\ai$,  $f_{\eione,\eitwo|\ai}(e_1,e_2|a)$, exists and the conditional density of $\eit|\ai$, $f_{\eit|\ai}(e|a)$, is strictly positive  for $(e,a)\in\mathcal{E}\times \mathcal{A}$, $t=1,2$. Then, $F_{\eione|\eitwo,\ai}(e_1|e_2,a)=F_{\eitwo|\eione,\ai}(e_1|e_2,a)$ $a.e.$ $(a,e_1,e_2)\in\mathcal{A}\times\mathcal{E}^2$. \label{lem:exchangeability_ii}
\end{enumerate}
\end{lemma}
\begin{proof}\textbf{(i)} By the definition of the marginal distribution, the conditional exchangeability restriction implies (i) by the following $a.e.$

\vspace{-0.5cm}

{\small\begin{align}F_{\eione|\ai}(e_1|a)&=\lim_{e_2\rightarrow\infty}F_{\eione,\eitwo|\ai}(e_1,e_2|a)=\lim_{e_2\rightarrow\infty}F_{\eione,\eitwo|\ai}(e_2,e_1|a)=F_{\eitwo|\ai}(e_1|a).
\end{align}}

\noindent \textbf{(ii)} To simplify notation, we prove this result assuming $\mathcal{E}=\mathbb{R}^k$, where $k\equiv \dim(\eit)$. Let $\tilde{e}=(\tilde{e}^1,\dots\tilde{e}^k)$ be the constant of integration and $e_1=(e_1^1,\dots,e_1^k)$. By the definition of the conditional distribution and (i) of this lemma, the conditional exchangeability restriction implies (ii) by the following

\vspace{-0.5cm}

{\footnotesize
\begin{eqnarray*}F_{\eione|\eitwo,\ai}(e_1|e_2,a)&=&\int_{-\infty}^{e_1^1}\dots\int_{-\infty}^{e_1^k}f_{\eione|\eitwo,\ai}(\tilde{e}|e_2,a)d\tilde{e}^1\dots d\tilde{e}^k=\int_{-\infty}^{e_1^1}\dots\int_{-\infty}^{e_1^k}\frac{f_{\eione,\eitwo|\ai}(\tilde{e},e_2|a)}{f_{\eitwo|\ai}(e_2|a)}d\tilde{e}^1\dots d\tilde{e}^k\nonumber\\
&=&\int_{-\infty}^{e_1^1}\dots\int_{-\infty}^{e_1^k}\frac{f_{\eitwo,\eione|\ai}(\tilde{e},e_2|a)}{f_{\eione|\ai}(e_2|a)}d\tilde{e}^1\dots d\tilde{e}^k=F_{\eitwo|\eione,\ai}(e_1|e_2,a).\label{eq:exchangeability_ii} \end{eqnarray*}}
\end{proof}

\subsection{Proof of Theorem \ref{thm:main-MP}}\label{proof:thm:main-MP}

\textbf{``$\Longrightarrow$''}: We first consider the case where $P(E[\dot{Y}_{it}(\infty)-\dot{Y}_{i(t-1)}(\infty))|\omega_i]>0)<1$ for $t\in\{2,\dots,T\}$. Since Assumption \ref{ass:PT-MP} holds for all $g\in \mathcal{G}_{\omega}$, it holds for the following selection mechanism, where $\mathcal{G}_{S}=\{2,\dots,T\}$ denotes the set of switcher groups,
\begin{align*}\check{g}(\omega_i,\nu_i)=
\left\{\begin{array}{l}
g \text{ if }  1\{E[\dot{Y}_{ig}(\infty) - \dot{Y}_{i(g-1)}(\infty)|\omega_i]\leq 0\}1\{\nu_i\in B_g\}=1, g\in\mathcal{G}_S\\
\infty \text{ otherwise}.
\end{array}\right.
\end{align*}
Next, we show that $P(\check{g}(\omega_i,\nu_i)=g)\in(0,1)$ for $g\in\{2,\dots,T,\infty\}$, such that $G_i=\check{g}(\omega_i,\nu_i)$ is a nondegenerate selection mechanism.  For $g\in\mathcal{G}_S$,
\begin{eqnarray}P(\check{g}(\omega_i,\nu_i)=g)=P(E[\dot{Y}_{ig}(\infty)-\dot{Y}_{i(g-1)}(\infty)|\omega_i]\leq 0)P(\nu_i\in B_g)\in(0,1),\label{eq:G_S_probability}\end{eqnarray}
where the first equality follows from the independence between $\omega_i$ and $\nu_i$ implied by Assumption \ref{ass:SEL-MP}. Since $P(E[\dot{Y}_{ig}(\infty)-\dot{Y}_{i(g-1)}(\infty)|\omega_i]\leq 0)>0$ and $P(\nu_i\in B_g)\in(0,1)$ for $\mathcal{G}_S$, it follows that $P(\check{g}(\omega_i,\nu_i)=g)\in(0,1)$ for $g\in\mathcal{G}_S$. This further implies that $\sum_{g\in\mathcal{G}_S}P(\check{g}(\omega_i,\nu_i)=g)>0$ and therefore $P(\check{g}(\omega_i,\nu_i)=\infty)=1-\sum_{g\in\mathcal{G}_S}P(\check{g}(\omega_i,\nu_i)=g)<1$. It remains to show that $P(\check{g}(\omega_i,\nu_i)=\infty)>0$.
\begin{eqnarray*}
P(\check{g}(\omega_i,\nu_i)=\infty)&=&1-\sum_{g\in\mathcal{G}_S}P(\check{g}(\omega_i,\nu_i)=g)\geq 1-\sum_{g\in\mathcal{G}_S}P(\nu_i\in B_g)=P(\nu_i\in B_\infty)>0,
\end{eqnarray*}
where the weak inequality follows from $P(\check{g}(\omega_i,\nu_i)=g)\leq P(\nu_i\in B_g)$ in \eqref{eq:G_S_probability}. The last equality follows from Assumption \ref{ass:SEL-MP}, in particular that $\{B_g\}_{g=2,\dots,T,\infty}$ is a non-overlapping partition of the support of $\nu_i$ and $P(\nu_i\in B_g)\in(0,1)$ for $g\in\{2,\dots,T,\infty\}$. It follows that $P(\check{g}(\omega_i,\nu_i)=g)\in(0,1)$ for $g\in\{2,\dots,T,\infty\}$.

Now we can invoke Lemma \ref{lem: PT-MP equivalent condition} to show the implication of Assumption \ref{ass:PT-MP} with $G_i=\check{g}(\omega_i,\nu_i)$; specifically  for any $g\in \mathcal{G}_S$
\begin{eqnarray*} 
0&=&E[1\{\check{g}(\omega_i,\nu_i)=g\}(\dot{Y}_{ig}(\infty) -\dot{Y}_{i(g-1)}(\infty)) ]\\&=&
E[ 1\{E[\dot{Y}_{ig}(\infty)- \dot{Y}_{i(g-1)}(\infty)|\omega_i]\leq 0\}1\{\nu_i\in B_g\}(\dot{Y}_{ig}(\infty) - \dot{Y}_{i(g-1)}(\infty))]\\
&=&P(\nu_i\in B_g)E[ 1\{E[\dot{Y}_{ig}(\infty)- \dot{Y}_{i(g-1)}(\infty)|\omega_i]\leq 0\}(\dot{Y}_{ig}(\infty) - \dot{Y}_{i(g-1)}(\infty))]\\
&=&P(\nu_i\in B_g)E[ 1\{E[\dot{Y}_{ig}(\infty)- \dot{Y}_{i(g-1)}(\infty)|\omega_i]\leq 0\}E[\dot{Y}_{ig}(\infty) - \dot{Y}_{i(g-1)}(\infty)|\omega_i]]. \end{eqnarray*}
The third equality follows by the independence condition in Assumption \ref{ass:SEL-MP} and the definition of $\omega_i$. The last equality follows by the LIE.  Since $P(\nu_i\in B_g)>0$ by Assumption \ref{ass:SEL-MP}, it follows that $E[ 1\{E[\dot{Y}_{ig}(\infty)- \dot{Y}_{i(g-1)}(\infty)|\omega_i]\leq 0\}E[\dot{Y}_{ig}(\infty) - \dot{Y}_{i(g-1)}(\infty)|\omega_i]]=0$ for $g\in\mathcal{G}_S$. By Lemma \ref{lem:necessary} with $W_i=E[\dot{Y}_{ig}(\infty) - \dot{Y}_{i(g-1)}(\infty)|\omega_i]$ for each $g\in\mathcal{G}_S$, it follows that $E[\dot{Y}_{ig}(\infty)-\dot{Y}_{i(g-1)}(\infty)|\omega_i]=0$ a.s.\ for each $g\in\mathcal{G}_S=\{2,\dots,T\}$, which implies the result. 

The proof for the case where $P(E[\dot{Y}_{it}(\infty)-\dot{Y}_{i(t-1)}(\infty)|\omega_i]<0)<1$ for $t\in\{2,\dots,T\}$ follows symmetrically using the selection mechanism,
\begin{align*}\check{g}(\omega_i,\nu_i)=\left\{\begin{array}{l}
g \text{ if }  1\{E[\dot{Y}_{ig}(\infty) - \dot{Y}_{i(g-1)}(\infty)|\omega_i]\ge 0\}1\{\nu_i\in B_g\}=1, g\in\mathcal{G}_S\\
\infty \text{ otherwise}.
\end{array}\right.
\end{align*}

The proof for the case where $P(E[\dot{Y}_{it}(\infty)-\dot{Y}_{i(t-1)}(\infty)|\omega_i]>0)<1$ for $t\in\mathcal{G}_{1}\subset \mathcal{G}_S$ and $P(E[\dot{Y}_{is}(\infty)-\dot{Y}_{i(s-1)}(\infty)|\omega_i]<0)<1$ for $s\in\mathcal{G}_2=\mathcal{G}_1^c\cap \mathcal{G}_S$ follows from using the following selection mechanism
\begin{align*}\check{g}(\omega_i,\nu_i)=
\left\{\begin{array}{l}
g \text{ if }  1\{E[\dot{Y}_{ig}(\infty) - \dot{Y}_{i(g-1)}(\infty)|\omega_i]\le 0\}1\{\nu_i\in B_g\}=1, g\in\mathcal{G}_1,\\
g \text{ if }  1\{E[\dot{Y}_{ig}(\infty) - \dot{Y}_{i(g-1)}(\infty)|\omega_i]\ge 0\}1\{\nu_i\in B_g\}=1, g\in\mathcal{G}_2,\\
\infty \text{ otherwise}.
\end{array}\right.
\end{align*}

\noindent\textbf{``$\Longleftarrow$''}: This direction is immediate by the LIE.
\subsection{Proof of Theorem \ref{thm:main-X}}\label{proof:thm:main-X}

\textbf{``$\Longrightarrow$''}: We prove the result for the case where $P(E[\ddot{Y}_{i2}(0)-\ddot{Y}_{i1}(0)|\omega_i,X_i]>0)<1$. The proof for the case where $P(E[\ddot{Y}_{i2}(0)-\ddot{Y}_{i1}(0)|\omega_i,X_i]<0)<1$ follows from the same arguments. 

Under Assumption \ref{ass:SEL-X} and because $P(E[\ddot{Y}_{i2}(0)-\ddot{Y}_{i1}(0)|\omega_i,X_i]>0)<1$, the selection mechanism 
\begin{equation}
G_i=1\{\nu_i>c\}1\{E[\ddot{Y}_{i2}(0)-\ddot{Y}_{i1}(0)|\oi,X_i]\le 0\}\label{eq:least_favorable_mechanism}
\end{equation}
is nondegenerate, that is,
$$
P(1\{\nu_i>c\}1\{E[\ddot{Y}_{i2}(0)-\ddot{Y}_{i1}(0)|\oi,X_i]\le 0\}=1|X_i)\in (0,1).
$$

If Assumption \ref{ass:PT-X} holds for all non-degenerate selection mechanisms $g\in \mathcal{G}_{\omega}^x$, then it holds for the mechanism \eqref{eq:least_favorable_mechanism}. By Lemma \ref{lem: PT equivalent condition}, Assumption \ref{ass:PT-X} holding for the mechanism in \eqref{eq:least_favorable_mechanism} is equivalent to 
\begin{equation*}
E[1\{\nu_i>c\}1\{E[\ddot{Y}_{i2}(0)-\ddot{Y}_{i1}(0)|\oi,X_i]\le 0\}(\ddot{Y}_{i2}(0)-\ddot{Y}_{i1}(0))|X_i]=0,
\end{equation*}
which, by Assumption \ref{ass:SEL-X}, is equivalent to
\begin{equation*}
E[1\{E[\ddot{Y}_{i2}(0)-\ddot{Y}_{i1}(0)|\oi,X_i]\le 0\}(\ddot{Y}_{i2}(0)-\ddot{Y}_{i1}(0))|X_i]=0.
\end{equation*}

By the law of iterated expectations (LIE), this is further equivalent to 
\begin{eqnarray*}
E[1\{E[\ddot{Y}_{i2}(0)-\ddot{Y}_{i1}(0)|\oi,X_i]\le 0 \}E[\ddot{Y}_{i2}(0)-\ddot{Y}_{i1}(0)|\oi,X_i]|X_i]=0
\end{eqnarray*}
Since $E[E[\ddot{Y}_{i2}(0)-\ddot{Y}_{i1}(0)|\oi,X_i]|X_i]=0$ by construction, the result follows by Lemma \ref{lem:necessary-X}.

\smallskip

\noindent \textbf{``$\Longleftarrow$''}: By the LIE,
\begin{eqnarray*}
E[\ddot{Y}_{i2}(0)-\ddot{Y}_{i1}(0)|G_i,X_i]&=&E[E[\ddot{Y}_{i2}(0)-\ddot{Y}_{i1}(0)|\oi,\nu_i,X_i]|G_i,X_i]\\
&=&E[E[\ddot{Y}_{i2}(0)-\ddot{Y}_{i1}(0)|\oi,X_i]|G_i,X_i]=0,
\end{eqnarray*}
where the penultimate equality follows from Assumption \ref{ass:SEL-X}. 

\subsection{Proof of Proposition \ref{prop:sufficient-NSP}}\label{proof:prop:sufficient-NSP}
In this proof, all equalities involving random variables are understood to hold a.s. 

\smallskip

First, by Lemma \ref{lem: PT equivalent condition}, Assumption \ref{ass:PT-NSP} under Assumption \ref{ass:NSP-X} holds if and only if
\begin{align}
&E[G_i(Y_{i2}(0)-Y_{i1}(0))|X_i^\lambda,X_{i1}^\mu=X_{i2}^\mu]\nonumber \\
=&E[G_i|X_i^\lambda,X_{i1}^\mu=X_{i2}^\mu]E[Y_{i2}(0)-Y_{i1}(0)|X_i^\lambda,X_{i1}^\mu=X_{i2}^\mu].\label{eq:CPT_NSP}
\end{align}
Next, we state some preliminary observations and then proceed to show each statement separately. 

Note that, by the LIE, Assumption \ref{ass:SEL-CI} and the definition of $\bar{g}(\cdot)$, the LHS of \eqref{eq:CPT_NSP} equals the following,
\begin{align}&E[G_i(Y_{i2}(0)-Y_{i1}(0))|X_i^\lambda,X_{i1}^\mu=X_{i2}^\mu]\nonumber\\
=&E[E[G_i|X_{i}^\mu,X_{i}^\lambda,\ai^\mu,\ai^\lambda,  \eione^\mu,\eitwo^\mu,\eione^\lambda,\eitwo^\lambda](Y_{i2}(0)-Y_{i1}(0))|X_i^\lambda,X_{i1}^\mu=X_{i2}^\mu]\nonumber\\
=&E[\bar{g}(X_{i1}^\mu,X_{i2}^\mu,X_{i1}^\lambda, X_{i2}^\lambda,\ai^\mu,\eione^\mu,\eitwo^\mu)(Y_{i2}(0)-Y_{i1}(0))|X_i^\lambda,X_{i1}^\mu=X_{i2}^\mu].
\end{align}
Similarly, by the LIE, the RHS of \eqref{eq:CPT_NSP} equals the following,

\vspace{-0.5cm}
{\small
\begin{align}&E[G_i|X_i^{\lambda},X_{i1}^\mu=X_{i2}^\mu]E[Y_{i2}(0)-Y_{i1}(0)|X_i^\lambda,X_{i1}^\mu=X_{i2}^\mu]\nonumber\\=&E[\bar{g}(X_{i1}^\mu,X_{i2}^\mu,X_{i1}^\lambda, X_{i2}^\lambda,\ai^\mu,\eione^\mu,\eitwo^\mu)|X_i^{\lambda},X_{i1}^\mu=X_{i2}^\mu]E[Y_{i2}(0)-Y_{i1}(0)|X_i^{\lambda},X_{i1}^\mu=X_{i2}^\mu]\end{align}}

\noindent As a result, in the following, to show that Assumptions \ref{ass:SC1-NSP}, \ref{ass:SC2-NSP}, and \ref{ass:SC3-NSP} are sufficient for Assumption \ref{ass:PT-NSP}, it suffices to show that each assumption implies the following equality,
\begin{align}&E[\bar{g}(X_{i1}^\mu,X_{i2}^\mu,X_{i1}^\lambda, X_{i2}^\lambda,\ai^\mu,\eione^\mu,\eitwo^\mu)(Y_{i2}(0)-Y_{i1}(0))|X_i^\lambda,X_{i1}^\mu=X_{i2}^\mu]\nonumber\\
=&E[\bar{g}(X_{i1}^\mu,X_{i2}^\mu,X_{i1}^\lambda, X_{i2}^\lambda,\ai^\mu,\eione^\mu,\eitwo^\mu)|X_i^{\lambda},X_{i1}^\mu=X_{i2}^\mu]E[Y_{i2}(0)-Y_{i1}(0)|X_i^{\lambda},X_{i1}^\mu=X_{i2}^\mu]\nonumber
\end{align}
\noindent \textbf{(i)} By Assumption \ref{ass:NSP-X}, it follows that

\vspace{-0.5cm}
{\footnotesize
\begin{align}&E[\bar{g}(X_{i1}^\mu,X_{i2}^\mu,X_{i1}^\lambda,X_{i2}^\lambda,\ai^\mu,\eione^\mu,\eitwo^\mu)(Y_{i2}(0)-Y_{i1}(0))|X_i^\lambda,X_{i1}^\mu=X_{i2}^\mu]\nonumber\\
=&E[\bar{g}(X_{i1}^\mu,X_{i2}^\mu,X_{i1}^\lambda,X_{i2}^\lambda,\ai^\mu,\eione^\mu,\eitwo^\mu)(\mu(X_{i2}^\mu,\ai^\mu,\eitwo^\mu)-\mu(X_{i1}^\mu,\ai^\mu,\eione^\mu))|X_i^\lambda,X_{i1}^\mu=X_{i2}^\mu]\nonumber\\
&+E[\bar{g}(X_{i1}^\mu,X_{i2}^\mu,X_{i1}^\lambda,X_{i2}^\lambda,\ai^\mu,\eione^\mu,\eitwo^\mu)(\lambda_2(X_{i2}^\lambda,\ai^\lambda,\eitwo^\lambda)-\lambda_1(X_{i1}^\lambda,\ai^\lambda,\eione^\lambda))|X_i^\lambda,X_{i1}^\mu=X_{i2}^\mu],\label{eq:CPT-NSP_nonsep2}
\end{align}}

\noindent We first examine the first term on the RHS of the above equality.  Note that by the symmetry restrictions in Assumptions \ref{ass:SC1-NSP}.\ref{ass:SC1-NSP_selection_symmetry} and \ref{ass:SC1-NSP}.\ref{ass:SC1-NSP_exchangeability}, it follows that a.e. $(a,x^\mu,x_1^\lambda,x_2^\lambda)\in\mathcal{A}\times\mathcal{X}_\mu\times\mathcal{X}_\lambda^2$

\vspace{-0.5cm}
{\footnotesize
\begin{align}&E[\bar{g}(X_{i1}^\mu,X_{i2}^\mu,X_{i1}^\lambda,X_{i2}^\lambda,\ai^\mu,\eione^\mu,\eitwo^\mu)\mu(X_{i1}^\mu,\ai^\mu,\eione^\mu)|X_i^\lambda=(x_1^\lambda,x_2^\lambda),X_{i1}^\mu=X_{i2}^\mu=x^\mu,\ai^\mu=a]\nonumber\\
=&\int{\bar{g}(x^\mu,x^\mu,x_1^\lambda,x_2^\lambda,a,e_1,e_2)\mu(x^\mu,a,e_1)}dF_{\eione^\mu,\eitwo^\mu|X_i^\lambda,X_{i1}^\mu=X_{i2}^\mu,\ai^\mu}(e_1,e_2|(x_1^\lambda,x_2^\lambda),x^\mu,a)\nonumber\\
=&\int{\bar{g}(x^\mu,x^\mu,x_1^\lambda,x_2^\lambda,a,e_2,e_1)\mu(x^\mu,a,e_1)}dF_{\eione^\mu,\eitwo^\mu|X_i^\lambda,X_{i1}^\mu=X_{i2}^\mu,\ai^\mu}(e_2,e_1|(x_1^\lambda,x_2^\lambda),x^\mu,a)\nonumber\\
=&E[\bar{g}(X_{i1}^\mu,X_{i2}^\mu,X_{i1}^\lambda,X_{i2}^\lambda,\ai^\mu,\eione^\mu,\eitwo^\mu)\mu(X_{i2}^\mu,\ai^\mu,\eitwo^\mu)|X_i^\lambda=(x_1^\lambda,x_2^\lambda),X_{i1}^\mu=X_{i2}^\mu=x^\mu,\ai^\mu=a].\label{eq:gmu_equality}\end{align}}
As a result, the first summand in \eqref{eq:CPT-NSP_nonsep2} equals zero by \eqref{eq:gmu_equality} and the LIE.

Next, we consider the second summand in \eqref{eq:CPT-NSP_nonsep2},

\vspace{-0.5cm}
{\footnotesize
\begin{align}&E[\bar{g}(X_{i1}^\mu,X_{i2}^\mu,X_{i1}^\lambda,X_{i2}^\lambda,\ai^\mu,\eione^\mu,\eitwo^\mu)(\lambda_2(X_{i2}^\lambda,\ai^\lambda,\eitwo^\lambda)-\lambda_1(X_{i1}^\lambda,\ai^\lambda,\eione^\lambda))|X_i^\lambda,X_{i1}^\mu=X_{i2}^\mu]\nonumber\\
=&E[\bar{g}(X_{i1}^\mu,X_{i2}^\mu,X_{i1}^\lambda,X_{i2}^\lambda,\ai^\mu,\eione^\mu,\eitwo^\mu)|X_i^\lambda,X_{i1}^\mu=X_{i2}^\mu]E[\lambda_2(X_{i2}^\lambda,\ai^\lambda,\eitwo^\lambda)-\lambda_1(X_{i1}^\lambda,\ai^\lambda,\eione^\lambda)|X_i^\lambda,X_{i1}^\mu=X_{i2}^\mu]\nonumber\\
=&E[\bar{g}(X_{i1}^\mu,X_{i2}^\mu,X_{i1}^\lambda,X_{i2}^\lambda,\ai^\mu,\eione^\mu,\eitwo^\mu)|X_i^\lambda,X_{i1}^\mu=X_{i2}^\mu]E[Y_{i2}(0)-Y_{i1}(0)|X_i^\lambda,X_{i1}^\mu=X_{i2}^\mu].
\end{align}}

\noindent The first equality follows from the conditional independence assumption in Assumption \ref{ass:SC1-NSP}.\ref{ass:SC1-NSP_independence}.  The last equality follows from the time homogeneity of $F_{\eit^\mu|X_i^\mu,X_{i}^\lambda,\ai^\mu}$, which follows from the exchangeability restriction in Assumption \ref{ass:SC1-NSP}.\ref{ass:SC1-NSP_exchangeability} by Lemma \ref{lem:exchangeability}.\ref{lem:exchangeability_i}, and implies that $E[\mu(X_{i2}^\mu,\ai^\mu,\eitwo^\mu)-\mu(X_{i1}^\mu,\ai^\mu,\eione^\mu)|X_i^\lambda,X_{i1}^\mu=X_{i2}^\mu,\ai^\mu]=0$ and 
$$
E[Y_{i2}(0)-Y_{i1}(0)|X_i^\lambda,X_{i1}^\mu=X_{i2}^\mu]=E[\lambda_2(X_{i2}^\lambda,\ai^\lambda,\eitwo^\lambda)-\lambda_1(X_{i1}^\lambda,\ai^\lambda,\eione^\lambda))|X_i^\lambda,X_{i1}^\mu=X_{i2}^\mu]
$$
by the LIE. As a result, the above implies that Assumption \ref{ass:PT-NSP} holds.

\bigskip

\noindent \textbf{(ii)} By Assumption  \ref{ass:SC2-NSP}.\ref{ass:SC2-NSP_selection}, we can define $\check{\bar{g}}(x_1^\mu,x_2^\mu,x_1^\lambda,x_2^\lambda,a^\mu,e_1^\mu)=\bar{g}(x_1^\mu,x_2^\mu,x_1^\lambda,x_2^\lambda,a^\mu,e_1^\mu,e_2^\mu).$ By Assumption \ref{ass:NSP-X}, it follows that

\vspace{-0.5cm}
{\footnotesize
\begin{align}&E[\check{\bar{g}}(X_{i1}^\mu,X_{i2}^\mu,X_{i1}^\lambda,X_{i2}^\lambda,\ai^\mu,\eione^\mu)(Y_{i2}(0)-Y_{i1}(0))|X_i^\lambda,X_{i1}^\mu=X_{i2}^\mu]\nonumber\\
=&E[\check{\bar{g}}(X_{i1}^\mu,X_{i2}^\mu,X_{i1}^\lambda,X_{i2}^\lambda,\ai^\mu,\eione^\mu)\Delta_{\mu,i}|X_i^\lambda,X_{i1}^\mu=X_{i2}^\mu]\nonumber\\
&+E[\check{\bar{g}}(X_{i1}^\mu,X_{i2}^\mu,X_{i1}^\lambda,X_{i2}^\lambda,\ai^\mu,\eione^\mu)(\lambda_2(X_{i2}^\lambda,\ai^\lambda,\eitwo^\lambda)-\lambda_1(X_{i1}^\lambda,\ai^\lambda,\eione^\lambda))|X_i^\lambda,X_{i1}^\mu=X_{i2}^\mu]\nonumber\\
=&E[\check{\bar{g}}(X_{i1}^\mu,X_{i2}^\mu,X_{i1}^\lambda,X_{i2}^\lambda,\ai^\mu,\eione^\mu)|X_i^\lambda,X_{i1}^\mu=X_{i2}^\mu]E[\Delta_{\mu,i}|X_i^\lambda,X_{i1}^\mu=X_{i2}^\mu]\nonumber\\
&+E[\check{\bar{g}}(X_{i1}^\mu,X_{i2}^\mu,X_{i1}^\lambda,X_{i2}^\lambda,\ai^\mu,\eione^\mu)|X_i^\lambda,X_{i1}^\mu=X_{i2}^\mu]E[\lambda_2(X_{i2}^\lambda,\ai^\lambda,\eitwo^\lambda)-\lambda_1(X_{i1}^\lambda,\ai^\lambda,\eione^\lambda)|X_i^\lambda,X_{i1}^\mu=X_{i2}^\mu]\nonumber\\
=&E[\check{\bar{g}}(X_{i1}^\mu,X_{i2}^\mu,X_{i1}^\lambda,X_{i2}^\lambda,\ai^\mu,\eione^\mu)|X_i^\lambda,X_{i1}^\mu=X_{i2}^\mu]E[Y_{i2}(0)-Y_{i1}(0)|X_i^\lambda,X_{i1}^\mu=X_{i2}^\mu]\label{eq:PT-NSPii}
\end{align}}

\noindent The second equality follows from the conditional independence conditions in Assumptions \ref{ass:SC2-NSP}.\ref{ass:SC2-NSP_separability} and \ref{ass:SC2-NSP}.\ref{ass:SC2-NSP_independence}. The last equality follows from Assumption \ref{ass:NSP-X}. Equation \eqref{eq:PT-NSPii} then implies Assumption \ref{ass:PT-NSP}.

\bigskip

\noindent \textbf{(iii)} By Assumption  \ref{ass:SC3-NSP}.\ref{ass:SC3-NSP_selection_ai}, we can define $\check{\bar{g}}(x_1^\mu,x_2^\mu,x_1^\lambda,x_2^\lambda,a^\mu)=\bar{g}(x_1^\mu,x_2^\mu,x_1^\lambda,x_2^\lambda,a^\mu,e_1^\lambda,e_2^\lambda).$ Now by Assumption \ref{ass:NSP-X} and \ref{ass:SC3-NSP}.\ref{ass:SC3-NSP_selection_ai}, it follows that

\vspace{-0.5cm}
{\footnotesize
\begin{align}&E[\check{\bar{g}}(X_{i1}^\mu,X_{i2}^\mu,X_{i1}^\lambda,X_{i2}^\lambda,\ai^\mu)(Y_{i2}(0)-Y_{i1}(0))|X_i^\lambda,X_{i1}^\mu=X_{i2}^\mu]\nonumber\\
=&E[\check{\bar{g}}(X_{i1}^\mu,X_{i2}^\mu,X_{i1}^\lambda,X_{i2}^\lambda,\ai^\mu)(\mu(X_{i2}^\mu,\ai^\mu,\eitwo^\mu)-\mu(X_{i1}^\mu,\ai^\mu,\eione^\mu))|X_i^\lambda,X_{i1}^\mu=X_{i2}^\mu]\nonumber\\
&+E[\check{\bar{g}}(X_{i1}^\mu,X_{i2}^\mu,X_{i1}^\lambda,X_{i2}^\lambda,\ai^\mu)(\lambda_2(X_{i2}^\lambda,\ai^\lambda,\eitwo^\lambda)-\lambda_1(X_{i1}^\lambda,\ai^\lambda,\eione^\lambda))|X_i^\lambda,X_{i1}^\mu=X_{i2}^\mu]\nonumber\\
=&E[\check{\bar{g}}(X_{i1}^\mu,X_{i2}^\mu,X_{i1}^\lambda,X_{i2}^\lambda,\ai^\mu)E[\mu(X_{i2}^\mu,\ai^\mu,\eitwo^\mu)-\mu(X_{i1}^\mu,\ai^\mu,\eione^\mu)|X_i^\lambda,X_{i1}^\mu=X_{i2}^\mu,\ai^\mu]|X_i^\lambda,X_{i1}^\mu=X_{i2}^\mu]\nonumber\\
&+E[\check{\bar{g}}(X_{i1}^\mu,X_{i2}^\mu,X_{i1}^\lambda,X_{i2}^\lambda,\ai^\mu)|X_i^\lambda,X_{i1}^\mu=X_{i2}^\mu]E[\lambda_2(X_{i2}^\lambda,\ai^\lambda,\eitwo^\lambda)-\lambda_1(X_{i1}^\lambda,\ai^\lambda,\eione^\lambda)|X_i^\lambda,X_{i1}^\mu=X_{i2}^\mu]\nonumber\\
=&E[\check{\bar{g}}(X_{i1}^\mu,X_{i2}^\mu,X_{i1}^\lambda,X_{i2}^\lambda,\ai^\mu)|X_i^\lambda,X_{i1}^\mu=X_{i2}^\mu]E[\lambda_2(X_{i2}^\lambda,\ai^\lambda,\eitwo^\lambda)-\lambda_1(X_{i1}^\lambda,\ai^\lambda,\eione^\lambda)|X_i^\lambda,X_{i1}^\mu=X_{i2}^\mu]\nonumber\\
=&E[\check{\bar{g}}(X_{i1}^\mu,X_{i2}^\mu,X_{i1}^\lambda,X_{i2}^\lambda,\ai^\mu)|X_i^\lambda,X_{i1}^\mu=X_{i2}^\mu]E[Y_{i2}(0)-Y_{i1}(0)|X_i^\lambda,X_{i1}^\mu=X_{i2}^\mu],\nonumber
\end{align}}

\noindent where the first equality follows from Assumption \ref{ass:NSP-X}. The second equality follows by applying the LIE to the first term and the conditional independence imposed in Assumption \ref{ass:SC3-NSP}.\ref{ass:SC3-NSP_independence} to the second term.  The first term on the RHS of the second equality equals zero by the conditioning on $X_{i1}^\mu=X_{i2}^\mu$ and the time homogeneity condition in Assumption \ref{ass:SC3-NSP}.\ref{ass:SC3-NSP_time_homogeneity}. The last equality follows from noting, similar to the proof of (i), that since $E[\mu(X_{i2}^\mu,\ai^\mu,\eitwo^\mu)-\mu(X_{i1}^\mu,\ai^\mu,\eione^\mu)|X_i^\lambda,X_{i1}^\mu=X_{i2}^\mu,\ai^\mu]=0$, $$E[Y_{i2}(0)-Y_{i1}(0)|X_i^\lambda,X_{i1}^\mu=X_{i2}^\mu]=E[\lambda_2(X_{i2}^\lambda,\ai^\lambda,\eitwo^\lambda)-\lambda_1(X_{i1}^\lambda,\ai^\lambda,\eione^\lambda)|X_i^\lambda,X_{i1}^\mu=X_{i2}^\mu]$$
by the LIE. This completes the proof.
\qed

\subsection{Proof of Proposition \ref{prop:TH+CRE imply parallel trends}}\label{proof:prop:TH+CRE imply parallel trends}
Under Assumption \ref{ass:NSP-X},
\begin{align}&E[Y_{i2}(0)-Y_{i1}(0)|G_i,X_i^\lambda,X_{i1}^\mu=X_{i2}^\mu]\nonumber\\
=&E[\mu(X_{i2}^\mu,\ai^\mu,\eitwo^\mu)-\mu(X_{i1}^\mu,\ai^\mu,\eione^\mu)|G_i,X_i^\lambda,X_{i1}^\mu=X_{i2}^\mu]\label{eq:EDiff1}\\
&+E[\lambda_2(X_{i2}^\lambda,\ai^\lambda,\eitwo^\lambda)-\lambda_1(X_{i1}^\lambda,\ai^\lambda,\eione^\lambda)|G_i,X_i^\lambda,X_{i1}^\mu=X_{i2}^\mu].\label{eq:EDiff2}\end{align}
The remainder of the proof follows in two steps. First, we show that the term in \eqref{eq:EDiff1} equals zero under our assumptions. Second, we show that the second term is conditionally mean independent of $G_i$, which implies Assumption \ref{ass:PT-NSP}.

We proceed to show that under Assumption \ref{ass:TH-X} the term in \eqref{eq:EDiff1} equals zero by the following,
\begin{align}&E[\mu(X_{i1}^\mu,\ai^\mu,\eione^\mu)|G_i=g,X_i^\lambda=(x_1^\lambda,x_2^\lambda),X_{i1}^\mu=X_{i2}^\mu=x^\mu]\nonumber\\
=&\int{\mu(x^\mu,a^\mu,e^\mu)}dF_{\ai^\mu,\eione^\mu|G_i,X_{i}^\mu,X_i^\lambda}(a^\mu,e^\mu|g,(x^\mu,x^\mu),(x_1^\lambda,x_2^\lambda))\nonumber\\
=&\int{\mu(x^\mu,a^\mu,e^\mu)}dF_{\ai^\mu,\eitwo^\mu|G_i,X_i^\mu,X_{i}^\lambda}(a^\mu,e^\mu|g,(x^\mu,x^\mu),(x_1^\lambda,x_2^\lambda))\nonumber\\
=&E[\mu(X_{i2}^\mu,\ai^\mu,\eitwo^\mu)|G_i=g,X_i^\lambda=(x_1^\lambda,x_2^\lambda),X_{i1}^\mu=X_{i2}^\mu=x^\mu],\end{align}
where the first and last equalities follow by definition, whereas the penultimate equality follows from Assumption \ref{ass:TH-X} noting that it implies $\ai^\mu,\eione^\mu|G_i,X_i^\mu,X_i^\lambda\overset{d}{=}\ai^\mu,\eitwo^\mu|G_i,X_i^\mu,X_i^\lambda$. 

Finally, we show that Assumption \ref{ass:RE-X} implies the following for \eqref{eq:EDiff2}
\begin{align}
&E[\lambda_2(X_{i2}^\lambda,\ai^\lambda,\eitwo^\lambda)-\lambda_1(X_{i1}^\lambda,\ai^\lambda,\eione^\lambda)|G_i=g,X_i^\lambda=(x_1^\lambda,x_2^\lambda),X_{i1}^\mu=X_{i2}^\mu=x^\mu]\nonumber\\
=&\int{(\lambda_2(x_2^\lambda,a^\lambda,e_2^\lambda)-\lambda_1(x_1^\lambda,a^\lambda,e_1^\lambda))}dF_{\ai^\lambda,\eione^\lambda,\eitwo^\lambda|G_i,X_{i}^\mu,X_i^\lambda}(a^\lambda,e_1^\lambda,e_2^\lambda|g,(x^\mu,x^\mu),(x_1^\lambda,x_2^\lambda))\nonumber\\
=&\int{(\lambda_2(x_2^\lambda,a^\lambda,e_2^\lambda)-\lambda_1(x_1^\lambda,a^\lambda,e_1^\lambda))}dF_{\ai^\lambda,\eione^\lambda,\eitwo^\lambda|X_i^\mu,X_i^\lambda}(a^\lambda,e_1^\lambda,e_2^\lambda|(x^\mu,x^\mu),(x_1^\lambda,x_2^\lambda))\nonumber\\
=&E[\lambda_2(X_{i2}^\lambda,\ai^\lambda,\eitwo^\lambda)-\lambda_1(X_{i1}^\lambda,\ai^\lambda,\eione^\lambda)|X_i^\lambda=(x_1^\lambda,x_2^\lambda),X_{i1}^\mu=X_{i2}^\mu=x^\mu],
\end{align}
where the penultimate equality follows by Assumption \ref{ass:RE-X}.
This completes the proof. \qed

\subsection{Proof of Proposition \ref{prop:connection-NSP-X}}\label{proof:prop:connection-NSP-X}
Throughout this proof, equalities involving conditioning statements are understood to hold $a.e.$ We proceed to show each result separately.  

\bigskip

\noindent \textbf{(i)} It suffices to show (i.a) Assumptions \ref{ass:SC1-NSP}.\ref{ass:SC1-NSP_selection_symmetry} and \ref{ass:SC1-NSP}.\ref{ass:SC1-NSP_exchangeability} imply Assumption \ref{ass:TH-X} and (i.b) Assumptions \ref{ass:SC1-NSP}.\ref{ass:SC1-NSP_selection_symmetry} and \ref{ass:SC1-NSP}.\ref{ass:SC1-NSP_independence} imply Assumption \ref{ass:RE-X}.

\textbf{(i.a)}
First, we define the conditional joint (mixed) density function of $\eione^\mu$ and $G_i$
 \begin{eqnarray*}&&f_{\eione^\mu,G_i|X_i^\mu,X_i^\lambda,\ai^\mu}(e_1,g|x^\mu,x^\lambda,a)\\&=&P(G_i=g|\eione^\mu=e_1,X_i^\mu=x^\mu,X_i^\lambda=x^\lambda,\ai^\mu=a)f_{\eione^\mu|X_i^\mu,X_i^\lambda,\ai^\mu}(e_1|x^\mu,x^\lambda,a)
 \end{eqnarray*}

Assumption \ref{ass:SC1-NSP}.\ref{ass:SC1-NSP_exchangeability} implies 
$f_{\eione^\mu|X_i^\mu,X_i^\lambda,\ai^\mu}(e|x^\mu,x^\lambda,a)=f_{\eitwo^\mu|X_i^\mu,X_i^\lambda,\ai^\mu}(e|x^\mu,x^\lambda,a)$ as well as $f_{\eione^\mu|\eitwo^\mu,X_i^\mu,X_i^\lambda,\ai^\mu}(e_1|e_2,x^\mu,x^\lambda,a)=f_{\eitwo^\mu|\eione^\mu,X_i^\mu,X_i^\lambda,\ai^\mu}(e_1|e_2,x^\mu,x^\lambda,a)$ by Lemma \ref{lem:exchangeability}. As a result, the second term on the RHS of the last equality equals $f_{\eitwo^\mu|X_i^\mu,X_i^\lambda,\ai^\mu}(e_1|x^\mu,x^\lambda,a)$. To complete the proof that Assumption \ref{ass:TH-X} holds, it remains to show the following

\begin{eqnarray*}
&&P(G_i=g|\eione^\mu=e_1,X_i^\mu=x^\mu,X_i^\lambda=x^\lambda,\ai^\mu=a)\\&=&\int{1\{g(x_1^\mu,x_2^\mu,x_1^\lambda,x_2^\lambda,a,e_1,e_2)= g\}f_{\eitwo^\mu|\eione^\mu,X_i^\mu,X_i^\lambda,\ai^\mu}(e_2|e_1,x^\mu,x^\lambda,a)de_2}\\
&=&\int{1\{g(x_1^\mu,x_2^\mu,x_1^\lambda,x_2^\lambda,a,e_2,e_1)= g\}f_{\eione^\mu|\eitwo^\mu,X_i^\mu,X_i^\lambda,\ai^\mu}(e_2|e_1,x^\mu,x^\lambda,a)de_2}\\
&=&P(G_i=g|\eitwo^\mu=e_1,X_i^\mu=x^\mu,X_i^\lambda=x^\lambda,\ai^\mu=a)
\end{eqnarray*}
where the penultimate equality follows from Assumptions \ref{ass:SC1-NSP}.\ref{ass:SC1-NSP_selection_symmetry} and the implication of \ref{ass:SC1-NSP}.\ref{ass:SC1-NSP_exchangeability} that follows from Lemma \ref{lem:exchangeability}.

As a result,
\begin{eqnarray*}f_{\eione^\mu,G_i|X_i^\mu,X_i^\lambda,\ai^\mu}(e_1,g|x^\mu,x^\lambda,a)
&=&f_{\eitwo^\mu,G_i|X_i^\mu,X_i^\lambda,\ai^\mu}(e_1,g|x^\mu,x^\lambda,a).
 \end{eqnarray*}
Dividing the last equality by $P(G_i=g|x^\mu,x^\lambda,a)\in(0,1)$ for $g\in\{0,1\}$, it follows that $f_{\eione^\mu|G_i,X_i^\mu,X_i^\lambda,\ai^\mu}=f_{\eitwo^\mu|G_i,X_i^\mu,X_i^\lambda,\ai^\mu}$, and therefore Assumption \ref{ass:TH-X} holds.

\textbf{(i.b)} This statement follows in a straightforward manner from the definition of $G_i$ in Assumption \ref{ass:SC1-NSP}.\ref{ass:SC1-NSP_selection_symmetry} and the conditional independence condition in Assumption \ref{ass:SC1-NSP}.\ref{ass:SC1-NSP_independence} which together imply Assumption \ref{ass:RE-X}.  This completes the proof of (i).

\bigskip

\noindent \textbf{(ii)} To show the result, it suffices to show that (ii.a) Assumptions \ref{ass:SC3-NSP}.\ref{ass:SC3-NSP_selection_ai} and \ref{ass:SC3-NSP}.\ref{ass:SC3-NSP_time_homogeneity} imply Assumption \ref{ass:TH-X} and (ii.b) Assumptions \ref{ass:SC3-NSP}.\ref{ass:SC3-NSP_selection_ai} and \ref{ass:SC3-NSP}.\ref{ass:SC3-NSP_independence} imply Assumption \ref{ass:RE-X}.

\textbf{(ii.a)} Under Assumptions  \ref{ass:SC3-NSP}.\ref{ass:SC3-NSP_selection_ai} and \ref{ass:SC3-NSP}.\ref{ass:SC3-NSP_time_homogeneity},  conditional on $X_i^\mu$, $X_i^\lambda$ and $\ai^\mu$, $G_i=g(X_{i1}^\mu,X_{i2}^\mu,X_{i1}^\lambda,X_{i2}^\lambda,\ai^\mu)$ is a degenerate random variable equaling either zero or one with probability one. As a result, 

\vspace{-0.5cm}
{\footnotesize{\begin{align}&F_{\eit^\mu|G_i,X_i^\mu,X_i^\lambda,\ai^\mu}(e|g,x^\mu,x^\lambda,a)\nonumber\\
=&\sum_{h=0,1}P(\eit^\mu\leq e|G_i=g(x_1^\mu,x_2^\mu,x_1^\lambda,x_2^\lambda,a),X_i^\mu=x^\mu,X_i^\lambda=x^\lambda,\ai^\mu=a)1\{g(x_1^\mu,x_2^\mu,x_1^\lambda,x_2^\lambda,a)=h\}\nonumber\\
=&\sum_{h=0,1}P(\eit^\mu\leq e|X_i^\mu=x^\mu,X_i^\lambda=x^\lambda,\ai^\mu=a)1\{g(x_1^\mu,x_2^\mu,x_1^\lambda,x_2^\lambda,a)=h\}\nonumber\\
=&\sum_{h=0,1}F_{\eit^\mu|X_i^\mu,X_i^\lambda,\ai^\mu}(e|x^\mu,x^\lambda,a)1\{g(x_1^\mu,x_2^\mu,x_1^\lambda,x_2^\lambda,a)=h\}.
\end{align}}}

\noindent As a result, Assumption \ref{ass:SC3-NSP}.\ref{ass:SC3-NSP_selection_ai} together with the time homogeneity of $F_{\eit^\mu|X_i^\mu,X_i^\lambda,\ai^\mu}$ in Assumption \ref{ass:SC3-NSP}.\ref{ass:SC3-NSP_time_homogeneity} is sufficient for the time homogeneity of $F_{\eit^\mu|G_i,X_i^\mu,X_i^\lambda,\ai^\mu}$, which yields Assumption \ref{ass:TH-X}.

\textbf{(ii.b)} The statement (ii.b) is immediate from noting that Assumption \ref{ass:SC3-NSP}.\ref{ass:SC3-NSP_independence} together with $G_i=g(X_{i1}^\mu,X_{i2}^\mu,X_{i1}^\lambda,X_{i2}^\lambda,\ai^\mu)$ imply that $g(X_{i1}^\mu,X_{i2}^\mu,X_{i1}^\lambda,X_{i2}^\lambda,\ai^\mu)\indep$ $ (\ai^\lambda,\eione^\lambda,\eitwo^\lambda)|X_i^\mu,X_i^\lambda$, which is equivalent to Assumption \ref{ass:RE-X}.  This completes the proof of (ii). \qed

\end{document}